\documentclass[11pt]{article}
\usepackage{epsf}
\usepackage{amsmath}
\usepackage{epsfig}
\usepackage{times}
\usepackage{amssymb}
\usepackage{amsthm}
\usepackage{setspace}
\usepackage{cite}

\usepackage{algorithmic}  
\usepackage{algorithm}

\usepackage{shadow}
\usepackage{fancybox}
\usepackage{fancyhdr}

\def\w{{\bf w}}

\def\y{{\bf y}}
\def\v{{\bf v}}
\def\x{{\bf x}}

\def\x{{\mathbf x}}

\def\w{{\bf w}}

\def\v{{\bf v}}
\def\x{{\bf x}}
\def\y{{\bf y}}
\def\z{{\bf z}}
\def\q{{\bf q}}
\def\a{{\bf a}}
\def\b{{\bf b}}

\def\h{{\bf h}}

\def\be{\begin{equation}}
\def\ee{\end{equation}}
\def\ba{\left[\begin{array}}
\def\ea{\end{array}\right]}

\def\t{{\bf t}}

\def\w{{\bf w}}

\def\v{{\bf v}}
\def\x{{\bf x}}
\def\y{{\bf y}}
\def\z{{\bf z}}
\def\q{{\bf q}}
\def\a{{\bf a}}
\def\b{{\bf b}}

\def\xtilde{\tilde{\x}}
\def\xhat{\hat{\x}}

\def\1{{\bf 1}}

\def\g{{\bf g}}
\def\0{{\bf 0}}

\def\erfinv{\mbox{erfinv}}

\newtheorem{theorem}{Theorem}
\newtheorem{corollary}{Corollary}

\newtheorem{lemma}{Lemma}

\begin{document}

\begin{singlespace}

\title {A framework to characterize performance of LASSO algorithms
}
\author{
\textsc{Mihailo Stojnic}
\\
\\
{School of Industrial Engineering}\\
{Purdue University, West Lafayette, IN 47907} \\
{e-mail: {\tt mstojnic@purdue.edu}} }
\date{}
\maketitle

\centerline{{\bf Abstract}} \vspace*{0.1in}

In this paper we consider solving \emph{noisy} under-determined systems of linear equations with sparse solutions. A noiseless equivalent attracted enormous attention in recent years, above all, due to work of \cite{CRT,CanRomTao06,DonohoPol} where it was shown in a statistical and large dimensional context that a sparse unknown vector (of sparsity proportional to the length of the vector) can be recovered from an under-determined system via a simple polynomial $\ell_1$-optimization algorithm. \cite{CanRomTao06} further established that even when the equations are \emph{noisy}, one can, through an SOCP noisy equivalent of $\ell_1$, obtain an approximate solution that is (in an $\ell_2$-norm sense) no further than a constant times the noise from the sparse unknown vector. In our recent works \cite{StojnicCSetam09,StojnicUpper10}, we created a powerful mechanism that helped us characterize exactly the performance of $\ell_1$ optimization in the noiseless case (as shown in \cite{StojnicEquiv10} and as it must be if the axioms of mathematics are well set, the results of \cite{StojnicCSetam09,StojnicUpper10} are in an absolute agreement with the corresponding exact ones from \cite{DonohoPol}). In this paper we design a mechanism, as powerful as those from \cite{StojnicCSetam09,StojnicUpper10}, that can handle the analysis of a LASSO type of algorithm (and many others) that can be (or typically are) used for ``solving" noisy under-determined systems. Using the mechanism we then, in a statistical context, compute the exact worst-case $\ell_2$ norm distance between the unknown sparse vector and the approximate one obtained through such a LASSO. The obtained results match the corresponding exact ones obtained in \cite{BayMon10,DonMalMon10}. Moreover, as a by-product of our analysis framework we recognize existence of an SOCP type of algorithm that achieves the same performance.

\vspace*{0.25in} \noindent {\bf Index Terms: Noisy linear systems of equations; LASSO; SOCP;
$\ell_1$-optimization; compressed sensing} .

\end{singlespace}

\section{Introduction}
\label{sec:back}

In recent years the problem of finding sparse solutions of under-determined systems of linear equations attracted enormous attention. Applications seem vast and as if they are growing almost on a daily basis (see, e.g. \cite{ECicm,DDTLSKB,CT,JRimaging,BCDH08,CRchannel,VPH,PVMHjournal,WM08,Olgica,RFPrank,MBPSZ08,RS08} and references therein). Given a substantial interest in the problem (and especially that it is coming from a variety of different fields), one may assume that designing efficient algorithms that would solve it could be of far-reaching importance. To that end, we believe that a precise mathematical understanding of the phenomena that make certain algorithms work well would help solidify belief in their success in current and future applications. Moreover, it is possible that down the road it can also help expand further the range of their applications.

Moving long the same lines, we in this paper focus on studying mathematical properties of under-determined systems of linear equations and certain algorithms used to solve them. We start the story by introducing an idealized version of the problem that we plan to study. In its simplest form it amounts to finding a $k$-sparse $\x$ such
that
\begin{equation}
A\x=\y \label{eq:system}
\end{equation}
where $A$ is an $m\times n$ ($m<n$) matrix and $\y$ is
an $m\times 1$ vector (see Figure
\ref{fig:model}; here and in the rest of the paper, under $k$-sparse vector we assume a vector that has at most $k$ nonzero
components). Of course, the assumption will be that such an $\x$ exists (clearly, the case of real interest is $k<m$). To make writing in the rest of the paper easier, we will assume the
so-called \emph{linear} regime, i.e. we will assume that $k=\beta n$
and that the number of equations is $m=\alpha n$ where
$\alpha$ and $\beta$ are constants independent of $n$ (more
on the non-linear regime, i.e. on the regime when $m$ is larger than
linearly proportional to $k$ can be found in e.g.
\cite{CoMu05,GiStTrVe06,GiStTrVe07}).
\begin{figure}[htb]
\centering
\centerline{\epsfig{figure=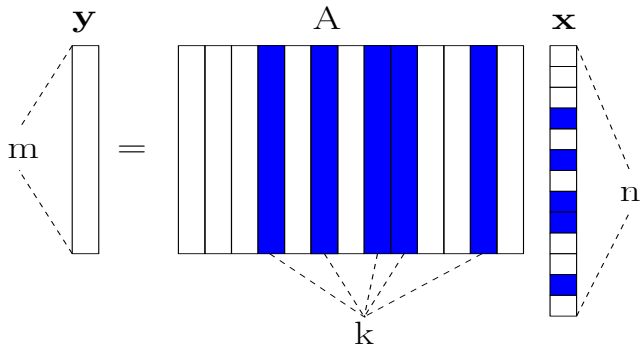,width=9cm,height=4.5cm}}
\caption{Model of a linear system; vector $\x$ is $k$-sparse}
\label{fig:model}
\end{figure}

If one has the freedom to design matrix $A$ then the results from \cite{FHicassp,Tarokh,MaVe05} demonstrated that the techniques from
coding theory (based on coding/decoding of Reed-Solomon codes)
can be employed to determine \emph{any} $k$-sparse $\x$ in
(\ref{eq:system}) for any $0<\alpha\leq 1$ and any
$\beta\leq\frac{\alpha}{2}$ in polynomial time. It is relatively easy to show that under the unique recoverability assumption
$\beta$ can not be greater than $\frac{\alpha}{2}$. Therefore, as long as one is concerned with the unique recovery of
$k$-sparse $\x$ in (\ref{eq:system}) in polynomial time the results from \cite{FHicassp,Tarokh,MaVe05} are
optimal. The complexity of algorithms from
\cite{FHicassp,Tarokh,MaVe05} is roughly $O(n^3)$. In a similar fashion one can, instead of using coding/decoding techniques associated with Reed/Solomon codes,
design the matrix and the corresponding recovery algorithm based on the techniques related to coding/decoding of
Expander codes (see e.g.
\cite{XHexpander,JXHC08,InRu08} and references therein). In that case recovering $\x$ in
(\ref{eq:system}) is significantly faster for large dimensions $n$. Namely, the complexity of the techniques from e.g. \cite{XHexpander,JXHC08,InRu08}
(or their slight modifications) is usually
$O(n)$ which is clearly for large $n$ significantly smaller than $O(n^3)$. However,
the techniques based on coding/decoding of Expander codes usually do not allow for $\beta$ to be as large as
$\frac{\alpha}{2}$.

On the other hand, if one has no freedom in choice of $A$ designing the algorithms to find $k$-sparse $\x$ in (\ref{eq:system}) is substantially harder. In fact, when there is no choice in $A$ the recovery
problem (\ref{eq:system}) becomes NP-hard. Two algorithms 1) \emph{Orthogonal matching pursuit - OMP} and 2) \emph{Basis pursuit -
$\ell_1$-optimization} (and their different
variations) have been often viewed historically as solid heuristics for solving (\ref{eq:system}) (in recent years belief propagation type of algorithms are emerging as strong alternatives as well). Roughly speaking, OMP algorithms are faster but can recover smaller sparsity whereas the BP ones are slower but recover higher sparsity. In a more precise way, under certain probabilistic assumptions on the elements of $A$ it can be shown (see e.g. \cite{JATGomp,JAT,NeVe07})
that if $m=O(k\log(n))$
OMP (or slightly modified OMP) can recover $\x$ in (\ref{eq:system})
with complexity of recovery $O(n^2)$. On the other hand a stage-wise
OMP from \cite{DTDSomp} recovers $\x$ in (\ref{eq:system}) with
complexity of recovery $O(n \log n)$. Somewhere in between OMP and BP are recent improvements CoSAMP (see e.g. \cite{NT08}) and Subspace pursuit (see e.g. \cite{DaiMil08}), which guarantee (assuming the linear regime) that the $k$-sparse $\x$ in (\ref{eq:system}) can be recovered in polynomial time with $m=O(k)$ equations which is the same performance guarantee established in \cite{CanRomTao06,DonohoPol} for the BP.

We now introduce the BP concept (or, as we will refer to it, the $\ell_1$-optimization concept; a slight modification/adaptation of it will actually be the main topic of this paper). Variations of the standard $\ell_1$-optimization from e.g.
\cite{CWBreweighted,SChretien08,SaZh08} as well as those from \cite{SCY08,FL08,GN03,GN04,GN07,DG08} related to $\ell_q$-optimization, $0<q<1$
are possible as well; moreover they can all be incorporated in what we will present below. The $\ell_1$-optimization concept suggests that one can maybe find the $k$-sparse $\x$ in
(\ref{eq:system}) by solving the following $\ell_1$-norm minimization problem
\begin{eqnarray}
\mbox{min} & & \|\x\|_{1}\nonumber \\
\mbox{subject to} & & A\x=\y. \label{eq:l1}
\end{eqnarray}
As is then shown in \cite{CanRomTao06} if
$\alpha$ and $n$ are given, $A$ is given and satisfies the restricted isometry property (RIP) (more on this property the interested reader can find in e.g. \cite{Crip,CRT,CanRomTao06,Bar,Ver,ALPTJ09}), then
any unknown vector $\x$ with no more than $k=\beta n$ (where $\beta$
is a constant dependent on $\alpha$ and explicitly
calculated in \cite{CanRomTao06}) non-zero elements can indeed be recovered by
solving (\ref{eq:l1}). In a statistical and large dimensional context in \cite{DonohoPol} and later in \cite{StojnicCSetam09} for any given value of $\beta$ the exact value of the maximum possible $\alpha$ was determined.

As we mentioned earlier the above scenario is in a sense idealistic. Namely, it assumes that $\y$ in (\ref{eq:l1}) was obtained through (\ref{eq:system}). On the other hand in many applications only a \emph{noisy} version of $A\x$ may be available for $\y$ (this is especially so in measuring type of applications) see, e.g. \cite{CRT,CanRomTao06,HN,W}. When that happens one has the following equivalent to (\ref{eq:system}) (see, Figure \ref{fig:modelnoise})
\begin{equation}
\y=A\x+\v, \label{eq:systemnoise}
\end{equation}
where $\v$ is an $m\times 1$ vector (often dubbed as the noise vector; the so-called ideal case presented above is of course a special case of the noisy case).
\begin{figure}[htb]
\centering
\centerline{\epsfig{figure=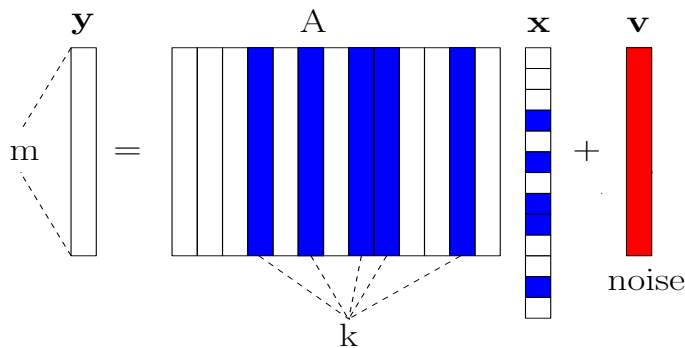,width=9cm,height=4.5cm}}
\caption{Model of a linear system; vector $\x$ is $k$-sparse}
\label{fig:modelnoise}
\end{figure}
Finding the $k$-sparse $\x$ in (\ref{eq:systemnoise}) is now incredibly hard. Basically, one is looking for a $k$-sparse $\x$ such that (\ref{eq:systemnoise}) holds and on top of that $\v$ is unknown. Although the problem is hard there are various heuristics throughout the literature that one can use to solve it approximately. Below we restrict our attention to two groups of algorithms that we believe are the most relevant to the results that we will present.

To introduce a bit or tractability in finding the $k$-sparse $\x$ in (\ref{eq:systemnoise}) one usually assumes certain amount of knowledge about either $\x$ or $\v$. As far as tractability assumptions on $\v$ are concerned one typically (and possibly fairly reasonably in applications of interest) assumes that $\|\v\|_2$ is bounded (or highly likely to be bounded) from above by a certain known quantity. The following second-order cone programming (SOCP) analogue to (\ref{eq:l1}) is one of the approaches that utilizes such an assumption (see, e.g. \cite{CanRomTao06})
\begin{eqnarray}
\min_{\x} & & \|\x\|_1\nonumber \\
\mbox{subject to} & & \|\y-A\x\|_2\leq r \label{eq:socp}
\end{eqnarray}
where, $r$ is a quantity such that $\|\v\|_2\leq r$ (or $r$ is a quantity such that $\|\v\|_2\leq r$ is say highly likely). For example, in \cite{CanRomTao06} a statistical context is assumed and based on the statistics of $\v$, $r$ was chosen such that $\|\v\|_2\leq r$ happens with overwhelming probability (as usual, under overwhelming probability we in this paper assume
a probability that is no more than a number exponentially decaying in $n$ away from $1$). Given that (\ref{eq:socp}) is now among few almost standard choices when it comes to finding the $\x$-sparse in (\ref{eq:systemnoise}), the literature on its properties is vast (see, e.g. \cite{CanRomTao06,DonElaTem06,Tropp06} and references therein). Also, given that this SOCP will not be the main topic of this paper we below briefly mention only what we consider to be the most influential work on this topic in recent years. Namely, in \cite{CanRomTao06} the authors analyzed performance of (\ref{eq:socp}) and showed a result similar in flavor to the one that holds in the ideal - noiseless - case. In a nutshell the following was shown in \cite{CanRomTao06}: let $\x$ be a $\beta n$-sparse vector such that (\ref{eq:systemnoise}) holds and let $\x_{socp}$ be the solution of (\ref{eq:socp}). Then $\|\x_{socp}-\x\|_2\leq C r$ where $\beta$ is a constant independent of $n$ and $C$ is a constant independent of $n$ and of course dependent on $\alpha$ and $\beta$. This result in a sense establishes a noisy equivalent to the fact that a linear sparsity can be recovered from an under-determined system of linear equations. In an informal language, it states that a linear sparsity can be \emph{approximately} recovered in polynomial time from a noisy under-determined system with the norm of the recovery error guaranteed to be within a constant multiple of the noise norm. Establishing such a result is, of course, a feat in its own class, not only because of its technical contribution but even more so because of the amount of interest that it generated in the field.

In this paper we will also consider an approximate recovery of the $k$-sparse $\x$ in (\ref{eq:systemnoise}). However, instead of the above mentioned SOCP we will focus on a group of highly successful algorithms called LASSO (the LASSO algorithms, as well as the SOCP ones, are of course well known in the statistics community and there is again a vast literature that covers their performance (see, e.g. \cite{CheDon95,CheDonSau98,Tibsh96,DonMalMon10,BayMon10,BunTsyWeg07,vandeGeer08,MeinYu09} and references therein). There are many variants of LASSO but the following one is probably the most well known
\begin{equation}
\min_{\x} \|\y-A\x\|_2^2+\lambda_{lasso}\|\x\|_1.\label{eq:biglasso}
\end{equation}
$\lambda_{lasso}$ in (\ref{eq:biglasso}) is a parameter to be chosen based on the amount of pre-knowledge one may have about $A$, $\v$, and/or $\x$. The results that relate to the characterization of the approximation error of (\ref{eq:biglasso}) that are similar to the SOCP ones mentioned above can be established (see, e.g. \cite{BayMon10lasso}). Of course, characterizing the performance of the recovery algorithm through the norm-2 of the error vector is only one possible way among many (more on other measures of performance can be found in e.g. \cite{W,BunTsyWeg07}). In this paper we will develop a novel framework for performance characterization of the LASSO algorithms. Among other things, in a statistical context, the framework will enable us to provide a precise characterization of the norm-2 of the approximation error of the LASSO algorithms.

While our main focus in this paper are algorithms from the LASSO group we mention that besides the SOCP and LASSO algorithms there are of course various other algorithms/heuristics that have been suggested as possible alternatives throughout the literature in recent years. Such an alternative that gained certain amount of popularity is for example the so-called Dantzig selector introduced in \cite{CanTao07}. The Dantzig selector amounts to solving the following optimization problem
\begin{eqnarray*}
\min & & \|\x\|_1 \nonumber \\
\mbox{subject to} & & \|A^T(A\x-\y)\|_{\infty}\leq C_{Dan},\label{eq:dantzsel}
\end{eqnarray*}
where $C_{Dan}$ is a carefully chosen parameter that of course should depend on $A,\v$, and/or $\x$. As a linear program the Danzig selector promises to be faster than SOCP or LASSO which are both quadratic programs. On the other hand recent improvements in numerical implementations of LASSO's and their solid approximate recovery abilities make them quite competitive as well (more on a thorough discussion/comparison, advantages/disdvatnages of the Dantzig selector and the LASSO algorithms can be found in e.g. \cite{MeiRocYu07,BicRitTsy09,FriSau07,EfrHatTib07,AsiRom10,JamRadLv09,Koltch09}).

To facilitate the exposition and the easiness of following we will present our framework on a version of the LASSO from (\ref{eq:biglasso}). Namely, we will consider,
\begin{eqnarray}
\min_{\x} & & \|\y-A\x\|_2 \nonumber \\
\mbox{subject to} & & \|\x\|_1\leq \|\xtilde\|_1\label{eq:lassol1}
\end{eqnarray}
where $\xtilde$ is the original $k$-sparse $\x$ that satisfies (\ref{eq:systemnoise}) (we just briefly mention that in a context that will be considered in this paper it is not that difficult to transform the LASSO from (\ref{eq:lassol1}) to one that is structurally equivalent to (\ref{eq:biglasso}); however, we stop short of exploring this connection further before presenting our main results and only mention that a section towards the end of the paper will explore it in more detail.). We do however mention right here that in order to run (\ref{eq:lassol1}) one does require the knowledge of $\|\xtilde\|_1$. In a sense this requirement is an equivalent to setting $r$ and $\lambda_{lasso}$ in (\ref{eq:socp}) and (\ref{eq:biglasso}), respectively. In order to be maximally effective both $r$ and $\lambda_{lasso}$ do require some amount of pre-knowledge about $A$, $\v$, and/or $\x$.


Before we proceed further we briefly summarize the organization of the rest of the paper. In Section
\ref{sec:unsigned}, we present a statistical framework for the performance analysis of the LASSO algorithms. To demonstrate its power we towards the end of Section \ref{sec:unsigned}, for any given $\alpha$  and $\beta$, compute the worst case norm-2 of the error that (\ref{eq:lassol1}) makes when used for approximate recovery of general sparse signals $\x$ from (\ref{eq:systemnoise}). In Section \ref{sec:signed} we then specialize results from Section \ref{sec:unsigned} to the so-called signed vectors $\x$. In Section \ref{sec:connectlasso} we discuss how the LASSO from (\ref{eq:lassol1}) can be connected to the LASSO from (\ref{eq:biglasso}). In Section \ref{sec:socplasso} we demonstrate that there is an SOCP algorithm (similar to the one given in (\ref{eq:socp})) that achieves the same performance as do (\ref{eq:lassol1}) and a corresponding (\ref{eq:biglasso}). In Section \ref{sec:numres} we present results that we obtained through numerical experiments. Finally, in Section \ref{sec:discuss} we discuss obtained results.

\section{LASSO's performance analysis framework -- general $\x$} \label{sec:unsigned}

In this section we create a statistical LASSO's performance analysis framework. Before proceeding further we will now explicitly state the major assumptions that we will make (the remaining ones, will be made appropriately throughout the analysis). Namely, in the rest of the paper we will assume that the elements of $A$ are i.i.d. standard normal random variables. We will also assume that the elements of $\v$ are i.i.d. Gaussian random variables with zero mean and variance $\sigma$. As stated earlier, we will assume that $\xtilde$ is the original $\x$ in (\ref{eq:systemnoise}) that we are trying to recover and that it is \emph{any} $k$-sparse vector with a given fixed location of its nonzero elements and a given fixed combination of their signs. Since the analysis (and the performance of (\ref{eq:lassol1})) will clearly be irrelevant with respect to what particular location and what particular combination of signs of nonzero elements are chosen, we can for the simplicity of the exposition and without loss of generality assume that the components $\x_{1},\x_{2},\dots,\x_{n-k}$ of $\x$ are equal to zero and the components $\x_{n-k+1},\x_{n-k+2},\dots,\x_n$ of $\x$ are greater than or equal to zero. Moreover, throughout the paper we will call such an $\x$ $k$-sparse and positive. In a more formal way we will set
\begin{eqnarray}
& & \xtilde_1=\xtilde_2 =  \dots=\xtilde_{n-k}=0\nonumber \\
& & \xtilde_{n-k+1}\geq 0,  \xtilde_{n-k+1}\geq 0, \dots, \xtilde_{n}\geq 0.\label{eq:xtildedef}
\end{eqnarray}
We also now take the opportunity to point out a rather obvious detail. Namely, the fact that $\xtilde$ is positive is assumed for the purpose of the analysis. However, this fact is not known \emph{a priori} and is not available to the solving algorithm (this will of course change in Section \ref{sec:signed}).

Once we establish the framework it will be clear that it can be used to characterize many of the LASSO features. We will defer these details to a collection of forthcoming papers. In this paper we will present only a small application that relates to a classical question of quantifying the approximation error that (\ref{eq:lassol1}) makes when used to recover \emph{any} $k$-sparse $\x$ that satisfies (\ref{eq:systemnoise}) and is from a set of $\x$'s with a given fixed location of nonzero elements and a given fixed combination of their signs.

Before proceeding further we will introduce a few definitions that will be useful in formalizing this application as well as in conducting the entire analysis.
As it is natural we start with the solution of (\ref{eq:lassol1}). Let $\xhat$ be the solution of (\ref{eq:lassol1}) and let $\w_{lasso}\in R^n$  be such that
\begin{equation}
\xhat=\xtilde+\w_{lasso}.\label{eq:xhatdef}
\end{equation}
As an application of our framework we will  compute the largest possible value of $\|\xhat-\xtilde\|_2=\|\w_{lasso}\|_2$ for any combination $(\alpha,\beta)$. Or more rigorously, for any combination $(\alpha,\beta)$, we will find a $d_{lasso}$ such that
\begin{equation}
\lim_{n\rightarrow\infty}P(d_{lasso}-\epsilon\leq \max_{\xtilde}\|\w_{lasso}\|_2\leq d_{lasso}+\epsilon)=1\label{eq:goallasso}
\end{equation}
for an arbitrarily small constant $\epsilon$. However, before doing so we will first present the general framework. The framework that we will present will center around finding the optimal value of the objective function in (\ref{eq:lassol1}) (of course in a probabilistic context). In the first of the following two subsections we will create a lower bound on this optimal value. We will then afterwards in the second of the subsections create an upper bound on this optimal value. Naturally in the third subsection we will show that the two bounds actually match. To make further writing easier and clearer we set already here
\begin{eqnarray}
\zeta_{obj}=\min_{\x} & & \|\y-A\x\|_2 \nonumber \\
\mbox{subject to} & & \|\x\|_1\leq \|\xtilde\|_1.\label{eq:objlassol1}
\end{eqnarray}

\subsection{Lower-bounding $\zeta_{obj}$} \label{sec:unsignedlbzetaobj}

In this section we present the part of the framework that relates to finding a ``high-probability" lower bound on $\zeta_{obj}$.
To make arguments that will follow less tedious we will make an assumption that is significantly weaker than what we will eventually prove. Namely, we will assume that there is a (if necessary, arbitrarily large) constant $C_\w$ such that
\begin{equation}
P(\|\w_{lasso}\|_2\leq C_\w)\geq 1-e^{-\epsilon_{C_{\w}}n}.\label{eq:assumplasso}
\end{equation}
To make our arguments flow more naturally, one should probably provide a direct proof of this statement right here. However, given the difficulty of the task ahead we refrain from that and assume that the statement is correct. Roughly speaking, what we assume is that $\|\w_{lasso}\|_2$ is bounded by an arbitrarily large constant (of course we hope to create a machinery that can prove much more than (\ref{eq:assumplasso})).

We start by noting that if one knows that $\y=A\xtilde+\v$ holds then (\ref{eq:objlassol1}) can be rewritten as
\begin{eqnarray}
\min_{\x} & & \|\v+A\xtilde-A\x\|_2 \nonumber \\
\mbox{subject to} & & \|\x\|_1\leq \|\xtilde\|_1.\label{eq:objlassol11}
\end{eqnarray}
After a small change of variables, $\x=\xtilde+\w$, (\ref{eq:objlassol11}) becomes
\begin{eqnarray}
\min_{\w} & & \|\v-A\w\|_2 \nonumber \\
\mbox{subject to} & & \|\xtilde+\w\|_1\leq \|\xtilde\|_1,\label{eq:objlassol12}
\end{eqnarray}
or in a more compact form
\begin{eqnarray}
\min_{\w} & & \|A_{\v}\begin{bmatrix} \w\\\sigma\end{bmatrix}\|_2 \nonumber \\
\mbox{subject to} & & \|\xtilde+\w\|_1\leq \|\xtilde\|_1,\label{eq:objlassol13}
\end{eqnarray}
where $A_{\v}=\begin{bmatrix} -A & \v \end{bmatrix}$ is now an $m\times (n+1)$ random matrix with i.i.d. standard normal components. Let
\begin{equation}
S_{\w}(\sigma,\xtilde,C_\w)=\{\begin{bmatrix}\w\\\sigma\end{bmatrix} \in R^{n+1}| \quad \|\w\|_2\leq C_\w \quad \mbox{and}\quad \|\xtilde+\w\|_1\leq \|\xtilde\|_1\}.\label{eq:defS}
\end{equation}
Further, let
\begin{equation}
f_{obj}(\sigma,\w)=\|A_{\v}\begin{bmatrix} \w\\\sigma\end{bmatrix}\|_2 \label{eq:deffobj}
\end{equation}
and set,
\begin{equation}
\zeta_{obj}^{(help)}=\min_{[\w^T \sigma]^T\in S_{\w}(\sigma,\xtilde,C_\w)} f_{obj}(\sigma,\w)= \min_{[\w^T \sigma]^T\in S_{\w}(\sigma,\xtilde,C_\w)}  \|A_{\v}\begin{bmatrix} \w\\\sigma\end{bmatrix}\|_2=
\min_{[\w^T \sigma]^T\in S_{\w}(\sigma,\xtilde,C_\w)}\max_{\|\a\|_2=1}  \a^T A_{\v}\begin{bmatrix} \w\\\sigma\end{bmatrix}.\label{eq:objlassol14}
\end{equation}
We now state a lemma from \cite{Gordon88} that will be of use in what follows.
\begin{lemma}(\cite{Gordon88})
Let $A$ be an $m\times n$ matrix with i.i.d. standard normal components. Let $\g$ and $\h$ be $m\times 1$ and $n\times 1$ vectors, respectively, with i.i.d. standard normal components. Also, let $g$ be a standard normal random variable and let $\Phi\subset R^n$ be an arbitrary subset. Then for all choices of real $\psi_{\phi}$
\begin{equation}
P(\min_{\phi\in \Phi}\max_{\|\a\|_2=1}(\a^T A\phi +\|\phi\|_2 g-\psi_{\phi})\geq 0)\geq P(\min_{\phi\in \Phi}\max_{\|\a\|_2=1}(\|\phi\|_2\sum_{i=1}^{m}\g_i\a_i+\sum_{i=1}^{n}\h_i\phi_i-\psi_{\phi})\geq 0).\label{eq:problemma}
\end{equation}\label{thm:unsignedlemma}
\end{lemma}
Now, after applying Lemma \ref{thm:unsignedlemma} one has
\begin{multline}
\hspace{-0in}P(\min_{[\w^T \sigma]^T\in S_{\w}(\sigma,\xtilde,C_\w)}(f_{obj}(\sigma,\w)+\sqrt{\|\w\|_2^2+\sigma^2}g)\geq \zeta_{obj}^{(l)})\\=P\left (\min_{[\w^T \sigma]^T\in S_{\w}(\sigma,\xtilde,C_\w)}\max_{\|\a\|_2=1}\left (  \a^T A_{\v}\begin{bmatrix} \w\\\sigma\end{bmatrix}+\sqrt{\|\w\|_2^2+\sigma^2}g\right ) \geq \zeta_{obj}^{(l)}\right ) \\
\geq P\left (\min_{[\w^T \sigma]^T\in S_{\w}(\sigma,\xtilde,C_\w)}\max_{\|\a\|_2=1}\left (  \sqrt{\|\w\|_2^2+\sigma^2}\sum_{i=1}^{m}\g_i\a_i+\sum_{i=1}^{n}\h_i\w_i+\h_{n+1}\sigma\right ) \geq \zeta_{obj}^{(l)}\right ).\label{eq:objlassol15}
\end{multline}
In what follows we will analyze the following probability
\begin{equation}
p_l=P\left (\min_{[\w^T \sigma]^T\in S_{\w}(\sigma,\xtilde,C_\w)}\max_{\|\a\|_2=1}\left (  \sqrt{\|\w\|_2^2+\sigma^2}\sum_{i=1}^{m}\g_i\a_i+\sum_{i=1}^{n}\h_i\w_i+\h_{n+1}\sigma\right ) \geq \zeta_{obj}^{(l)}\right ),\label{eq:probint}
\end{equation}
which is of course nothing but the probability on the left-hand side of the inequality in (\ref{eq:objlassol15}). We will essentially show that for certain $\zeta_{obj}^{(l)}$ this probability is close to $1$. That will rather obviously imply that we have a ``high probability" lower bound on $\zeta_{obj}$. To that end, we first note that the maximization over $\a$ is trivial and one obtains
\begin{equation}
p_l=P\left (\min_{[\w^T \sigma]^T\in S_{\w}(\sigma,\xtilde,C_\w)}\left (  \sqrt{\|\w\|_2^2+\sigma^2}\|\g\|_2+\sum_{i=1}^{n}\h_i\w_i\right )+\h_{n+1}\sigma \geq \zeta_{obj}^{(l)}\right ).\label{eq:probint1}
\end{equation}
To facilitate the exposition that will follow let
\begin{equation}
\xi(\sigma,\g,\h,\xtilde)=\min_{[\w^T \sigma]^T\in S_{\w}(\sigma,\xtilde,C_\w)} \left ( \sqrt{\|\w\|_2^2+\sigma^2}\|\g\|_2+\sum_{i=1}^{n}\h_i\w_i\right ).\label{eq:defxi}
\end{equation}
One should note here that, although present in the definition of $S_{\w}$,  $\sigma$ clearly does not have an impact on the result of the above optimization.
Now we split the analysis into two parts. The first one will be the deterministic analysis of $\xi(\sigma,\g,\h,\xtilde)$ and will be presented in Subsection \ref{sec:unsigneddet}. In the second part (that will be presented in Subsection \ref{sec:unsignedconc}) we will use the results of such a deterministic analysis and continue the above probabilistic analysis applying various concentration results.

\subsubsection{Optimizing $\xi(\sigma,\g,\h,\xtilde)$} \label{sec:unsigneddet}

In this section we compute $\xi(\sigma,\g,\h)$. We first rewrite the optimization problem from (\ref{eq:defxi}) in the following possibly clearer form
\begin{eqnarray}
\xi(\sigma,\g,\h,\xtilde)=\min_{\w} & & \sqrt{\|\w\|_2^2+\sigma^2}\|\g\|_2+\sum_{i=1}^{n}\h_i\w_i \nonumber \\
\mbox{subject to} & & \|\xtilde+\w\|_1\leq \|\xtilde\|_1\nonumber \\
& & \sqrt{\|\w\|_2^2+\sigma^2}\leq \sqrt{C_\w^2+\sigma^2}.\label{eq:defxi2}
\end{eqnarray}
To remove the absolute values we introduce auxiliary variables $\t_i,1\leq i\leq n$ and transform the above problem to
\begin{eqnarray}
\xi(\sigma,\g,\h,\xtilde)=\min_{\w,\t} & & \sqrt{\|\w\|_2^2+\sigma^2}\|\g\|_2+\sum_{i=1}^{n}\h_i\w_i\nonumber \\
\mbox{subject to} & & \sum_{i=1}^n \t_i\leq \|\xtilde\|_1\nonumber \\
& & \xtilde_i+\w_i-\t_i \leq 0, n-k+1\leq i\leq n\nonumber \\
& & -\xtilde_i-\w_i-\t_i\leq 0, n-k+1\leq i\leq n\nonumber \\
& & \w_i-\t_i\leq 0, 1\leq i\leq n-k\nonumber \\
& & -\w_i-\t_i\leq 0, 1\leq i\leq n-k\nonumber \\
& & \sqrt{\|\w\|_2^2+\sigma^2}\leq \sqrt{C_\w^2+\sigma^2}.\label{eq:defxi4}
\end{eqnarray}
The Lagrange dual of the above problem then becomes
\begin{multline}
{\cal L}(\nu,\lambda^{(1)},\lambda^{(2)},\w,\t,\gamma)=\sqrt{\|\w\|_2^2+\sigma^2}\|\g\|_2+\sum_{i=1}^{n}\h_i\w_i +\nu\sum_{i=1}^n \t_i-\nu\|\xtilde\|_1+\sum_{i=n-k+1}^{n}\lambda_i^{(1)}(\xtilde_i+\w_i-\t_i)\\
+\sum_{i=n-k+1}^{n}\lambda_i^{(2)}(-\xtilde_i-\w_i-\t_i)
+\sum_{i=1}^{n-k}\lambda_i^{(1)}(\w_i-\t_i)+\sum_{i=1}^{n-k}\lambda_i^{(2)}(-\w_i-\t_i)+\gamma(\sqrt{\|\w\|_2^2+\sigma^2}- \sqrt{C_\w^2+\sigma^2}).\label{eq:Lagran1}
\end{multline}
After rearranging the terms we further have
\begin{multline}
{\cal L}(\nu,\lambda^{(1)},\lambda^{(2)},\w,\t,\gamma)=\sqrt{\|\w\|_2^2+\sigma^2}\|\g\|_2+\sum_{i=1}^{n}\h_i\w_i -\nu\|\xtilde\|_1+\sum_{i=1}^n \t_i(\nu-\lambda_i^{(1)}-\lambda_i^{(2)})  +\sum_{i=n-k+1}^{n}\lambda_i^{(1)}(\xtilde_i+\w_i)\\
+\sum_{i=n-k+1}^{n}\lambda_i^{(2)}(-\xtilde_i-\w_i)
+\sum_{i=1}^{n-k}\lambda_i^{(1)}\w_i-\sum_{i=1}^{n-k}\lambda_i^{(2)}\w_i+\gamma(\sqrt{\|\w\|_2^2+\sigma^2}- \sqrt{C_\w^2+\sigma^2}).\label{eq:Lagran2}
\end{multline}
After a few further arrangements we finally have
\begin{multline}
{\cal L}(\nu,\lambda^{(1)},\lambda^{(2)},\w,\t,\gamma)=\sqrt{\|\w\|_2^2+\sigma^2}(\|\g\|_2+\gamma)+\sum_{i=1}^{n}\h_i\w_i -\nu\|\xtilde\|_1+\sum_{i=1}^n \t_i(\nu-\lambda_i^{(1)}-\lambda_i^{(2)}) \\ +\sum_{i=n-k+1}^{n}(\lambda_i^{(1)}-\lambda_i^{(2)})\xtilde_i
+\sum_{i=1}^{n}(\lambda_i^{(1)}-\lambda_i^{(2)})\w_i-\gamma \sqrt{C_\w^2+\sigma^2}.\label{eq:Lagran3}
\end{multline}
Setting $(\nu-\lambda_i^{(1)}-\lambda_i^{(2)})=0,1\leq i\leq n$, (to insure that the dual is bounded) and combining (\ref{eq:defxi4}) and (\ref{eq:Lagran3}) is enough to obtain
\begin{eqnarray}
\xi(\sigma,\g,\h,\xtilde)=\max_{\nu,\lambda^{(1)},\lambda^{(2)},\gamma}\min_{\w,\t} & & {\cal L}(\nu,\lambda^{(1)},\lambda^{(2)},\w,\t)\nonumber \\
\mbox{subject to} & & \lambda_j^{(i)}\geq 0, 1\leq j\leq n, 1\leq i \leq 2\nonumber \\
& & \nu\geq 0\nonumber \\
& & \nu-\lambda_i^{(1)}-\lambda_i^{(2)}=0,1\leq i\leq n\nonumber \\
& & \gamma\geq 0,\label{eq:Lagran4}
\end{eqnarray}
where we of course use the fact that the strict duality obviously holds. After removing the minimization over $\t$ we have
\begin{eqnarray}
\xi(\sigma,\g,\h,\xtilde)=\max_{\nu,\lambda^{(1)},\lambda^{(2)},\gamma}\min_{\w} & & {\cal L}(\nu,\lambda^{(1)},\lambda^{(2)},\w,\gamma)\nonumber \\
\mbox{subject to} & & \lambda^{(i)}\geq 0, 1\leq j\leq n, 1\leq i \leq 2\nonumber \\
& & \nu\geq 0\nonumber \\
& & \nu-\lambda_i^{(1)}-\lambda_i^{(2)}=0,1\leq i\leq n\nonumber \\
& & \gamma\geq 0.\label{eq:Lagran5}
\end{eqnarray}
where
\begin{equation}
\hspace{-.8in}{\cal L}(\nu,\lambda^{(1)},\lambda^{(2)},\w,\gamma)=\sqrt{\|\w\|_2^2+\sigma^2}(\|\g\|_2+\gamma)+\sum_{i=1}^{n}\h_i\w_i -\nu\|\xtilde\|_1  +\sum_{i=n-k+1}^{n}(\lambda_i^{(1)}-\lambda_i^{(2)})\xtilde_i+\sum_{i=1}^{n}(\lambda_i^{(1)}-\lambda_i^{(2)})\w_i-\gamma \sqrt{C_\w^2+\sigma^2}.\label{eq:shortlagran}
\end{equation}
The inner minimization over $\w$ is now doable. Setting the derivatives with respect to $\w_i$ to zero one obtains
\begin{equation}
\frac{\w(\|\g\|_2+\gamma)}{\sqrt{\|\w\|_2^2+\sigma^2}}+(\h+\lambda^{(1)}-\lambda^{(2)})=0,\label{eq:solvew}
\end{equation}
where $\lambda^{(1)}=[\lambda_1^{(1)},\lambda_2^{(1)},\dots,\lambda_n^{(1)}]^T$, $\lambda^{(2)}=[\lambda_1^{(2)},\lambda_2^{(2)},\dots,\lambda_n^{(2)}]^T$.
From (\ref{eq:solvew}) one then has
\begin{equation}
\w(\|\g\|_2+\gamma)=-\sqrt{\|\w\|_2^2+\sigma^2}(\h+\lambda^{(1)}-\lambda^{(2)})\label{eq:solvew1}
\end{equation}
or in a norm form
\begin{equation}
\|\w\|_2^2(\|\g\|_2+\gamma)^2=(\|\w\|_2^2+\sigma^2)\|\h+\lambda^{(1)}-\lambda^{(2)}\|_2^2.\label{eq:solvew2}
\end{equation}
From (\ref{eq:solvew2}) we then find
\begin{equation}
\|\w_{sol}\|_2=\frac{\sigma\|\h+\lambda^{(1)}-\lambda^{(2)}\|_2}{\sqrt{(\|\g\|_2+\gamma)^2-\|\h+\lambda^{(1)}-\lambda^{(2)}\|_2^2}},\label{eq:solvew3}
\end{equation}
and from (\ref{eq:solvew1})
\begin{equation}
\w_{sol}=\frac{\sigma(\h+\lambda^{(1)}-\lambda^{(2)})}{\sqrt{(\|\g\|_2+\gamma)^2-\|\h+\lambda^{(1)}-\lambda^{(2)}\|_2^2}}\label{eq:solvew4}
\end{equation}
where $\w_{sol}$ is of course the solution of the inner minimization over $\w$. Now, one should note that (\ref{eq:solvew3}) and (\ref{eq:solvew4}) are of course possible only if $\|\g\|_2+\gamma-\|\h+\lambda^{(1)}-\lambda^{(2)}\|_2\geq 0$. Later in the paper we will recognize, that for $\lambda^{(1)}$ and $\lambda^{(2)}$ that are optimal in (\ref{eq:Lagran5}), validity of this condition essentially implies the regime (in $(\alpha,\beta)$ plane) where the worst-case $\|\w\|_2$ is finite with overwhelming probability (or equivalently, if for such $\lambda^{(1)}$ and $\lambda^{(2)}$ the condition is not valid then for the corresponding ($\alpha,\beta$) the worst-case $\|\w\|_2$ is infinite with overwhelming probability). Plugging the value of $\w_{sol}$ from (\ref{eq:solvew4}) back in (\ref{eq:Lagran5}) gives
\begin{eqnarray}
\hspace{-.5in}\xi(\sigma,\g,\h,\xtilde)=\max_{\nu,\lambda^{(1)},\lambda^{(2)},\gamma} & & \sigma\sqrt{(\|\g\|_2+\gamma)^2-\|\h+\lambda^{(1)}-\lambda^{(2)}\|_2^2} -\nu\|\xtilde\|_1  +\sum_{i=n-k+1}^{n}(\lambda_i^{(1)}-\lambda_i^{(2)})\xtilde_i-\gamma \sqrt{C_\w^2+\sigma^2}\nonumber \\
\mbox{subject to} & & \lambda^{(i)}\geq 0, 1\leq j\leq n, 1\leq i \leq 2\nonumber \\
& & \nu\geq 0\nonumber \\
& & \nu-\lambda_i^{(1)}-\lambda_i^{(2)}=0,1\leq i\leq n\nonumber \\
& & \|\g\|_2+\gamma-\|\h+\lambda^{(1)}-\lambda^{(2)}\|_2\geq 0\nonumber \\
& & \gamma\geq 0.\label{eq:Lagran7}
\end{eqnarray}
Let $\z^{(1)}=[1,1,\dots,1]^T$. By plugging the constraint $\lambda^{(1)}=\nu\z^{(1)}-\lambda^{(2)}$ back into the objective function and making sure that $\nu-\lambda_i^{(2)}\geq 0,1\geq i\geq n$, one can remove $\lambda^{(1)}$ from the above optimization and get the following
\begin{eqnarray}
\hspace{-.5in}\xi(\sigma,\g,\h,\xtilde)=\max_{\nu,\lambda^{(2)},\gamma} & & \sigma\sqrt{(\|\g\|_2+\gamma)^2-\|\h+\nu\z^{(1)}-2\lambda^{(2)}\|_2^2}-\nu\|\xtilde\|_1  +\sum_{i=n-k+1}^{n}(\nu-2\lambda_i^{(2)})\xtilde_i-\gamma \sqrt{C_\w^2+\sigma^2}\nonumber \\
\mbox{subject to}
& & \nu\geq 0\nonumber \\
& & 0 \leq\lambda_i^{(2)}\leq \nu,1\leq i\leq n\nonumber \\
& & \|\g\|_2+\gamma-\|\h+\nu\z^{(1)}-2\lambda^{(2)}\|_2\geq 0\nonumber \\
& & \gamma\geq 0.\label{eq:Lagran8}
\end{eqnarray}
Since we assumed that $\xtilde_i\geq 0,n-k+1\leq i\leq n$, and $\xtilde_i=0,1\leq i\leq n-k$ one then from (\ref{eq:Lagran8}) has
\begin{eqnarray}
\xi(\sigma,\g,\h,\xtilde)=\max_{\nu,\lambda^{(2)},\gamma} & & \sigma\sqrt{(\|\g\|_2+\gamma)^2-\|\h+\nu\z^{(1)}-2\lambda^{(2)}\|_2^2} -2\sum_{i=n-k+1}^{n}\lambda_i^{(2)}\xtilde_i-\gamma \sqrt{C_\w^2+\sigma^2}\nonumber \\
\mbox{subject to}
& & \nu\geq 0\nonumber \\
& & 0 \leq\lambda_i^{(2)}\leq \nu,1\leq i\leq n\nonumber \\
& & \|\g\|_2+\gamma-\|\h+\nu\z^{(1)}-2\lambda^{(2)}\|_2\geq 0\nonumber \\
& & \gamma\geq 0.\label{eq:Lagran9}
\end{eqnarray}
After a simple scaling of $\lambda^{(2)}$ one finds that the following is an equivalent to (\ref{eq:Lagran9})
\begin{eqnarray}
\xi(\sigma,\g,\h,\xtilde)=\max_{\nu,\lambda^{(2)},\gamma} & & \sigma\sqrt{(\|\g\|_2+\gamma)^2-\|\h+\nu\z^{(1)}-\lambda^{(2)}\|_2^2} -\sum_{i=n-k+1}^{n}\lambda_i^{(2)}\xtilde_i-\gamma \sqrt{C_\w^2+\sigma^2}\nonumber \\
\mbox{subject to}
& & \nu\geq 0\nonumber \\
& & 0 \leq\lambda_i^{(2)}\leq 2\nu,1\leq i\leq n\nonumber \\
& & \|\g\|_2+\gamma-\|\h+\nu\z^{(1)}-\lambda^{(2)}\|_2\geq 0\nonumber \\
& & \gamma\geq 0.\label{eq:Lagran10}
\end{eqnarray}
Now, the maximization over $\gamma$ can be done. After setting the derivative to zero one finds
\begin{equation}
\frac{\|\g\|_2+\gamma}{\sqrt{(\|\g\|_2+\gamma)^2-\|\h+\nu\z^{(1)}-\lambda^{(2)}\|_2^2}}-\sqrt{C_\w^2+\sigma^2}=0\label{eq:dergamma}
\end{equation}
and after some algebra
\begin{equation}
\gamma_{opt}=\sqrt{1+\frac{\sigma^2}{C_\w^2}}\|\h+\nu\z^{(1)}-\lambda^{(2)}\|_2-\|\g\|_2,\label{eq:optgamma}
\end{equation}
where of course $\gamma_{opt}$ would be the solution of (\ref{eq:Lagran10}) only if larger than or equal to zero. Alternatively of course $\gamma_{opt}=0$. Now, based on these two scenarios we distinguish two different optimization problems:
\begin{enumerate}
\item \underline{\emph{The ``overwhelming" optimization}}
\begin{eqnarray}
\xi_{ov}(\sigma,\g,\h,\xtilde)=\max_{\nu,\lambda^{(2)}} & & \sigma\sqrt{\|\g\|_2^2-\|\h+\nu\z^{(1)}-\lambda^{(2)}\|_2^2} -\sum_{i=n-k+1}^{n}\lambda_i^{(2)}\xtilde_i\nonumber \\
\mbox{subject to}
& & \nu\geq 0\nonumber \\
& & 0 \leq\lambda_i^{(2)}\leq 2\nu,1\leq i\leq n.\label{eq:Lagran12}
\end{eqnarray}
\item \underline{\emph{The ``non-overwhelming" optimization}}
\begin{eqnarray}
\xi_{nov}(\sigma,\g,\h,\xtilde)=\max_{\nu,\lambda^{(2)}} & & \sqrt{C_\w^2+\sigma^2}\|\g\|_2-C_\w\|\h+\nu\z^{(1)}-\lambda^{(2)}\|_2 -\sum_{i=n-k+1}^{n}\lambda_i^{(2)}\xtilde_i\nonumber \\
\mbox{subject to}
& & \nu\geq 0\nonumber \\
& & 0 \leq\lambda_i^{(2)}\leq 2\nu,1\leq i\leq n.\label{eq:Lagran13}
\end{eqnarray}
\end{enumerate}
The ``overwhelming" optimization is the equivalent to (\ref{eq:Lagran10}) if for its optimal values $\hat{\nu}$ and $\widehat{\lambda^{(2)}}$ holds
\begin{equation}
\sqrt{1+\frac{\sigma^2}{C_\w^2}}\|\h+\hat{\nu}\z^{(1)}-\widehat{\lambda^{(2)}}\|_2\leq \|\g\|_2,\label{eq:ovnoncond}
\end{equation}
We now summarize in the following lemma the results of this subsection.
\begin{lemma}
Let $\hat{\nu}$ and $\widehat{\lambda^{(2)}}$ be the solutions of (\ref{eq:Lagran12}) and analogously let $\tilde{\nu}$ and $\widetilde{\lambda^{(2)}}$ be the solutions of (\ref{eq:Lagran13}). Let $\xi(\sigma,\g,\h,\xtilde)$ be, as defined in (\ref{eq:defxi}), the optimal value of the objective function in (\ref{eq:defxi}). Then
\begin{equation}
\hspace{-.8in}\xi(\sigma,\g,\h,\xtilde)=\begin{cases}\sigma\sqrt{\|\g\|_2^2-\|\h+\hat{\nu}\z^{(1)}-\widehat{\lambda^{(2)}}\|_2^2} -\sum_{i=n-k+1}^{n}\widehat{\lambda_i^{(2)}}\xtilde_i, &
\mbox{if}\quad  \sqrt{1+\frac{\sigma^2}{C_\w^2}}\|\h+\hat{\nu}\z^{(1)}-\widehat{\lambda^{(2)}}\|_2\leq \|\g\|_2\\
\sqrt{C_\w^2+\sigma^2}\|\g\|_2-C_\w\|\h+\tilde{\nu}\z^{(1)}-\widetilde{\lambda^{(2)}}\|_2 -\sum_{i=n-k+1}^{n}\widetilde{\lambda_i^{(2)}}\xtilde_i, & \mbox{otherwise} \end{cases}.\label{eq:defhatxi}
\end{equation}
Moreover, let $\hat{\w}$ be the solution of (\ref{eq:defxi}). Then
\begin{equation}
\hat{\w}(\sigma,\g,\h,\xtilde)=\begin{cases}
\frac{\sigma(\h+\hat{\nu}\z^{(1)}-\widehat{\lambda^{(2)}})}{\sqrt{\|\g\|_2^2-\|\h+\hat{\nu}\z^{(1)}-\widehat{\lambda^{(2)}}\|_2^2}}, &
\mbox{if}\quad  \sqrt{1+\frac{\sigma^2}{C_\w^2}}\|\h+\hat{\nu}\z^{(1)}-\widehat{\lambda^{(2)}}\|_2\leq \|\g\|_2\\
\frac{C_\w(\h+\tilde{\nu}\z^{(1)}-\widetilde{\lambda^{(2)}})}{\|\h+\tilde{\nu}\z^{(1)}-\widetilde{\lambda^{(2)}}\|_2}, &
\mbox{otherwise}\end{cases},\label{eq:defhatw}
\end{equation}
and
\begin{equation}
\|\hat{\w}(\sigma,\g,\h,\xtilde)\|_2=\begin{cases}
\frac{\sigma\|\h+\hat{\nu}\z^{(1)}-\widehat{\lambda^{(2)}})\|_2}{\sqrt{\|\g\|_2^2-\|\h+\hat{\nu}\z^{(1)}-\widehat{\lambda^{(2)}}\|_2^2}}, &
\mbox{if}\quad  \sqrt{1+\frac{\sigma^2}{C_\w^2}}\|\h+\hat{\nu}\z^{(1)}-\widehat{\lambda^{(2)}}\|_2\leq \|\g\|_2\\
C_\w, & \mbox{otherwise}
\end{cases}.
\label{eq:defhatwnorm}
\end{equation}\label{thm:optsollower}
\end{lemma}
\begin{proof}
The first part follows trivially. The second one follows from (\ref{eq:solvew4}) by choosing the optimal $\hat{\nu}$ and $\widehat{\lambda^{(2)}}$ or alternatively $\tilde{\nu}$ and $\widetilde{\lambda^{(2)}}$.
\end{proof}

\subsubsection{Concentration of $\xi(\sigma,\g,\h,\xtilde)$} \label{sec:unsignedconc}

In this section we will show that $\xi(\sigma,\g,\h,\xtilde)$ concentrates with high probability around its mean. To do so we will instead of looking at (\ref{eq:defhatxi}) look back at (\ref{eq:defxi}) which is the original definition of $\xi(\sigma,\g,\h,\xtilde)$. Now, before proceeding further we first recall on the following incredible result from \cite{CIS76} related to the concentrations of Lipschitz functions of Gaussian random variables.
\begin{lemma}[\cite{CIS76,Pisier86}]
Let $f_{lip}(\cdot):R^n\longrightarrow R$ be a Lipschitz function such that $|f_{lip}(\a)-f_{lip}(\b)|\leq c_{lip}\|\a-\b\|_2$. Let $\a$ be a vector comprised of i.i.d. zero-mean, unit variance Gaussian random variables and let $\epsilon_{lip}>0$. Then
\begin{equation}
P(|f_{lip}(\a)-Ef_{lip}(\a)|\geq \epsilon_{lip}Ef_{lip}(\a))\leq \exp \left \{  -\frac{(\epsilon_{lip} Ef_{lip}(\a))^2}{2c_{lip}^2} \right \}.\label{eq:lipsch}
\end{equation}\label{thm:lipsch}
\end{lemma}
In the following lemma we will show that $\xi(\sigma,\g,\h,\xtilde)$ is a Lipschitz function. To do so, we will, roughly speaking, assume that $\|\w\|_2$ in the definition of $\xi(\sigma,\g,\h,\xtilde)$ is bounded by a large constants say $C_\w$.  We recall here that our goal in this paper, though, is much bigger than creating a ``constant type" bound on $\|\w\|_2$. Namely, we will actually establish the precise value that $\|\w\|_2$ takes in the worst case with overwhelming probability. Clearly, knowing that one could then use much better value than $C_\w$ to upper bound $\|\w\|_2$ in the definition of $\xi(\sigma,\g,\h,\xtilde)$. However, for the purposes of the concentration inequalities any constant (of course independent of $n$) is fine. In fact, any sub-root dependence on $n$ would be fine too, it is just that in that case ``overwhelming" wouldn't be negative exponential any more.
\begin{lemma}
Let $\g$ and $\h$ be $m$ and $n$ dimensional vectors, respectively, with i.i.d. standard normal variables as their components. Let $\sigma>0$ be an arbitrary scalar. Let $\xi(\sigma,\g,\h,\xtilde)$ be as in (\ref{eq:defxi}). Further let $\epsilon_{lip}>0$ be any constant. Then
\begin{equation}
P(|\xi(\sigma,\g,\h,\xtilde)-E\xi(\sigma,\g,\h,\xtilde)|\geq \epsilon_{lip}E\xi(\sigma,\g,\h,\xtilde))\leq \exp \left \{  -\frac{(\epsilon_{lip} E\xi(\sigma,\g,\h,\xtilde))^2}{2(2C_\w^2+\sigma^2)} \right \}.\label{eq:lipsch1}
\end{equation}
\label{thm:lipschunsigned}
\end{lemma}
\begin{proof}We start by setting
\begin{equation}
f_{lip}(\g^{(1)},\h^{(1)})=\min_{[\w^T \sigma]^T\in S_{\w}(\sigma,\xtilde,C_\w)} \left ( \sqrt{\|\w\|_2^2+\sigma^2}\|\g^{(1)}\|_2+\sum_{i=1}^{n}\h_i^{(1)}\w_i\right ).\label{eq:lipproof1}
\end{equation}
Further, let $\w_{lip}^{(1)}$ be the solution of the minimization in (\ref{eq:lipproof1}). Then, clearly
\begin{equation}
f_{lip}(\g^{(1)},\h^{(1)})= \left ( \sqrt{\|\w_{lip}^{(1)}\|_2^2+\sigma^2}\|\g^{(1)}\|_2+\sum_{i=1}^{n}\h_i^{(1)}(\w_{lip}^{(1)})_i\right ),\label{eq:lipproof2}
\end{equation}
where $(\w_{lip}^{(1)})_i$ is the $i$-th index of $\w_{lip}^{(1)}$. In an analogous fashion set
\begin{equation}
f_{lip}(\g^{(2)},\h^{(2)})=\min_{[\w^T \sigma]^T\in S_{\w}(\sigma,\xtilde,C_\w)} \left ( \sqrt{\|\w\|_2^2+\sigma^2}\|\g^{(2)}\|_2+\sum_{i=1}^{n}\h_i^{(2)}\w_i\right ),\label{eq:lipproof3}
\end{equation}
and let $\w_{lip}^{(2)}$ be the solution of the minimization in (\ref{eq:lipproof2}). Then again clearly
\begin{equation}
f_{lip}(\g^{(2)},\h^{(2)})= \left ( \sqrt{\|\w_{lip}^{(2)}\|_2^2+\sigma^2}\|\g^{(2)}\|_2+\sum_{i=1}^{n}\h_i^{(2)}(\w_{lip}^{(2)})_i\right ),\label{eq:lipproof4}
\end{equation}
where of course $(\w_{lip}^{(2)})_i$ is the $i$-th index of $\w_{lip}^{(2)}$. Now assume that $f_{lip}(\g^{(1)},\h^{(1)})\neq f_{lip}(\g^{(2)},\h^{(2)})$ (if they are equal we are trivially done). Further let $f_{lip}(\g^{(1)},\h^{(1)})< f_{lip}(\g^{(2)},\h^{(2)})$ (the rest of the argument of course can trivially be flipped if $f_{lip}(\g^{(1)},\h^{(1)})> f_{lip}(\g^{(2)},\h^{(2)})$). We then have
\begin{multline}
|f_{lip}(\g^{(2)},\h^{(2)})- f_{lip}(\g^{(1)},\h^{(1)})|=f_{lip}(\g^{(2)},\h^{(2)})- f_{lip}(\g^{(1)},\h^{(1)})\\
= \left ( \sqrt{\|\w_{lip}^{(2)}\|_2^2+\sigma^2}\|\g^{(2)}\|_2+\sum_{i=1}^{n}\h_i^{(2)}(\w_{lip}^{(2)})_i\right )-
\left ( \sqrt{\|\w_{lip}^{(1)}\|_2^2+\sigma^2}\|\g^{(1)}\|_2+\sum_{i=1}^{n}\h_i^{(1)}(\w_{lip}^{(1)})_i\right )\\ \leq
\left ( \sqrt{\|\w_{lip}^{(1)}\|_2^2+\sigma^2}\|\g^{(2)}\|_2+\sum_{i=1}^{n}\h_i^{(2)}(\w_{lip}^{(1)})_i\right )-
\left ( \sqrt{\|\w_{lip}^{(1)}\|_2^2+\sigma^2}\|\g^{(1)}\|_2+\sum_{i=1}^{n}\h_i^{(1)}(\w_{lip}^{(1)})_i\right )\\
=\sqrt{\|\w_{lip}^{(1)}\|_2^2+\sigma^2}(\|\g^{(2)}\|_2-\|\g^{(1)}\|_2)+\sum_{i=1}^{n}(\h_i^{(2)}-\h_i^{(1)})(\w_{lip}^{(1)})_i\\ \leq
\sqrt{\|\w_{lip}^{(1)}\|_2^2+\sigma^2}(\|\g^{(2)}-\g^{(1)}\|_2)+\|(\h^{(2)}-\h^{(1)}\|_2\|\w_{lip}^{(1)}\|_2\\ \leq
\sqrt{2\|\w_{lip}^{(1)}\|_2^2+\sigma^2}\sqrt{\|\g^{(2)}-\g^{(1)}\|_2^2+\|\h^{(2)}-\h^{(1)}\|_2^2}\\ \leq
\sqrt{2C_\w^2+\sigma^2}\sqrt{\|\g^{(2)}-\g^{(1)}\|_2^2+\|\h^{(2)}-\h^{(1)}\|_2^2},\label{eq:lipproof5}
\end{multline}
where the first inequality follows by sub-optimality of $\w_{lip}^{(1)}$ in (\ref{eq:lipproof3}). Connecting beginning and end in
(\ref{eq:lipproof5}) and combining it with (\ref{eq:lipproof1}) one then has that
$\xi(\sigma,\g,\h,\xtilde)$ is Lipschitz with $c_{lip}=\sqrt{2C_\w^2+\sigma^2}$. (\ref{eq:lipsch1}) then easily follows by Lemma \ref{thm:lipsch}.
\end{proof}
One then has that $\|\h+\hat{\nu}\z^{(1)}-\widehat{\lambda^{(2)}}\|_2$ and $\|\h+\tilde{\nu}\z^{(1)}-\widetilde{\lambda^{(2)}}\|_2$ concentrate as well which automatically implies that $\hat{\w}$ also concentrates. More formally, one then has analogues to (\ref{eq:lipsch1})
\begin{eqnarray}
P(|\|\h+\hat{\nu}\z^{(1)}-\widehat{\lambda^{(2)}}\|_2-E\|\h+\hat{\nu}\z^{(1)}-\widehat{\lambda^{(2)}}\|_2|\geq
\epsilon_1^{(norm)}E\|\h+\hat{\nu}\z^{(1)}-\widehat{\lambda^{(2)}}\|_2) & \leq & e^{-\epsilon_2^{(norm)}n}\nonumber \\
P(|\|\h+\tilde{\nu}\z^{(1)}-\widetilde{\lambda^{(2)}}\|_2-E\|\h+\tilde{\nu}\z^{(1)}-\widetilde{\lambda^{(2)}}\|_2|\geq
\epsilon_3^{(norm)}E\|\h+\tilde{\nu}\z^{(1)}-\widetilde{\lambda^{(2)}}\|_2) & \leq & e^{-\epsilon_4^{(norm)}n}\nonumber \\
P(|\|\hat{\w}\|_2-E\|\hat{\w}\|_2|\geq
\epsilon_1^{(\w)}E\|\hat{\w}\|_2) & \leq & e^{-\epsilon_2^{(\w)}n},\label{eq:conchw}
\end{eqnarray}
where as usual $\epsilon_1^{(norm)}>0$, $\epsilon_2^{(norm)}>0$, and $\epsilon_1^{(\w)}>0$ are arbitrarily small constants and $\epsilon_3^{(norm)}$, $\epsilon_4^{(norm)}$, and $\epsilon_2^{(\w)}$ are constant dependent on $\epsilon_1^{(norm)}>0$, $\epsilon_2^{(norm)}>0$, and $\epsilon_1^{(\w)}>0$, respectively, but independent of $n$.

Now, we return to the probabilistic analysis of (\ref{eq:probint1}). Combining (\ref{eq:probint1}), (\ref{eq:defxi}), and (\ref{eq:lipsch1}) we have
\begin{eqnarray}
p_l & = & P\left (\min_{[\w^T \sigma]^T\in S_{\w}(\sigma,\xtilde,C_\w)}\left (  \sqrt{\|\w\|_2^2+\sigma^2}\|\g\|_2+\sum_{i=1}^{n}\h_i\w_i\right )+\h_{n+1}\sigma \geq \zeta_{obj}^{(l)}\right )\nonumber \\
& = & P\left (\xi(\sigma,\g,\h,\xtilde)+\h_{n+1}\sigma \geq \zeta_{obj}^{(l)}\right )\nonumber \\
& \geq & \left ( 1-\exp \left \{  -\frac{(\epsilon_{lip} E\xi(\sigma,\g,\h,\xtilde))^2}{2(2C_\w^2+\sigma^2)} \right \} \right )P\left ((1-\epsilon_{lip})E\xi(\sigma,\g,\h,\xtilde)+\h_{n+1}\sigma \geq \zeta_{obj}^{(l)}\right ).\nonumber \\\label{eq:probanalcont1}
\end{eqnarray}
Since $\h_{n+1}$ is a standard normal one easily has $P(\h_{n+1}\sigma\geq -\epsilon_1^{(\h)}\sqrt{n})\geq 1-e^{-\epsilon_2^{(\h)}n}$ where $\epsilon_1^{(\h)}>0$ is an arbitrarily small constant and $\epsilon_2^{(\h)}$ is a constant dependent on $\epsilon_1^{(\h)}$ and $\sigma$ but independent on $n$. By choosing
\begin{equation}
\zeta_{obj}^{(l)}=(1-\epsilon_{lip})E\xi(\sigma,\g,\h,\xtilde)-\epsilon_1^{(\h)}\sqrt{n},\label{eq:defzetaobjl}
\end{equation}
one then from (\ref{eq:probanalcont1}) has
\begin{eqnarray}
p_l & = & P\left (\min_{[\w^T \sigma]^T\in S_{\w}(\sigma,\xtilde,C_\w)}\left (  \sqrt{\|\w\|_2^2+\sigma^2}\|\g\|_2+\sum_{i=1}^{n}\h_i\w_i\right )+\h_{n+1}\sigma \geq (1-\epsilon_{lip})E\xi(\sigma,\g,\h,\xtilde)-\epsilon_1^{(\zeta)}\sqrt{n}\right )\nonumber \\
& \geq & \left ( 1-\exp \left \{  -\frac{(\epsilon_{lip} E\xi(\sigma,\g,\h,\xtilde))^2}{2(2C_\w^2+\sigma^2)} \right \} \right )(1-e^{-\epsilon_2^{(\h)}n}).\label{eq:probanalcont2}
\end{eqnarray}
As stated after (\ref{eq:probint}), (\ref{eq:probanalcont2}) is conceptually enough to establish a ``high probability" lower bound on $\zeta_{obj}$. The next few steps that formally do so are rather obvious but we include them for the completeness. Combining (\ref{eq:objlassol15}) and (\ref{eq:probanalcont2}) we obtain
\begin{multline}
P(\min_{[\w^T \sigma]^T\in S_{\w}(\sigma,\xtilde,C_\w)}(f_{obj}(\sigma,\w)+\sqrt{\|\w\|_2^2+\sigma^2}g)\geq \zeta_{obj}^{(l)})\\\geq \left ( 1-\exp \left \{  -\frac{(\epsilon_{lip} E\xi(\sigma,\g,\h,\xtilde))^2}{2(2C_\w^2+\sigma^2)} \right \} \right )(1-e^{-\epsilon_2^{(\h)}n}),,\label{eq:probanalcont22}
\end{multline}
where $\zeta_{obj}^{(l)}$ is as in (\ref{eq:defzetaobjl}). Now, one further has
\begin{multline}
P(\min_{[\w^T \sigma]^T\in S_{\w}(\sigma,\xtilde,C_\w)}(f_{obj}(\sigma,\w)+\sqrt{\|\w\|_2^2+\sigma^2}g)\geq \zeta_{obj}^{(l)})\\\leq
P(\min_{[\w^T \sigma]^T\in S_{\w}(\sigma,\xtilde,C_\w)}(f_{obj}(\sigma,\w))+\sqrt{C_\w^2+\sigma^2}g\geq \zeta_{obj}^{(l)}).\label{eq:probanalcont3}
\end{multline}
Since $g$ is a standard normal one easily again has $P(g\sqrt{C_\w^2+\sigma^2}\leq \epsilon_1^{(g)}\sqrt{n})\geq 1-e^{-\epsilon_1^{g)}n}$ where $\epsilon_1^{(g)}>0$ is an arbitrarily small constant and $\epsilon_2^{(g)}$ is a constant dependent on $\epsilon_1^{(g)}$, $\sigma$, and $C_\w$ but independent on $n$. Applying this to the first term on the right hand side of the above inequality one obtains
\begin{multline}
P(\min_{[\w^T \sigma]^T\in S_{\w}(\sigma,\xtilde,C_\w)}(f_{obj}(\sigma,\w))+\sqrt{C_\w^2+\sigma^2}g\geq \zeta_{obj}^{(l)})\\
\leq
P(\min_{[\w^T \sigma]^T\in S_{\w}(\sigma,\xtilde,C_\w)}(f_{obj}(\sigma,\w))\geq \zeta_{obj}^{(l)}-\epsilon_1^{(g)}\sqrt{n})+e^{-\epsilon_1^{g)}n}.\label{eq:probanalcont4}
\end{multline}
Now let $\zeta_{obj}^{lower}=\zeta_{obj}^l-\epsilon_1^{(g)}\sqrt{n}$. From (\ref{eq:defzetaobjl}) then obviously
\begin{equation}
\zeta_{obj}^{(lower)}=(1-\epsilon_{lip})E\xi(\sigma,\g,\h,\xtilde)-\epsilon_1^{(\h)}\sqrt{n}-\epsilon_1^{(g)}\sqrt{n}.\label{eq:defzetaobjlower}
\end{equation}
Also let $\epsilon_{lower}$ be a constant such that
\begin{equation}
1-e^{-\epsilon_{lower}n}\leq\left ( 1-\exp \left \{  -\frac{(\epsilon_{lip} E\xi(\sigma,\g,\h,\xtilde))^2}{2(2C_\w^2+\sigma^2)} \right \} \right )(1-e^{-\epsilon_2^{(\h)}n})-e^{-\epsilon_1^{g)}n}.\label{eq:defepslower}
\end{equation}
Then a combination of (\ref{eq:assumplasso}), (\ref{eq:probanalcont22}), (\ref{eq:probanalcont3}), (\ref{eq:probanalcont4}), (\ref{eq:defzetaobjlower}), and (\ref{eq:defepslower}) gives
\begin{multline}
P(\zeta_{obj}\geq \zeta_{obj}^{(lower)})\geq P(\zeta_{obj}^{(help)}\geq \zeta_{obj}^{(lower)})(1-e^{-\epsilon_{C_{\w}}n})\\=P(\min_{[\w^T \sigma]^T\in S_{\w}(\sigma,\xtilde,C_\w)}(f_{obj}(\sigma,\w))\geq \zeta_{obj}^{(lower)})\geq (1-e^{-\epsilon_{lower}n})(1-e^{-\epsilon_{C_{\w}}n}).\label{eq:lowerboundobj}
\end{multline}
We summarize the results from this subsection in the following lemma.
\begin{lemma}
Let $\v$ be an $n\times 1$ vector of i.i.d. zero-mean variance $\sigma^2$ Gaussian random variables and let $A$ be an $m\times n$ matrix of i.i.d. standard normal random variables. Consider an $\xtilde$ defined in (\ref{eq:xtildedef}) and a $\y$ defined in (\ref{eq:systemnoise}) for $\x=\xtilde$. Let then $\zeta_{obj}$ be as defined in (\ref{eq:objlassol1}) and let $\w$ be the solution of (\ref{eq:objlassol13}).
Assume $P(\|\w\|_2\leq C_\w)\geq 1-e^{-\epsilon_{C_\w}n}$ for an arbitrarily large constant $C_\w$ and a constant $\epsilon_{C_\w}>0$ dependent on $C_\w$ but independent of $n$. Then there is a constant $\epsilon_{lower}>0$
\begin{equation}
P(\zeta_{obj}\geq \zeta_{obj}^{(lower)})\geq (1-e^{-\epsilon_{lower}n})(1-e^{-\epsilon_{C_{\w}}n}),\label{eq:lowerboundobjthm1}
\end{equation}
where
\begin{equation}
\zeta_{obj}^{(lower)}=(1-\epsilon_{lip})E\xi(\sigma,\g,\h,\xtilde)-\epsilon_1^{(\h)}\sqrt{n}-\epsilon_1^{(g)}\sqrt{n},\label{eq:lowerboundobjthm2}
\end{equation}
$\xi(\sigma,\g,\h,\xtilde)$ is as defined in (\ref{eq:defxi}), and $\epsilon_{lip},\epsilon_1^{(\h)},\epsilon_1^{(g)}$ are all positive arbitrarily small constants.
\label{thm:lowerbound}
\end{lemma}
\begin{proof}
Follows from the previous discussion.
\end{proof}

\subsection{Upper-bounding $\zeta_{obj}$} \label{sec:unsignedubzetaobj}

In this section we present a general framework for finding a ``high-probability" upper bound on $\zeta_{obj}$. To that end, let $r$ and $C_{\w_{up}}$ be positive scalars (in this subsection we present a general framework and take these scalars to be arbitrary; however to make the bound as tight sa possible in the following subsection we will make them take particular values). Now, if we can show that there is a $\w\in R^n$ such that
$\|\xtilde+\w\|_1\leq \|\xtilde\|_1$ and $\|\v-A\w\|_2\leq r$ with overwhelming probability then $r$ can act as an upper bound on $\zeta_{obj}$. We then start by looking at the following optimization problem
\begin{eqnarray}
\min_{\w} & & \|\xtilde+\w\|_1-\|\xtilde\|_1 \nonumber \\
& & \|A_\v\begin{bmatrix}\w \\ \sigma\end{bmatrix}\|_2\leq r\nonumber \\
& & \|\w\|_2^2\leq C_{\w_{up}}^2,\label{eq:upperobjlassol11}
\end{eqnarray}
where $A_\v$ is as defined right after (\ref{eq:objlassol13}). If we can show that with overwhelming probability the objective value of the above optimization problem is negative then $r$ will be a valid ``high probability" upper-bound on $\zeta_{obj}$. Moreover, it will be achieved by a $\w$ for which it will hold that  $\|\w\|_2\leq C_{\w_{up}}$.

We now proceed in a fashion similar to the one from Subsection \ref{sec:unsigneddet}. To remove the absolute values we introduce auxiliary variables $\t_i,1\leq i\leq n$, and transform the above problem to
\begin{eqnarray}
\min_{\w,\t} & & \sum_{i=1}^n \t_i - \|\xtilde\|_1\nonumber \\
\mbox{subject to}
& & \xtilde_i+\w_i-\t_i \leq 0, n-k+1\leq i\leq n\nonumber \\
& & -\xtilde_i-\w_i-\t_i\leq 0, n-k+1\leq i\leq n\nonumber \\
& & \w_i-\t_i\leq 0, 1\leq i\leq n-k\nonumber \\
& & -\w_i-\t_i\leq 0, 1\leq i\leq n-k\nonumber \\
& & \|A_\v\begin{bmatrix}\w \\ \sigma\end{bmatrix}\|_2\leq r\nonumber \\
& & \|\w\|_2^2\leq C_{\w_{up}}^2.\label{eq:upperdefxi4}
\end{eqnarray}
We also slightly modify the first of the constraints from (\ref{eq:upperobjlassol11}) in the following way
\begin{eqnarray}
\min_{\w,\t,\b} & & \sum_{i=1}^n \t_i - \|\xtilde\|_1\nonumber \\
\mbox{subject to}
& & \xtilde_i+\w_i-\t_i \leq 0, n-k+1\leq i\leq n\nonumber \\
& & -\xtilde_i-\w_i-\t_i\leq 0, n-k+1\leq i\leq n\nonumber \\
& & \w_i-\t_i\leq 0, 1\leq i\leq n-k\nonumber \\
& & -\w_i-\t_i\leq 0, 1\leq i\leq n-k\nonumber \\
& & \|\b\|_2^2\leq r^2\nonumber \\
& &  \begin{bmatrix} -A \v\end{bmatrix}\begin{bmatrix}\w \\ \sigma\end{bmatrix}=b\nonumber \\
& & \|\w\|_2^2\leq C_{\w_{up}}^2.\label{eq:upperdefxi4}
\end{eqnarray}
The Lagrange dual of the above problem then becomes
\begin{multline}
{\cal L}(\lambda^{(1)},\lambda^{(2)},\nu^{(1)},\gamma_1,\gamma_2,\w,\t,\b)=\sum_{i=1}^n \t_i-\|\xtilde\|_1+\sum_{i=n-k+1}^{n}\lambda_i^{(1)}(\xtilde_i+\w_i-\t_i)
+\sum_{i=n-k+1}^{n}\lambda_i^{(2)}(-\xtilde_i-\w_i-\t_i)\\
+\sum_{i=1}^{n-k}\lambda_i^{(1)}(\w_i-\t_i)+\sum_{i=1}^{n-k}\lambda_i^{(2)}(-\w_i-\t_i)-\nu^{(1)} A\w +\nu^{(1)}\v\sigma -\nu^{(1)}\b+\gamma_1(\sum_{i=1}^{n}\b_1^2-r^2)+\gamma_2(\|\w\|_2^2- C_{\w_{up}}^2),\label{eq:upperLagran1}
\end{multline}
where $\nu^{(1)}$ is $1\times m$  row vector of Lagrange variables and $\lambda^{(1)},\lambda^{(2)}$ are as in previous sections.
After rearranging terms we further have
\begin{multline}
{\cal L}(\lambda^{(1)},\lambda^{(2)},\nu^{(1)},\gamma_1,\gamma_2,\w,\t,\b)=-\|\xtilde\|_1+\sum_{i=1}^n \t_i(1-\lambda_i^{(1)}-\lambda_i^{(2)})  +\sum_{i=n-k+1}^{n}\lambda_i^{(1)}(\xtilde_i+\w_i)\\
+\sum_{i=n-k+1}^{n}\lambda_i^{(2)}(-\xtilde_i-\w_i)
+\sum_{i=1}^{n-k}\lambda_i^{(1)}\w_i-\sum_{i=1}^{n-k}\lambda_i^{(2)}\w_i-\nu^{(1)} A\w +\nu^{(1)}\v\sigma -\nu^{(1)}\b+\gamma_1(\sum_{i=1}^{n}\b_1^2-r^2)+\gamma_2(\|\w\|_2^2- C_{\w_{up}}^2).\label{eq:upperLagran2}
\end{multline}
After a few further arrangements we finally have
\begin{multline}
{\cal L}(\lambda^{(1)},\lambda^{(2)},\nu^{(1)},\gamma_1,\gamma_2,\w,\t,\b)=\sum_{i=1}^n \t_i(1-\lambda_i^{(1)}-\lambda_i^{(2)}) +\sum_{i=n-k+1}^{n}(\lambda_i^{(1)}-\lambda_i^{(2)}-1)\xtilde_i\\
+((\lambda^{(1)}-\lambda^{(2)})^T-\nu^{(1)} A)\w +\nu^{(1)}\v\sigma -\nu^{(1)}\b+\gamma_1(\sum_{i=1}^{n}\b_1^2-r^2)+\gamma_2(\|\w\|_2^2- C_{\w_{up}}^2).\label{eq:upperLagran3}
\end{multline}
Setting $(\nu-\lambda_i^{(1)}-\lambda_i^{(2)})=0,1\leq i\leq n$, (to insure that the dual is bounded) we have
\begin{multline}
{\cal L}(\lambda^{(2)},\nu^{(1)},\gamma_1,\gamma_2,\w,\b)=-2\sum_{i=n-k+1}^{n}\lambda_i^{(2)}\xtilde_i\\
+((\z^{(1)}-2\lambda^{(2)})^T-\nu^{(1)} A)\w +\nu^{(1)}\v\sigma -\nu^{(1)}\b+\gamma_1(\sum_{i=1}^{n}\b_1^2-r^2)+\gamma_2(\|\w\|_2^2- C_{\w_{up}}^2).\label{eq:upperLagran4}
\end{multline}
Finally we can write a dual problem to (\ref{eq:upperdefxi4})
\begin{eqnarray}
\max_{\lambda^{(2)},\nu^{(1)},\gamma_1,\gamma_2}\min_{\w,\b} & & {\cal L}(\lambda^{(2)},\nu^{(1)},\gamma_1,\gamma_2,\w,\b)\nonumber \\
\mbox{subject to} & & 0\leq \lambda_i^{(2)}\leq 1, 1\leq i\leq n\nonumber \\
& & \gamma_1\geq 0,\nonumber \\
& & \gamma_2\geq 0,\label{eq:upperLagran5}
\end{eqnarray}
where we of course use the fact that the strict duality obviously holds. Now, we minimize over $\w$ by setting the derivatives to zero
\begin{equation}
\frac{d{\cal L}(\lambda^{(2)},\nu^{(1)},\gamma_1,\gamma_2,\w,\b)}{d\w}=((\z^{(1)}-2\lambda^{(2)})^T-\nu^{(1)} A)^T+2\gamma_2\w.\label{eq:upperLagran6}
\end{equation}
From (\ref{eq:upperLagran6}) we easily have
\begin{equation}
\w=\frac{((\z^{(1)}-2\lambda^{(2)})^T-\nu^{(1)} A)^T}{2\gamma_2}.\label{eq:upperLagran7}
\end{equation}
Plugging (\ref{eq:upperLagran6}) back in (\ref{eq:upperLagran4}) we further have
\begin{multline}
{\cal L}(\lambda^{(2)},\nu^{(1)},\gamma_1,\gamma_2,\b)=-2\sum_{i=n-k+1}^{n}\lambda_i^{(2)}\xtilde_i\\
-\frac{\|(\z^{(1)}-2\lambda^{(2)})^T-\nu^{(1)} A\|_2}{4\gamma_2} +\nu^{(1)}\v\sigma -\nu^{(1)}\b+\gamma_1(\sum_{i=1}^{n}\b_1^2-r^2)-\gamma_2 C_{\w_{up}}^2.\label{eq:upperLagran8}
\end{multline}
 Now, we minimize ${\cal L}(\lambda^{(1)},\lambda^{(2)},\nu^{(1)},\gamma_1,\gamma_2,\b)$ over $\b$ by setting the derivatives to zero
\begin{equation}
\frac{d{\cal L}(\lambda^{(1)},\lambda^{(2)},\nu^{(1)},\gamma_1,\gamma_2,\b)}{d\b}=-\nu^{(1)}+2\gamma_1\w.\label{eq:upperLagran9}
\end{equation}
From (\ref{eq:upperLagran9}) we easily have
\begin{equation}
\b=\frac{\nu^{(1)}}{2\gamma_1}.\label{eq:upperLagran10}
\end{equation}
Plugging (\ref{eq:upperLagran10}) back in (\ref{eq:upperLagran8}) we have
\begin{equation}
{\cal L}(\lambda^{(2)},\nu^{(1)},\gamma_1,\gamma_2)=-2\sum_{i=n-k+1}^{n}\lambda_i^{(2)}\xtilde_i+
-\frac{\|(\z^{(1)}-2\lambda^{(2)})^T-\nu^{(1)} A\|_2}{2\gamma_2} +\nu^{(1)}\v\sigma -\frac{\|\nu^{(1)}\|_2}{4\gamma_1}-\gamma_1 r^2-\gamma_2 C_{\w_{up}}^2,\label{eq:upperLagran11}
\end{equation}
and finally an equivalent to
(\ref{eq:upperdefxi4})
\begin{eqnarray}
\max_{\lambda^{(2)},\nu^{(1)},\gamma_1,\gamma_2} & & {\cal L}(\lambda^{(1)},\lambda^{(2)},\nu^{(1)},\gamma_1,\gamma_2)\nonumber \\
\mbox{subject to} & & 0\leq \lambda_i^{(2)}\leq 1, 1\leq i\leq n\nonumber \\
& & \gamma_1\geq 0,\nonumber \\
& & \gamma_2\geq 0.\label{eq:upperLagran12}
\end{eqnarray}
After doing the trivial maximization over $\gamma_1$ and $\gamma_2$ one obtains
\begin{eqnarray}
\max_{\lambda^{(2)},\nu^{(1)}} & & -2\sum_{i=n-k+1}^{n}\lambda_i^{(2)}\xtilde_i
-C_{\w_{up}}\|(\z^{(1)}-2\lambda^{(2)})^T-\nu^{(1)} A\|_2 +\nu^{(1)}\v\sigma -\|\nu^{(1)}\|_2r\nonumber \\
\mbox{subject to} & & 0\leq \lambda_i^{(2)}\leq 1, 1\leq i\leq n.\label{eq:upperLagran13}
\end{eqnarray}
We rewrite (\ref{eq:upperLagran13}) in a slightly more convenient form
\begin{eqnarray}
-\min_{\lambda^{(2)},\nu^{(1)}} \max_{\|\a\|_2=C_{\w_{up}}}& &
((\z^{(1)}-2\lambda^{(2)})^T-\nu^{(1)} A)\a -\nu^{(1)}\v\sigma +\|\nu^{(1)}\|_2r+2\sum_{i=n-k+1}^{n}\lambda_i^{(2)}\xtilde_i\nonumber \\
\mbox{subject to} & & 0\leq \lambda_i^{(2)}\leq 1, 1\leq i\leq n.\label{eq:upperLagran14}
\end{eqnarray}
Now let us define $f_{obj}^{(up)}$ as
\begin{eqnarray}
-f_{obj}^{(up)}=-\min_{\lambda^{(2)},\nu^{(1)}} \max_{\|\a\|_2=C_{\w_{up}}}& &
((\z^{(1)}-2\lambda^{(2)})^T-\nu^{(1)} A)\a -\nu^{(1)}\v\sigma +\|\nu^{(1)}\|_2r+2\sum_{i=n-k+1}^{n}\lambda_i^{(2)}\xtilde_i\nonumber \\
\mbox{subject to} & & 0\leq \lambda_i^{(2)}\leq 1, 1\leq i\leq n.\label{eq:upperLagran14}
\end{eqnarray}
Any $r$ such that $\lim_{n\rightarrow}P(f_{obj}^{(up)}\geq 0)=1$ is then a valid ``high-probability" upper bound.

We now introduce a refinement of a lemma from \cite{StojnicUpper10} which itself is a slightly modified Lemma \ref{thm:unsignedlemma} (Lemma \ref{thm:unsignedlemma} is of course the backbone of the escape through a mesh theorem utilized in \cite{StojnicCSetam09}).
\begin{lemma}
Let $A$ be an $m\times n$ matrix with i.i.d. standard normal components. Let $\g$ and $\h$ be $m\times 1$ and $(n+1)\times 1$ vectors, respectively, with i.i.d. standard normal components. Also, let $g$ be a standard normal random variable and let $\Lambda$ be a set such that $\Lambda=(\lambda^{(2)}|0\leq \lambda_i^{(2)}\leq 1, 1\leq i\leq n)$. Then
\begin{multline}
P(\min_{\lambda^{(2)}\in \Lambda,\nu^{(1)}\in R^n\setminus 0}\max_{\|\a\|_2=C_{\w_{up}}}(-\nu^{(1)}\begin{bmatrix} A & \v\end{bmatrix}\begin{bmatrix}\a \\\sigma\end{bmatrix} +\|\nu^{(1)}\|_2 g-\psi_{\a,\lambda^{(2)},\nu^{(1)}})\geq 0)\\\geq P(\min_{\lambda^{(2)}\in \Lambda,\nu^{(1)}\in R^n\setminus 0}\max_{\|\a\|_2=1}(\|\nu^{(1)}\|_2(\sum_{i=1}^{n}\h_i\a_i+\h_{n+1}\sigma)+\sqrt{C_{\w_{up}}^2+\sigma^2}\sum_{i=1}^{m}\g_i\nu_i^{(1)}-\psi_{\a,\lambda^{(2)},\nu^{(1)}})\geq 0).\label{eq:upperproblemma}
\end{multline}\label{eq:upperunsignedlemma1}
\end{lemma}
Let
\begin{equation}
\psi_{\a,\lambda^{(2)},\nu^{(1)}}=\epsilon_{3}^{(g)}\sqrt{n}\|\nu^{(1)}\|_2-\a^T(\z^{(1)}-2\lambda^{(2)})
-\|\nu^{(1)}\|_2r-2\sum_{i=n-k+1}^{n}\lambda_i^{(2)}\xtilde_i,\label{eq:upperdefpsi}
\end{equation}
with $\epsilon_{3}^{(g)}>0$ being an arbitrarily small constant independent of $n$. The left-hand side of the inequality in (\ref{eq:upperproblemma}) is then the following probability of interest
\begin{multline}
p_u=P(\min_{\lambda^{(2)}\in \Lambda,\nu^{(1)}\in R^n\setminus 0}\max_{\|\a\|_2=C_{\w_{up}}} ( \|\nu^{(1)}\|_2(\sum_{i=1}^{n}\h_i\a_i+\h_{n+1}\sigma)+\sqrt{C_{\w_{up}}^2+\sigma^2}\sum_{i=1}^{m}\g_i\nu_i^{(1)}\nonumber \\
-\epsilon_{3}^{(g)}\sqrt{n}\|\nu^{(1)}\|_2+\a^T(\z^{(1)}-2\lambda^{(2)})+\|\nu^{(1)}\|_2r+2\sum_{i=n-k+1}^{n}\lambda_i^{(2)}\xtilde_i ) \geq 0).
\end{multline}
After solving the inner maximization over $\a$ and pulling out $\|\nu\|_2$ one has
\begin{multline}
p_u=P(\min_{\lambda^{(2)}\in \Lambda,\nu\in R^n\setminus 0}(C_{\w_{up}}\|\h+\frac{1}{\|\nu^{(1)}\|_2}(\z^{(1)}-2\lambda^{(2)})\|_2+\h_{n+1}\sigma\\+\sqrt{C_{\w_{up}}^2+\sigma^2}\sum_{i=1}^{m}\g_i\frac{\nu_i^{(1)}}{\|\nu^{(1)}\|_2}-\epsilon_{3}^{(g)}\sqrt{n}+
r+2\sum_{i=n-k+1}^{n}\frac{\lambda_i^{(2)}}{\|\nu^{(1)}\|_2}\xtilde_i)\geq 0)\nonumber.
\end{multline}
After minimization of the second term over a unit norm vector we further have
\begin{multline}
p_u=P(\min_{\lambda^{(2)}\in \Lambda,\nu\in R^n\setminus 0}(C_{\w_{up}}\|\h+\frac{1}{\|\nu^{(1)}\|_2}(\z^{(1)}-2\lambda^{(2)})\|_2+\h_{n+1}\sigma\\
-\sqrt{C_{\w_{up}}^2+\sigma^2}\|\g\|-2-\epsilon_{3}^{(g)}\sqrt{n}+
r+2\sum_{i=n-k+1}^{n}\frac{\lambda_i^{(2)}}{\|\nu^{(1)}\|_2}\xtilde_i)\geq 0).\label{eq:upperprobanal1}
\end{multline}
Now we change variables so that $\nu=\frac{1}{\|\nu^{(1)}\|_2}$ and $\lambda^{(2)}=\frac{2\lambda^{(2)}}{\|\nu^{(1)}\|_2}$ and redefine $\Lambda$ by setting
\begin{equation}
\Lambda^{(2)}=\{\lambda^{(2)}\in R^n | 0\leq \lambda_i^{(2)}\leq 2\nu,1\leq i\leq n\}.\label{eq:upperlambda2}
\end{equation}
We also recall that $\z^{(1)}$ remains as defined right after (\ref{eq:Lagran7}). Plugging all of this back in
(\ref{eq:upperprobanal1}) gives us
\begin{equation}
p_u=P(r+\h_{n+1}\sigma-\epsilon_{3}^{(g)}\sqrt{n}-\max_{\lambda^{(2)}\in \Lambda^{(2)},\nu\geq 0}(\sqrt{C_{\w_{up}}^2+\sigma^2}\|\g\|_2-C_{\w_{up}}\|\h+\nu\z^{(1)}-\lambda^{(2)})\|_2
-\sum_{i=n-k+1}^{n}\lambda_i^{(2)}\xtilde_i)\geq 0).\label{eq:upperprobanal2}
\end{equation}
Now, let
\begin{equation}
\xi_{up}(\sigma,\g,\h,\xtilde,C_{\w_{up}})=\max_{\lambda^{(2)}\in \Lambda^{(2)},\nu\geq 0}(\sqrt{C_{\w_{up}}^2+\sigma^2}\|\g\|_2-C_{\w_{up}}\|\h+\nu\z^{(1)}-\lambda^{(2)})\|_2
-\sum_{i=n-k+1}^{n}\lambda_i^{(2)}\xtilde_i).\label{eq:upperdefxi}
\end{equation}
In the following lemma we will show that $\xi_{up}(\sigma,\g,\h,\xtilde,C_{\w_{up}})$ is a Lipschitz function.
\begin{lemma}
Let $\g$ and $\h$ be $m$ and $n$ dimensional vectors, respectively, with i.i.d. standard normal variables as their components. Let $\sigma>0$ be an arbitrary scalar. Let $\xi_{up}(\sigma,\g,\h,\xtilde,C_{\w_{up}})$ be as in (\ref{eq:upperdefxi}). Further let $\epsilon_{lip}>0$ be any constant. Then
\begin{equation}
\hspace{-.5in}P(|\xi_{up}(\sigma,\g,\h,\xtilde,C_{\w_{up}})-E\xi_{up}(\sigma,\g,\h,\xtilde,C_{\w_{up}})|\geq \epsilon_{lip}E\xi_{up}(\sigma,\g,\h,\xtilde,C_{\w_{up}}))\leq \exp \left \{  -\frac{(\epsilon_{lip} E\xi_{up}(\sigma,\g,\h,\xtilde,C_{\w_{up}}))^2}{2(2C_\w^2+\sigma^2)} \right \}.\label{eq:upperlipsch1}
\end{equation}\label{thm:upperlipsch1}
\end{lemma}
\begin{proof}The proof will be similar to the corresponding one from Subsection \ref{sec:unsignedconc}. We start by setting
\begin{equation}
f_{lip}(\g^{(1)},\h^{(1)})=\max_{\lambda^{(2)}\in \Lambda^{(2)},\nu\geq 0}(\sqrt{C_{\w_{up}}^2+\sigma^2}\|\g^{(1)}\|_2-C_{\w_{up}}\|\h^{(1)}+\nu\z^{(1)}-\lambda^{(2)})\|_2
-\sum_{i=n-k+1}^{n}\lambda_i^{(2)}\xtilde_i).\label{eq:upperlipproof1}
\end{equation}
Further, let $\nu^{(lip_1)}$ and $\lambda^{(lip_1)}$ be the solutions of the minimization in (\ref{eq:upperlipproof1}). Then, clearly
\begin{equation}
f_{lip}(\g^{(1)},\h^{(1)})= (\sqrt{C_{\w_{up}}^2+\sigma^2}\|\g^{(1)}\|_2-C_{\w_{up}}\|\h^{(1)}+\nu^{(lip_1)}\z^{(1)}-\lambda^{(lip_1)})\|_2
-\sum_{i=n-k+1}^{n}\lambda_i^{(lip_1)}\xtilde_i).\label{eq:upperlipproof2}
\end{equation}
In an analogous fashion set
\begin{equation}
f_{lip}(\g^{(2)},\h^{(2)})=\max_{\lambda^{(2)}\in \Lambda^{(2)},\nu\geq 0}(\sqrt{C_{\w_{up}}^2+\sigma^2}\|\g^{(2)}\|_2-C_{\w_{up}}\|\h^{(2)}+\nu\z^{(1)}-\lambda^{(2)})\|_2
-\sum_{i=n-k+1}^{n}\lambda_i^{(2)}\xtilde_i),\label{eq:upperlipproof3}
\end{equation}
and let $\nu^{(lip_2)}$ and $\lambda^{(lip_2)}$ be the solutions of the minimization in (\ref{eq:upperlipproof3}). Then, clearly
\begin{equation}
f_{lip}(\g^{(2)},\h^{(2)})= (\sqrt{C_{\w_{up}}^2+\sigma^2}\|\g^{(2)}\|_2-C_{\w_{up}}\|\h^{(2)}+\nu^{(lip_2)}\z^{(1)}-\lambda^{(lip_2)})\|_2
-\sum_{i=n-k+1}^{n}\lambda_i^{(lip_2)}\xtilde_i),\label{eq:upperlipproof4}
\end{equation}
Now assume that $f_{lip}(\g^{(1)},\h^{(1)})\neq f_{lip}(\g^{(2)},\h^{(2)})$ (if they are equal we are trivially done). Further let $f_{lip}(\g^{(1)},\h^{(1)})< f_{lip}(\g^{(2)},\h^{(2)})$ (the rest of the argument of course can trivially be flipped if $f_{lip}(\g^{(1)},\h^{(1)})> f_{lip}(\g^{(2)},\h^{(2)})$). We then have
\begin{multline}
|f_{lip}(\g^{(2)},\h^{(2)})- f_{lip}(\g^{(1)},\h^{(1)})|=f_{lip}(\g^{(2)},\h^{(2)})- f_{lip}(\g^{(1)},\h^{(1)})\\
= (\sqrt{C_{\w_{up}}^2+\sigma^2}\|\g^{(2)}\|_2-C_{\w_{up}}\|\h^{(2)}+\nu^{(lip_2)}\z^{(1)}-\lambda^{(lip_2)}\|_2
-\sum_{i=n-k+1}^{n}\lambda_i^{(lip_2)}\xtilde_i)\\
-(\sqrt{C_{\w_{up}}^2+\sigma^2}\|\g^{(1)}\|_2-C_{\w_{up}}\|\h^{(1)}+\nu^{(lip_1)}\z^{(1)}-\lambda^{(lip_1)}\|_2
-\sum_{i=n-k+1}^{n}\lambda_i^{(lip_1)}\xtilde_i)\\
\leq (\sqrt{C_{\w_{up}}^2+\sigma^2}\|\g^{(2)}\|_2-C_{\w_{up}}\|\h^{(2)}+\nu^{(lip_1)}\z^{(1)}-\lambda^{(lip_1)}\|_2
-\sum_{i=n-k+1}^{n}\lambda_i^{(lip_1)}\xtilde_i)\\
-(\sqrt{C_{\w_{up}}^2+\sigma^2}\|\g^{(1)}\|_2-C_{\w_{up}}\|\h^{(1)}+\nu^{(lip_1)}\z^{(1)}-\lambda^{(lip_1)}\|_2
-\sum_{i=n-k+1}^{n}\lambda_i^{(lip_1)}\xtilde_i)\\
= \sqrt{C_{\w_{up}}^2+\sigma^2}(\|\g^{(2)}\|_2-\|\g^{(1)}\|_2)-C_{\w_{up}}(\|\h^{(2)}+\nu^{(lip_1)}\z^{(1)}-\lambda^{(lip_1)}\|_2
-\|\h^{(2)}+\nu^{(lip_1)}\z^{(1)}-\lambda^{(lip_1)}\|_2)\\
\leq \sqrt{C_{\w_{up}}^2+\sigma^2}\|\g^{(2)}-\g^{(1)}\|_2+C_{\w_{up}}(\|\h^{(2)}-\h^{(1)}\|_2)\\
\leq \sqrt{2C_{\w_{up}}^2+\sigma^2}\sqrt{\|\g^{(2)}-\g^{(1)}\|_2^2+(\|\h^{(2)}-\h^{(1)}\|_2^2)},\label{eq:upperlipproof5}
\end{multline}
where the first inequality follows by sub-optimality of $\nu^{lip_1}$ and $\lambda^{(lip_1)}$ in (\ref{eq:upperlipproof3}). Connecting beginning and end in
(\ref{eq:upperlipproof5}) and combining it with (\ref{eq:upperlipproof1}) one then has that
$\xi_{up}(\sigma,\g,\h,\xtilde,C_{\w_{up}})$ is Lipschitz with $c_{lip}=\sqrt{2C_\w^2+\sigma^2}$. (\ref{eq:upperlipsch1}) then easily follows by Lemma \ref{thm:lipsch}.
\end{proof}

We continue by following the line of arguments right after (\ref{eq:probanalcont1}). As stated there $P(\h_{n+1}\sigma\geq -\epsilon_1^{(\h)}\sqrt{n})\geq 1-e^{-\epsilon_2^{(\h)}n}$ where $\epsilon_1^{(\h)}>0$ is an arbitrarily small constant and $\epsilon_2^{(\h)}$ is a constant dependent on $\epsilon_1^{(\h)}$ and $\sigma$ but independent on $n$. Set
\begin{equation}
r=\zeta_{obj}^{(u)}=(1+\epsilon_{lip})E\xi(\sigma,\g,\h,\xtilde,C_{\w_{up}})+\epsilon_1^{(\h)}\sqrt{n}+\epsilon_{3}^{(g)}\sqrt{n}.\label{eq:upperdefzetaobjl}
\end{equation}
One then has after combing (\ref{eq:upperprobanal2}) and Lemma \ref{thm:upperlipsch1}
\begin{multline}
p_u=P(r+\h_{n+1}\sigma-\epsilon_{3}^{(g)}\sqrt{n}-\max_{\lambda^{(2)}\in \Lambda^{(2)},\nu\geq 0}(\sqrt{C_{\w_{up}}^2+\sigma^2}\|\g\|_2-C_{\w_{up}}\|\h+\nu\z^{(1)}-\lambda^{(2)})\|_2
-\sum_{i=n-k+1}^{n}\lambda_i^{(2)}\xtilde_i)\geq 0)\\
 \geq  P((1+\epsilon_{lip})E\xi(\sigma,\g,\h,\xtilde,C_{\w_{up}})\\ \geq \max_{\lambda^{(2)}\in \Lambda^{(2)},\nu\geq 0}(\sqrt{C_{\w_{up}}^2+\sigma^2}\|\g\|_2-C_{\w_{up}}\|\h+\nu\z^{(1)}-\lambda^{(2)})\|_2
-\sum_{i=n-k+1}^{n}\lambda_i^{(2)}\xtilde_i))(1-e^{-\epsilon_2^{(\h)}n})\\
 \geq  \left ( 1-\exp \left \{ -\frac{(\epsilon_{lip} E\xi(\sigma,\g,\h,\xtilde,C_{\w_{up}}))^2}{2(2C_{\w_{up}}^2+\sigma^2)} \right \} \right )(1-e^{-\epsilon_2^{(\h)}n}).\label{eq:upperprobanalcont2}
\end{multline}
As stated after (\ref{eq:probint}), (\ref{eq:upperprobanalcont2}) is conceptually enough to establish a ``high probability" upper bound on $\zeta_{obj}$. What is left is to connect it with (\ref{eq:upperLagran14}). Combining (\ref{eq:upperprobanalcont2}),
(\ref{eq:upperproblemma}), and (\ref{eq:upperLagran14}) we then obtain
\begin{equation}
P(f_{obj}^{(up)}\geq 0)\geq \left ( 1-\exp \left \{ -\frac{(\epsilon_{lip} E\xi(\sigma,\g,\h,\xtilde,C_{\w_{up}}))^2}{2(2C_{\w_{up}}^2+\sigma^2)} \right \} \right )(1-e^{-\epsilon_2^{(\h)}n})(1-e^{-\epsilon_4^{(g)}n}),\label{eq:upperprobanalcont3}
\end{equation}
where we used the fact that $g$ is the standard normal and therefore $P(g-\epsilon_3^{(g)}\sqrt{n}\leq 0)\geq (1-e^{-\epsilon_4^{(g)}n})$ for an arbitrarily small $\epsilon_3^{(g)}>0$ and a constant $\epsilon_4^{(g)}$ dependent on $\epsilon_3^{(g)}$ but independent of $n$.

We are now in position to summarize results from this subsection in the following lemma which is essentially an ``upper-bound" analogue to Lemma \ref{thm:lowerbound}.
\begin{lemma}
Let $\v$ be an $n\times 1$ vector of i.i.d. zero-mean variance $\sigma^2$ Gaussian random variables and let $A$ be an $m\times n$ matrix of i.i.d. standard normal random variables. Consider an $\xtilde$ defined in (\ref{eq:xtildedef}) and a $\y$ defined in (\ref{eq:systemnoise}) for $\x=\xtilde$. Let then $\zeta_{obj}$ be as defined in (\ref{eq:objlassol1}) and let $\w$ be the solution of (\ref{eq:objlassol13}). There is a constant $\epsilon_{upper}>0$
\begin{equation}
P(\zeta_{obj}\leq \zeta_{obj}^{(upper)})\geq 1-e^{-\epsilon_{upper}n},\label{eq:lowerboundobjthm1}
\end{equation}
where
\begin{equation}
\zeta_{obj}^{(upper)}=(1+\epsilon_{lip})E\xi_{up}(\sigma,\g,\h,\xtilde,C_{\w_{up}})+\epsilon_1^{(\h)}\sqrt{n}+\epsilon_3^{(g)}\sqrt{n},\label{eq:upperboundobjthm2}
\end{equation}
$\xi_{up}(\sigma,\g,\h,\xtilde,C_{\w_{up}})$ is as defined in (\ref{eq:upperdefxi}), $\epsilon_{lip},\epsilon_1^{(\h)},\epsilon_3^{(g)}$ are all positive arbitrarily small constants, and $C_{\w_{up}}$ is a constant such that $\|\w\|_2\leq C_{\w_{up}}$.
\label{thm:upperbound}
\end{lemma}
\begin{proof}
Follows from the previous discussion.
\end{proof}

\subsection{Matching upper and lower bounds}\label{sec:matching}

In this section we specialize the general bounds introduced above and show how they can match each other. We will divide presentation in three subsections. In the first of the subsections we will make a connection to the noiseless case and show how one can then remove the constraint from (\ref{eq:defhatxi}), (\ref{eq:defhatw}), and (\ref{eq:defhatwnorm}). In the second subsection we will consider a $\w$ such that $|\|\w\|_2-\|\hat{\w}\|_2|\geq \epsilon_{\w_{up}}\|\hat{\w}\|_2$. We will then quantify how much the lower bound that can be computed for such a $\w$ through the framework presented in Section \ref{sec:unsignedlbzetaobj} deviates from the optimal one obtained for $\hat{\w}$. In the last subsection we will then show that there will be a $\w$ such that the upper bound computed through the framework presented in Section \ref{sec:unsignedubzetaobj} will deviate less. That will in essence establish that upper and lower bounds computed in the previous sections indeed match. We will then draw conclusions as for the consequences which such a matching of the bounds leaves on a couple of LASSO parameters.

\subsubsection{Connection to the $\ell_1$ optimization}\label{sec:connectl1}

In this subsection we establish a connection between the constraint in (\ref{eq:defhatxi}), (\ref{eq:defhatw}), and (\ref{eq:defhatwnorm}) and the fundamental performance characterization of $\ell_1$ optimization derived in \cite{StojnicUpper10} (and of course earlier in the context of neighborly polytopes in \cite{DonohoPol}). We first recall on the condition from Lemma \ref{thm:optsollower}. The condition states
\begin{equation}
\sqrt{1+\frac{\sigma^2}{C_\w^2}}\|\h+\hat{\nu}\z^{(1)}-\widehat{\lambda^{(2)}}\|_2\leq \|\g\|_2,\label{eq:condoptsollower}
\end{equation}
where $C_\w$ is an arbitrarily large constant and $\hat{\nu}$ and $\widehat{\lambda^{(2)}}$ are the solutions of
\begin{eqnarray}
\max & & \sigma\sqrt{\|\g\|_2^2-\|\h+\nu\z^{(1)}-\lambda^{(2)}\|_2^2} -\sum_{i=n-k+1}^{n}\lambda_i^{(2)}\xtilde_i\nonumber \\
\mbox{subject to} & & 0\leq \lambda_i^{(2)}\leq 2\nu,1\leq i\leq n\nonumber \\
& & \nu\geq 0.\label{eq:matchopt}
\end{eqnarray}
Now we note the following equivalent to (\ref{eq:matchopt}) for the case when nonzero components of $\xtilde$ are infinite
\begin{eqnarray}
\max & & \sigma\sqrt{\|\g\|_2^2-\|\h+\nu\z^{(1)}-\lambda^{(2)}\|_2^2} \nonumber \\
\mbox{subject to} & & 0\leq \lambda_i^{(2)}\leq 2\nu,1\leq i\leq n-k\nonumber \\
 & & \lambda_i^{(2)}=0,n-k+1\leq i\leq n\nonumber \\
& & \nu\geq 0.\label{eq:matchl1}
\end{eqnarray}
Now, to make the new observations easily comparable to the corresponding ones from \cite{StojnicCSetam09,StojnicEquiv10} we set
\begin{equation}
\bar{\h}=[|\h|_{(1)}^{(1)},|\h|_{(2)}^{(2)},\dots,|\h|_{(n-k)}^{(n-k)},\h_{n-k+1},\h_{n-k+2},\dots,\h_n]^T,\label{eq:defhbar}
\end{equation}
where $[|\h|_{(1)}^{(1)},|\h|_{(2)}^{(2)},\dots,|\h|_{(n-k)}^{(n-k)}]$ are magnitudes of $[\h_{1},\h_{2},\dots,\h_{n-k}]$ sorted in increasing order (possible ties in the sorting process are of course broken arbitrarily). Also we let $\z^{(2)}$ be such that $\z_i^{(2)}=-\z_i^{(1)},n-k+1\leq i\leq n$ and $\z_i^{(2)}=\z_i^{(1)},1\leq i\leq n-k$. It is then relatively easy to see that the above optimization problem is equivalent to
\begin{eqnarray}
\max & & \sigma\sqrt{\|\g\|_2^2-\|\bar{\h}-\nu\z^{(2)}+\lambda^{(2)}\|_2^2} \nonumber \\
\mbox{subject to}
& & 0 \leq\lambda_i^{(2)}\leq \nu, 1\leq i\leq n-k\nonumber \\
& & \lambda_i^{(2)}=0,n-k+1\leq i\leq n\nonumber \\
& & \nu\geq 0.
\label{eq:matchl11}
\end{eqnarray}
Let $\nu_{\ell_1}$ and $\lambda^{(\ell_1)}$ be the solution of the above maximization. Then, as we showed in \cite{StojnicCSetam09} and \cite{StojnicUpper10}, the inequality
\begin{equation}
E\|\g\|_2> E\|\bar{\h}-\nu_{\ell_1}\z^{(2)}+\lambda^{(\ell_1)}\|_2\label{eq:fundl1exp}
\end{equation}
establishes the following fundamental performance characterization of the $\ell_1$ optimization algorithm from (\ref{eq:l1}) that could be used instead of LASSO to recover $\x$ in (\ref{eq:system}) (which is a noiseless version of (\ref{eq:systemnoise}))
\begin{equation}
(1-\beta_w)\frac{\sqrt{\frac{2}{\pi}}e^{-(\erfinv(\frac{1-\alpha_w}{1-\beta_w}))^2}}{\alpha_w}-\sqrt{2}\erfinv (\frac{1-\alpha_w}{1-\beta_w})=0,
\label{eq:fundl1}
\end{equation}
where of course $\alpha_w=\frac{m}{n}$ and $\beta_w=\frac{k}{n}$. As it is also shown in \cite{StojnicCSetam09} and \cite{StojnicUpper10} both of the quantities under the expected values in (\ref{eq:fundl1exp}) nicely concentrate. Then with overwhelming probability one has that for any pair $(\alpha,\beta)$ that satisfies (or lies below) the above fundamental performance characterization of $\ell_1$ optimization
\begin{equation}
\|\g\|_2> \|\bar{\h}-\nu_{\ell_1}\z^{(2)}+\lambda^{(\ell_1)}\|_2.\label{eq:fundl1noexp}
\end{equation}
Moreover, since $\lambda_i^{(2)}\geq 0, n-k+1\leq i\leq n$, in (\ref{eq:matchopt}) one actually has that (\ref{eq:fundl1noexp}) implies that with overwhelming probability
\begin{equation}
\|\g\|_2> \|\h+\hat{\nu}\z^{(1)}-\widehat{\lambda^{(2)}}\|_2,
\end{equation}
which for sufficiently large $C_\w$ is the same as (\ref{eq:condoptsollower}).  We then in what follows assume that pair $(\alpha,\beta)$ is such that it satisfies the fundamental $\ell_1$ optimization performance characterization (or is in the region below it) and therefore proceed by ignoring the condition (\ref{eq:condoptsollower}). (Strictly speaking, all our overwhelming probabilities below should be multiplied with an overwhelming probability that (\ref{eq:fundl1}) holds; to maintain writing easier we will skip this detail.)

\subsubsection{Deviation from the lower-bound}\label{sec:devlb}

In this subsection we show that $\|\w_{lasso}\|_2$ can not deviate substantially from $\|\hat{\w}\|_2$ without substantially affecting the value of the lower bound on the objective in (\ref{eq:lassol1}) that is derived in Section \ref{sec:unsignedlbzetaobj}. To that end let us assume that there is a $\w_{off}$ that is the solution of the LASSO from (\ref{eq:lassol1}) (or to be slightly more precise that is such that $\hat{\x}=\xtilde+\w_{off}$, where obviously $\hat{\x}$ is the solution of (\ref{eq:lassol1})). Further, let $|\|\w_{off}\|_2-\|\hat{\w}\|_2|\geq \epsilon_{\w_{up}}\|\hat{\w}\|_2$, where $\epsilon_{\w_{up}}$ is an arbitrarily small constant.

One can then proceed by repeating the same line of thought as in Section \ref{sec:unsignedlbzetaobj}. The only difference will be that now $C_\w=\|\w_{off}\|_2$ and consequently in the definition of $S_\w(\sigma,\xtilde,C_\w)$, $\|\w\|_2\leq C_\w$ changes to $\|\w\|_2=C_\w=\|\w_{off}\|_2$. This difference will of course not affect the concept presented in Section \ref{sec:unsignedlbzetaobj}. The only real consequence will be the change of (\ref{eq:defxi}). Adapted to the new scenario (\ref{eq:defxi}) becomes
\begin{eqnarray}
\xi_{off}(\sigma,\g,\h,\xtilde,\w_{off})=\min_{\w,\t} & & \sqrt{\|\w_{off}\|_2^2+\sigma^2}\|\g\|_2+\sum_{i=1}^{n}\h_i\w_i\nonumber \\
\mbox{subject to} & & \sum_{i=1}^n \t_i\leq \|\xtilde\|_1\nonumber \\
& & \xtilde_i+\w_i-\t_i \leq 0, n-k+1\leq i\leq n\nonumber \\
& & -\xtilde_i-\w_i-\t_i\leq 0, n-k+1\leq i\leq n\nonumber \\
& & \w_i-\t_i\leq 0, 1\leq i\leq n-k\nonumber \\
& & -\w_i-\t_i\leq 0, 1\leq i\leq n-k\nonumber \\
& & \sqrt{\|\w\|_2^2+\sigma^2}\leq \sqrt{\w_{off}^2+\sigma^2}.\label{eq:matchdefxi4}
\end{eqnarray}
One can then proceed further with solving the Lagrangian to obtain
\begin{equation}
\xi_{off}(\sigma,\g,\h,\xtilde,\w_{off})=\max_{\lambda^{(2)}\in \Lambda^{(2)},\nu\geq 0}(\sqrt{\w_{off}^2+\sigma^2}\|\g\|_2-\w_{off}\|\h+\nu\z^{(1)}-\lambda^{(2)})\|_2
-\sum_{i=n-k+1}^{n}\lambda_i^{(2)}\xtilde_i).\label{eq:matchdefxi}
\end{equation}
Using the probabilistic arguments from Section \ref{sec:unsignedlbzetaobj} one then from Lemma \ref{thm:lowerbound} has that if $\w_{off}$ is the solution of (\ref{eq:lassol1}) then its objective value with overwhelming probability is lower bounded by $(1-\epsilon_{lip})E\xi_{off}(\sigma,\g,\h,\xtilde,\w_{off})$
($\xi_{off}(\sigma,\g,\h,\xtilde,\w_{off})$ is structurally the same as $\xi_{up}(\sigma,\g,\h,\xtilde,C_{\w_{up}})$ from (\ref{eq:upperdefxi}) and therefore easily concentrates based on Lemma \ref{thm:upperlipsch1}).
We will now consider in parallel the following lower bound from (\ref{eq:Lagran12}) (clearly, choosing $\w_{off}=\hat{\w}$ would make (\ref{eq:matchdefxi}) equivalent to (\ref{eq:Lagran12})).
\begin{equation}
\xi_{ov}(\sigma,\g,\h,\xtilde)=\max_{\nu\geq 0,\lambda^{(2)}\in \Lambda^{(2)}} \sigma\sqrt{\|\g\|_2^2-\|\h+\nu\z^{(1)}-\lambda^{(2)}\|_2^2} -\sum_{i=n-k+1}^{n}\lambda_i^{(2)}\xtilde_i.\label{eq:matchoptlower}
\end{equation}
Now, let as usual $\hat{\nu}$ and $\widehat{\lambda^{(2)}}$ be the solutions of (\ref{eq:matchoptlower}). Let
\begin{equation}
\xi_{help}(\sigma,\g,\h,\xtilde,\w_{off})=\sqrt{\w_{off}^2+\sigma^2}\|\g\|_2-\w_{off}\|\h+\hat{\nu}\z^{(1)}-\widehat{\lambda^{(2)}}\|_2
-\sum_{i=n-k+1}^{n}\widehat{\lambda_i^{(2)}}\xtilde_i.\label{eq:matchdefxihelp}
\end{equation}
Then
\begin{multline}
\xi_{off}(\sigma,\g,\h,\xtilde,\w_{off})-\xi_{ov}(\sigma,\g,\h,\xtilde)\geq \xi_{help}(\sigma,\g,\h,\xtilde,\w_{off})
-\xi_{ov}(\sigma,\g,\h,\xtilde)\\=\sqrt{\w_{off}^2+\sigma^2}\|\g\|_2-\w_{off}\|\h+\hat{\nu}\z^{(1)}-\widehat{\lambda^{(2)}}\|_2-
\sigma\sqrt{\|\g\|_2^2-\|\h+\hat{\nu}\z^{(1)}-\widehat{\lambda^{(2)}}\|_2^2}.\label{eq:matchdiff1}
\end{multline}
For the simplicity let $|\|\w_{off}\|_2-\|\hat{\w}\|_2|= \epsilon_{\w_{up}}\|\hat{\w}\|_2$ (this restriction is clearly more conservative than $|\|\w_{off}\|_2-\|\hat{\w}\|_2|\geq \epsilon_{\w_{up}}\|\hat{\w}\|_2$). Now, we switch to expectations and ignore all $\epsilon$ except $\epsilon_{\w_{up}}$. Since every quantity that we will consider (see (\ref{eq:conchw})) concentrates $\epsilon$'s in concentration inequalities can be made arbitrarily close to zero; moreover once $\epsilon_{\w_{up}}$ is fixed all other $\epsilon$'s can be made arbitrarily small compared to $\epsilon_{\w_{up}}$. Also, we will show derivation  for $\w_{off}=(1+\epsilon_{\w_{up}})\|\hat{\w}\|_2$ (the derivation for the case $\w_{off}=(1-\epsilon_{\w_{up}})\|\hat{\w}\|_2$ is completely analogous).

Now, to facilitate writing we then set all $\epsilon$'s except $\epsilon_{\w_{up}}$ to zero. We then have
\begin{multline}
E\xi_{off}(\sigma,\g,\h,\xtilde,\w_{off})-E\xi_{ov}(\sigma,\g,\h,\xtilde)\geq E\xi_{help}(\sigma,\g,\h,\xtilde,\w_{off})
-E\xi_{ov}(\sigma,\g,\h,\xtilde)\\\doteq \sqrt{(1+\epsilon_{\w_{up}})^2(E\|\hat{\w}\|_2)^2+\sigma^2}E\|\g\|_2-(1+\epsilon_{\w_{up}})E\|\hat{\w}\|_2E\|\h+\hat{\nu}\z^{(1)}
-\widehat{\lambda^{(2)}}\|_2\\-\sigma\sqrt{(E\|\g\|_2)^2-(E\|\h+\hat{\nu}\z^{(1)}
-\widehat{\lambda^{(2)}}\|_2)^2}\label{eq:matchdiff2}
\end{multline}
where $\doteq$ means that equality is not exact but for a fixed $\epsilon_{\w_{up}}$ can be made as close to it as needed. In a similar fashion we have
\begin{equation}
E\|\hat{\w}\|_2\doteq\frac{\sigma E\|\h+\hat{\nu}\z^{(1)}-\widehat{\lambda^{(2)}}\|_2}{\sqrt{(E\|\g\|_2)^2-(E\|\h+\hat{\nu}\z^{(1)}-\widehat{\lambda^{(2)}}\|_2)^2}}.\label{eq:matchdiff3}
\end{equation}
Before we proceed further we simplify the notation with the following change of variables.
\begin{eqnarray}
g_E & = & E\|\g\|_2\nonumber \\
h_E & = & E\|\h+\hat{\nu}\z^{(1)}-\widehat{\lambda^{(2)}}\|_2\nonumber \\
\xi_E & = & \sigma\sqrt{(E\|\g\|_2)^2-(E\|\h+\hat{\nu}\z^{(1)}
-\widehat{\lambda^{(2)}}\|_2)^2}=\sigma\sqrt{g_E^2-h_E^2}\nonumber \\
w_E & =  & E\|\hat{\w}\|_2=\frac{\sigma h_E}{\sqrt{g_E^2-h_E^2}}.\label{eq:defgEhExiE}
\end{eqnarray}
From (\ref{eq:defgEhExiE}) one easily has
\begin{eqnarray}
h_E^2 & = & g_E^2-\frac{\xi_E^2}{\sigma^2}\nonumber \\
w_E & = & \frac{h_E\sigma^2}{\xi_E}.\label{eq:defgEhExiE1}
\end{eqnarray}
Then a combination of (\ref{eq:matchdiff2}), (\ref{eq:defgEhExiE}), and (\ref{eq:defgEhExiE1}) gives
\begin{multline}
E\xi_{help}(\sigma,\g,\h,\xtilde,\w_{off})-E\xi_{ov}(\sigma,\g,\h,\xtilde) \doteq \sqrt{(1+\epsilon_{\w_{up}})^2(\frac{h_E^2\sigma^4}{\xi_E^2})+\sigma^2}g_E-(1+\epsilon_{\w_{up}})w_E h_E-\xi_E \\
=(1+\epsilon_{\w_{up}})\frac{g_E^2\sigma^2}{\xi_E}\sqrt{1-\frac{\xi_E^2(2\epsilon_{\w_{up}}+\epsilon_{\w_{up}}^2)}{(1+\epsilon_{\w_{up}})^2g_E^2\sigma^2}}
-(1+\epsilon_{\w_{up}})\frac{g_E^2\sigma^2}{\xi_E}+\epsilon_{\w_{up}}\xi_E.\label{eq:matchdiff4}
\end{multline}
Now, assuming that $\epsilon_{\w_{up}}$ is small (and recognizing that $\xi_E\leq g_E\sigma$) from (\ref{eq:matchdiff4}) we have
\begin{multline}
(1+\epsilon_{\w_{up}})\frac{g_E^2\sigma^2}{\xi_E}\sqrt{1-\frac{\xi_E^2(2\epsilon_{\w_{up}}+\epsilon_{\w_{up}}^2)}{(1+\epsilon_{\w_{up}})^2g_E^2\sigma^2}}
-(1+\epsilon_{\w_{up}})\frac{g_E^2\sigma^2}{\xi_E}+\epsilon_{\w_{up}}\xi_E\\\approx
(1+\epsilon_{\w_{up}})\frac{g_E^2\sigma^2}{\xi_E}(1-\frac{\xi_E^2(2\epsilon_{\w_{up}}+\epsilon_{\w_{up}}^2)}{2(1+\epsilon_{\w_{up}})^2g_E^2\sigma^2})
-(1+\epsilon_{\w_{up}})\frac{g_E^2\sigma^2}{\xi_E}+\epsilon_{\w_{up}}\xi_E\\
=-\frac{\xi_E(2\epsilon_{\w_{up}}+\epsilon_{\w_{up}}^2)}{2(1+\epsilon_{\w_{up}})}+\epsilon_{\w_{up}}\xi_E=
\frac{2\epsilon_{\w_{up}}\xi_E(1+\epsilon_{\w_{up}})-\xi_E(2\epsilon_{\w_{up}}+\epsilon_{\w_{up}}^2)}{2(1+\epsilon_{\w_{up}})}=
\frac{\xi_E\epsilon_{\w_{up}}^2}{2(1+\epsilon_{\w_{up}})}.\label{eq:matchdiff5}
\end{multline}
Combining (\ref{eq:matchdiff1}), (\ref{eq:matchdiff4}), and (\ref{eq:matchdiff5}) we finally have
\begin{equation}
E\xi_{off}(\sigma,\g,\h,\xtilde,\w_{off})-E\xi_{ov}(\sigma,\g,\h,\xtilde)\geq \frac{\epsilon_{\w_{up}}^2}{2(1+\epsilon_{\w_{up}})}E\xi_E
\geq \frac{\epsilon_{\w_{up}}^2}{2(1+\epsilon_{\w_{up}})}E\xi_{ov}(\sigma,\g,\h,\xtilde)\label{eq:matchdiff6}
\end{equation}
where the last inequality follows by noting that in the definition of $\xi_{ov}(\sigma,\g,\h,\xtilde)$ the elements of $\xtilde$ and $\lambda^{(2)}$ are non-negative.

Now, roughly speaking, (\ref{eq:matchdiff6}) shows that if $\|\w_{lasso}\|_2$ were to deviate from $\|\hat{\w}\|_2$ the optimal value of the objective in (\ref{eq:lassol1}) would be higher than the lower bound derived in Section \ref{sec:unsignedlbzetaobj}. We summarize these observations in the following lemma (essentially a deviating equivalent of Lemma \ref{thm:lowerbound} from Section \ref{sec:unsignedlbzetaobj}).
\begin{lemma}
Let $\v$ be an $n\times 1$ vector of i.i.d. zero-mean variance $\sigma^2$ Gaussian random variables and let $A$ be an $m\times n$ matrix of i.i.d. standard normal random variables. Consider an $\xtilde$ defined in (\ref{eq:xtildedef}) and a $\y$ defined in (\ref{eq:systemnoise}) for $\x=\xtilde$. Let then $\zeta_{obj}$ be as defined in (\ref{eq:objlassol1}) or (\ref{eq:objlassol13}) and let $\w_{off}$ be the solution of (\ref{eq:objlassol13}). Let $\alpha$ and $\beta$ be below the fundamental characterization (\ref{eq:fundl1}) and let $\hat{\w}$ be as defined in (\ref{eq:defhatw}). Assume that $|\|\w_{off}\|_2-\|\hat{\w}\|_2|\geq \epsilon_{\w_{up}}\|\hat{\w}\|_2$, where $\epsilon_{\w_{up}}$ is an arbitrarily small but fixed constant. Then there would be a constant $\epsilon_{off}>0$, and arbitrarily small positive constants $\epsilon_{lip},\epsilon_1^{(\h)},\epsilon_1^{(g)}$ such that
\begin{equation}
P(\zeta_{obj}\geq \zeta_{obj}^{(off)})\geq 1-e^{-\epsilon_{off}n},\label{eq:matchlowerboundobjthm1}
\end{equation}
where
\begin{equation}
\zeta_{obj}^{(off)}=
(1-\epsilon_{lip})(1+\frac{\epsilon_{\w_{up}}^2}{2(1+\epsilon_{\w_{up}})})E\xi_{ov}
(\sigma,\g,\h,\xtilde)-\epsilon_1^{(\h)}\sqrt{n}-\epsilon_1^{(g)}\sqrt{n},\label{eq:matchlowerboundobjthm2}
\end{equation}
and $\xi_{ov}(\sigma,\g,\h,\xtilde)$ is as defined in (\ref{eq:Lagran12}) (or (\ref{eq:matchoptlower})).
\label{thm:matchlowerbound}
\end{lemma}
\begin{proof}
Follows from the previous discussion, discussion from Section \ref{sec:connectl1}, and a combination of (\ref{eq:matchdefxi}), (\ref{eq:matchdiff1}), (\ref{eq:matchdiff6}), arguments right after (\ref{eq:matchdefxi}), and Lemma \ref{thm:lowerbound}.
\end{proof}

\subsubsection{Deviation of the upper bound}\label{sec:devub}

In this section we will show that $\|\w_{lasso}\|_2$ can not deviate from $\|\hat{\w}\|_2$ as much as it was assumed in the previous section. To do so we will actually continue to assume that it can and then eventually reach a contradiction. As in previous section, let then $|\|\w_{off}\|_2-\|\hat{\w}\|_2|\geq \epsilon_{\w_{up}}\|\hat{\w}\|_2$, where $\epsilon_{\w_{up}}$ is an arbitrarily small constant. Further, let $\xi_{dual}(\sigma,\g,\h,\xtilde)$ be
\begin{eqnarray}
\xi_{dual}(\sigma,\g,\h,\xtilde)=\min_{d\geq 0}\max_{\nu,\lambda^{(2)}} & & \sqrt{d^2+\sigma^2}\|\g\|_2-d\|\h+\nu\z^{(1)}-\lambda^{(2)}\|_2 -\sum_{i=n-k+1}^{n}\lambda_i^{(2)}\xtilde_i\nonumber \\
\mbox{subject to}
& & \nu\geq 0\nonumber \\
& & 0 \leq\lambda_i^{(2)}\leq 2\nu,1\leq i\leq n.\label{eq:devubLagran11}
\end{eqnarray}
Rewriting (\ref{eq:devubLagran11}) with a simple sign flipping turns out to be useful in what follows
\begin{eqnarray}
-\xi_{dual}(\sigma,\g,\h,\xtilde)=\max_{d\geq 0}\min_{\nu,\lambda^{(2)}} & & -\sqrt{d^2+\sigma^2}\|\g\|_2+d\|\h+\nu\z^{(1)}-\lambda^{(2)}\|_2 +\sum_{i=n-k+1}^{n}\lambda_i^{(2)}\xtilde_i\nonumber \\
\mbox{subject to}
& & \nu\geq 0\nonumber \\
& & 0 \leq\lambda_i^{(2)}\leq 2\nu,1\leq i\leq n.\label{eq:devubLagran12}
\end{eqnarray}
The following lemma provides a powerful tool to deal with (\ref{eq:devubLagran12}).
\begin{lemma}
Let $\xi_{dual}(\sigma,\g,\h,\xtilde)$ be as defined in (\ref{eq:devubLagran12}). Further, let
\begin{eqnarray}
-\xi_{ov}(\sigma,\g,\h,\xtilde)=\min_{\nu,\lambda^{(2)}}\max_{d\geq 0} & & -\sqrt{d^2+\sigma^2}\|\g\|_2+d\|\h+\nu\z^{(1)}-\lambda^{(2)}\|_2 +\sum_{i=n-k+1}^{n}\lambda_i^{(2)}\xtilde_i\nonumber \\
\mbox{subject to}
& & \nu\geq 0\nonumber \\
& & 0 \leq\lambda_i^{(2)}\leq 2\nu,1\leq i\leq n.\label{eq:devublemmadet1}
\end{eqnarray}
Then
\begin{equation}
\xi_{dual}(\sigma,\g,\h,\xtilde)=\xi_{ov}(\sigma,\g,\h,\xtilde).\label{eq:devublemmadet2}
\end{equation}
\label{thm:devublemmadet}
\end{lemma}
\begin{proof}
After solving the inner maximization over $d$ in (\ref{eq:devublemmadet1}) one has
\begin{equation}
d_{opt}=\sigma\frac{\|\h+\nu\z^{(1)}-\lambda^{(2)}\|_2}{\sqrt{\|\g\|_2^2-\|\h+\nu\z^{(1)}-\lambda^{(2)}\|_2^2}}.\label{eq:devubdopt}
\end{equation}
Such a $d$ then establishes that the right-hand side of (\ref{eq:devublemmadet1}) is indeed $\xi_{ov}(\sigma,\g,\h,\xtilde)$, i.e, one has as in (\ref{eq:Lagran12})
\begin{eqnarray}
-\xi_{ov}(\sigma,\g,\h,\xtilde)=\min_{\nu,\lambda^{(2)}} & & -\sigma\sqrt{\|\g\|_2^2-\|\h+\nu\z^{(1)}-\lambda^{(2)}\|_2^2} +\sum_{i=n-k+1}^{n}\lambda_i^{(2)}\xtilde_i\nonumber \\
\mbox{subject to}
& & \nu\geq 0\nonumber \\
& & 0 \leq\lambda_i^{(2)}\leq 2\nu,1\leq i\leq n.\label{eq:devublemmadet3}
\end{eqnarray}
Now we digress for a moment and consider the following optimization problem
\begin{eqnarray}
\min_{\nu,\lambda^{(2)},\q_1,\q_2} & & -\sigma\q_1 +\sum_{i=n-k+1}^{n}\lambda_i^{(2)}\xtilde_i\nonumber \\
\mbox{subject to}
& & \|\h+\nu\z^{(1)}-\lambda^{(2)}\|_2\leq \q_2\nonumber \\
& & \q_1^2+\q_2^2\leq\|\g\|_2^2\nonumber \\
& & \nu\geq 0\nonumber \\
& & 0 \leq\lambda_i^{(2)}\leq 2\nu,1\leq i\leq n.\label{eq:devubprimaldet}
\end{eqnarray}
Let $-\xi_{ov}^{(1)}(\sigma,\g,\h,\xtilde)$ be the optimal value of its objective function.
Let quadruplet $\hat{\nu},\widehat{\lambda^{(2)}},\hat{\q_1},\hat{\q_2}$ be the solution of the above optimization problem. Then it must be
\begin{equation}
\|\h+\hat{\nu}\z^{(1)}-\widehat{\lambda^{(2)}}\|_2=\hat{\q_2} \label{eq:devuboptq2det}
\end{equation}
and consequently
\begin{eqnarray}
\hat{\q_1}&=&\sqrt{\|\g\|_2^2-\|\h+\hat{\nu}\z^{(1)}-\widehat{\lambda^{(2)}}\|_2^2}\nonumber \\
-\xi_{ov}^{(1)}(\sigma,\g,\h,\xtilde)& =&-\sigma\sqrt{\|\g\|_2^2-\|\h+\hat{\nu}\z^{(1)}-\widehat{\lambda^{(2)}}\|_2^2}+\sum_{i=n-k+1}^{n}\widehat{\lambda_i^{(2)}}\xtilde_i.\label{eq:devuboptprimaldet}
\end{eqnarray}
The above claim is rather obvious but for the completeness we sketch the argument that supports it. Assume that $\|\h+\hat{\nu}\z^{(1)}-\widehat{\lambda^{(2)}}\|_2<\hat{\q_2}$, then $\hat{\q_1}<\sqrt{\|\g\|_2^2-\|\h+\hat{\nu}\z^{(1)}-\widehat{\lambda^{(2)}}\|_2^2}$, and
$-\xi_{ov}^{(1)}(\sigma,\g,\h)$ would be smaller then the expression on the right-hand side of (\ref{eq:devuboptprimaldet}). Now, since (\ref{eq:devuboptq2det})
and (\ref{eq:devuboptprimaldet}) hold one has that $-\xi_{ov}^{(1)}(\sigma,\g,\h,\xtilde)$ can be determined through the following equivalent to (\ref{eq:devubprimaldet})
\begin{eqnarray}
-\xi_{ov}^{(1)}(\sigma,\g,\h,\xtilde)=\min_{\nu,\lambda^{(2)}} & & -\sigma\sqrt{\|\g\|_2^2-\|\h+\nu\z^{(1)}-\lambda^{(2)}\|_2^2}+\sum_{i=n-k+1}^{n}\lambda_i^{(2)}\xtilde_i\nonumber \\
\mbox{subject to}
& & \nu\geq 0\nonumber \\
& & 0 \leq\lambda_i^{(2)}\leq 2\nu,1\leq i\leq n\label{eq:devubprimaldet1}
\end{eqnarray}
After comparing (\ref{eq:devublemmadet3}) and (\ref{eq:devubprimaldet1}) we have
\begin{equation}
-\xi_{ov}^{(1)}(\sigma,\g,\h,\xtilde)=-\xi_{ov}(\sigma,\g,\h,\xtilde).\label{eq:devubdeteq1}
\end{equation}
Now, let us write the Lagrange dual of the optimization problem in (\ref{eq:devubprimaldet}). Let $d$ and $\gamma_1$ be Lagrangian variables such that
\begin{eqnarray}
\max_{d\geq 0,\gamma_1\geq 0}\min_{\nu,\lambda^{(2)},\q_1,\q_2} & & -\sigma\q_1 +\sum_{i=n-k+1}^{n}\lambda_i^{(2)}\xtilde_i+d\|\h+\nu\z^{(1)}-\lambda^{(2)}\|_2-d\q_2
+\gamma_1(\q_1^2+\q_2^2)-\gamma_1\|\g\|_2^2\nonumber \\
\mbox{subject to}
& & \nu\geq 0\nonumber \\
& & 0 \leq\lambda_i^{(2)}\leq 2\nu,1\leq i\leq n.\label{eq:devubprimaldet2}
\end{eqnarray}
After solving the inner minimization over $\q_1,\q_2$ in (\ref{eq:devubprimaldet2}) we have
\begin{eqnarray}
\max_{d\geq 0,\gamma_1\geq 0}\min_{\nu,\lambda^{(2)}} & & -\frac{\sigma^2+d^2}{4\gamma_1}-\gamma_1\|\g\|_2^2 +\sum_{i=n-k+1}^{n}\lambda_i^{(2)}\xtilde_i+d\|\h+\nu\z^{(1)}-\lambda^{(2)}\|_2\nonumber \\
\mbox{subject to}
& & \nu\geq 0\nonumber \\
& & 0 \leq\lambda_i^{(2)}\leq 2\nu,1\leq i\leq n.\label{eq:devubprimaldet3}
\end{eqnarray}
Since the first two terms in the objective function in (\ref{eq:devubprimaldet3}) do not involve neither $\nu$ nor $\lambda^{(2)}$ one can then maximize their sum over $\gamma_1$ for any $d$. After that we finally have
\begin{eqnarray}
\max_{d\geq 0}\min_{\nu,\lambda^{(2)}} & & -\sqrt{\sigma^2+d^2}\|\g\|_2 +\sum_{i=n-k+1}^{n}\lambda_i^{(2)}\xtilde_i+d\|\h+\nu\z^{(1)}-\lambda^{(2)}\|_2\nonumber \\
\mbox{subject to}
& & \nu\geq 0\nonumber \\
& & 0 \leq\lambda_i^{(2)}\leq 2\nu,1\leq i\leq n.\label{eq:devubprimaldet4}
\end{eqnarray}
Let $-\xi_{ov}^{(2)}(\sigma,\g,\h,\xtilde)$ be the optimal value of the objective function in (\ref{eq:devubprimaldet4}).
Since (\ref{eq:devubprimaldet4}) is the dual of (\ref{eq:devubprimaldet}) and since the strict duality obviously holds (the optimization problem in (\ref{eq:devubprimaldet}) is clearly convex) one has
\begin{equation}
-\xi_{ov}^{(2)}(\sigma,\g,\h,\xtilde)=-\xi_{ov}^{(1)}(\sigma,\g,\h,\xtilde).\label{eq:devubdualprimal41}
\end{equation}
 On the other hand the optimization problem in (\ref{eq:devubprimaldet4}) is the same as the one in (\ref{eq:devubLagran12}) and therefore
\begin{equation}
-\xi_{ov}^{(2)}(\sigma,\g,\h,\xtilde)=-\xi_{dual}(\sigma,\g,\h,\xtilde).\label{eq:devubdualprimal42}
\end{equation}
Connecting (\ref{eq:devubdeteq1}), (\ref{eq:devubdualprimal41}), and (\ref{eq:devubdualprimal42}) one finally has
\begin{equation}
-\xi_{dual}(\sigma,\g,\h,\xtilde)=-\xi_{ov}(\sigma,\g,\h,\xtilde)\label{eq:devublemmadetcond}
\end{equation}
which is what is stated in (\ref{eq:devublemmadet2}). This concludes the proof.
\end{proof}
Let $\hat{d},\hat{\nu},\widehat{\lambda^{(2)}}$ be the solution of (\ref{eq:devubLagran11}). Clearly, $\hat{d}=\|\hat{\w}\|_2=\sigma\frac{\|\h+\hat{\nu}\z^{(1)}-\widehat{\lambda^{(2)}}\|_2}{\sqrt{\|\g\|_2^2-\|\h+\hat{\nu}\z^{(1)}-\widehat{\lambda^{(2)}}\|_2^2}}$ and since all quantities concentrate
$E\hat{d}=E\|\hat{\w}\|_2\doteq\sigma\frac{E\|\h+\hat{\nu}\z^{(1)}-\widehat{\lambda^{(2)}}\|_2}{\sqrt{E\|\g\|_2^2-E\|\h+\hat{\nu}\z^{(1)}-\widehat{\lambda^{(2)}}\|_2^2}}$. Now, set $C_{\w_{up}}=E\|\hat{\w}\|_2$ in (\ref{eq:upperdefxi}). Then a combination of (\ref{eq:upperdefxi}), (\ref{eq:devubLagran11}), and Lemma \ref{thm:devublemmadet} gives
\begin{multline}
E\xi_{up}(\sigma,\g,\h,\xtilde,E\|\hat{\w}\|_2)=E\max_{\lambda^{(2)}\in \Lambda^{(2)},\nu\geq 0}(\sqrt{(E\|\hat{\w}\|_2)^2+\sigma^2}\|\g\|_2-E\|\hat{\w}\|_2\|\h+\nu\z^{(1)}-\lambda^{(2)})\|_2
-\sum_{i=n-k+1}^{n}\lambda_i^{(2)}\xtilde_i)\\
=E\max_{\lambda^{(2)}\in \Lambda^{(2)},\nu\geq 0}(\sqrt{(E\hat{d})^2+\sigma^2}\|\g\|_2-E\hat{d}\|\h+\nu\z^{(1)}-\lambda^{(2)})\|_2
-\sum_{i=n-k+1}^{n}\lambda_i^{(2)}\xtilde_i)\\
\doteq E\min_{d\geq 0}\max_{\lambda^{(2)}\in \Lambda^{(2)},\nu\geq 0}(\sqrt{d^2+\sigma^2}\|\g\|_2-d\|\h+\nu\z^{(1)}-\lambda^{(2)})\|_2
-\sum_{i=n-k+1}^{n}\lambda_i^{(2)}\xtilde_i)
= E \xi_{ov}(\sigma,\g,\h,\xtilde).\label{eq:devubfinal}
\end{multline}
Combining Lemma \ref{thm:upperbound} and (\ref{eq:devubfinal}) one has that with overwhelming probability there is a $\w$ such that the objective in (\ref{eq:lassol1}) is upper bounded by a quantity arbitrarily close from above to $E \xi_{ov}(\sigma,\g,\h,\xtilde)$. On the other hand Lemma \ref{thm:lowerbound} states that for any $\w$ such that $|\|\w\|_2-\|\hat{\w}\|_2\|\geq \epsilon_{w_{up}}\|\hat{\w}\|_2$, $\epsilon_{w_{up}}>0$, the objective value of (\ref{eq:lassol1}) is with overwhelming probability lower bounded by a quantity that is arbitrarily close from below to $(1+\frac{\epsilon_{w_{up}}^2}{2(1+\epsilon_{w_{up}})})E \xi_{ov}(\sigma,\g,\h,\xtilde)$. Clearly then the assumption of Lemma \ref{thm:lowerbound} is unsustainable and one has that $\|\w_{lasso}\|_2$ can not deviate substantially from $\|\hat{\w}\|_2$. This then implies that with overwhelming probability the objective value of (\ref{eq:lassol1}) concentrates around $E \xi_{ov}(\sigma,\g,\h,\xtilde)$  and consequently that $\|\w_{lasso}\|_2$ concentrates around $E\|\hat{\w}\|_2$.

\subsection{Connecting all pieces}\label{sec:connectpieces}

In this section we connect all of the above. We will summarize the results obtained so far in the following theorem.
\begin{theorem}
Let $\v$ be an $n\times 1$ vector of i.i.d. zero-mean variance $\sigma^2$ Gaussian random variables and let $A$ be an $m\times n$ matrix of i.i.d. standard normal random variables. Further, let $\g$ and $\h$ be $m\times 1$ and $n\times 1$ vectors of i.i.d. standard normals, respectively. Consider a $k$-sparse $\xtilde$ defined in (\ref{eq:xtildedef}) and a $\y$ defined in (\ref{eq:systemnoise}) for $\x=\xtilde$. Let the solution of (\ref{eq:lassol1}) be $\hat{\x}$ and let the so-called error vector of LASSO from (\ref{eq:lassol1}) be $\w_{lasso}=\hat{\x}-\xtilde$. Let $n$ be large and let constants $\alpha=\frac{m}{n}$ and $\beta=\frac{k}{n}$ be below the fundamental characterization (\ref{eq:fundl1}). Consider the following optimization problem:
\begin{eqnarray}
\xi_{ov}(\sigma,\g,\h,\xtilde)=\max_{\nu,\lambda^{(2)}} & & \sigma\sqrt{\|\g\|_2^2-\|\h+\nu\z^{(1)}-\lambda^{(2)}\|_2^2} -\sum_{i=n-k+1}^{n}\lambda_i^{(2)}\xtilde_i\nonumber \\
\mbox{subject to}
& & \nu\geq 0\nonumber \\
& & 0 \leq\lambda_i^{(2)}\leq 2\nu,1\leq i\leq n.\label{eq:mainlasso1}
\end{eqnarray}
Let $\hat{\nu}$ and $\widehat{\lambda^{(2)}}$ be the solution of (\ref{eq:mainlasso1}). Set
\begin{equation}
\|\hat{\w}\|_2=\sigma\frac{\|\h+\hat{\nu}\z^{(1)}-\widehat{\lambda^{(2)}}\|_2}{\sqrt{\|\g\|_2^2-\|\h+\hat{\nu}\z^{(1)}-\widehat{\lambda^{(2)}}\|_2^2}}.\label{eq:mainlasso2}
\end{equation}
Then:
\begin{equation}
P((1-\epsilon_1^{(lasso)})E\xi_{ov}(\sigma,\g,\h,\xtilde)\leq \|\y-A\hat{\x}\|_2
\leq (1+\epsilon_1^{(lasso)})E\xi_{ov}(\sigma,\g,\h,\xtilde))=1-e^{-\epsilon_2^{(lasso)}n}\label{eq:mainlasso3}
\end{equation}
and
\begin{equation}
P((1-\epsilon_1^{(lasso)})E\|\hat{\w}\|_2\leq \|\w_{lasso}\|_2
\leq (1+\epsilon_1^{(lasso)})E\|\hat{\w}\|_2) =1-e^{-\epsilon_2^{(lasso)}n},\label{eq:mainlasso4}
\end{equation}
where $\epsilon_1^{(lasso)}>0$ is an arbitrarily small constant and $\epsilon_2^{(lasso)}$ is a constant dependent on $\epsilon_1^{(lasso)}$ and $\sigma$ but independent of $n$.
\label{thm:mainlasso}
\end{theorem}
\begin{proof}
Follows from the above discussion and a combination of (\ref{eq:Lagran12}), Lemma \ref{thm:optsollower}, discussion in Section \ref{sec:connectl1}, and Lemmas \ref{thm:lowerbound}, \ref{thm:upperbound}, and \ref{thm:devublemmadet}.
\end{proof}

It may not be clear immediately but the result presented in the above theorem is incredibly powerful. Among other things, it enables one to precisely estimate the norm of the error vector in ``noisy" under-determined systems of linear equations. Moreover, it can do so for any given $k$-sparse vector $\xtilde$. Furthermore, all of it is done through a transformation of the original LASSO from (\ref{eq:lassol1}) to a much simpler optimization program (\ref{eq:mainlasso1}). While many quantities of interest in LASSO recovery can be computed through the mechanism presented above, below we focus only on quantities that relate to what we will call LASSO's \emph{generic} performance. Computation of all other quantities that we consider are of interest will be presented in a series of forthcoming papers.

\subsubsection{LASSO's generic performance}\label{sec:generic}

The results presented in the above theorem are fairly general and pertain to pretty much any possible scenario one can imagine. Here we will focus on the so-called ``worst-case" scenario or as we will refer to it ``generic" performance scenario. We will now show that $E\|\hat{\w}\|_2$ from Theorem \ref{thm:mainlasso} can be upper-bounded over the set of all $\xtilde$'s. To that end let us assume that all nonzero components of $\xtilde$ are infinite. The optimization problem from (\ref{eq:mainlasso1}) then becomes
\begin{eqnarray}
\xi_{ov}^{(gob)}(\sigma,\g,\h)=\max_{\nu,\lambda^{(2)}} & & \sigma\sqrt{\|\g\|_2^2-\|\h+\nu\z^{(1)}-\lambda^{(2)}\|_2^2} \nonumber \\
\mbox{subject to}
& & \nu\geq 0\nonumber \\
& & 0 \leq\lambda_i^{(2)}=0,n-k+1\leq i\leq n\nonumber \\
& & 0 \leq\lambda_i^{(2)}\leq 2\nu, 1\leq i\leq n-k.\label{eq:genlasso1}
\end{eqnarray}
Let $\nu_{gen}$ and $\lambda^{(gen)}$ be the solution of (\ref{eq:genlasso1}) and let $\w_{gen}$ be the error vector in case when all nonzero compoenents of $\xtilde$ are infinite (in Section \ref{sec:connectl1} for a slightly changed version of (\ref{eq:genlasso1}) $\nu_{gen}$ and $\lambda^{(gen)}$ were referred to as $\nu_{\ell_1}$ and $\lambda^{(\ell_1)}$). Now, let us assume that some of nonzero components of $\xtilde$ in (\ref{eq:mainlasso1}) are finite. And let as usual $\hat{\nu}$ and $\widehat{\lambda^{(2)}}$ be the solution of (\ref{eq:mainlasso1}) and let $\|\hat{\w}\|_2$ be the norm of the LASSO's error vector. Since
\begin{equation}
\sigma\sqrt{\|\g\|_2^2-\|\h+\hat{\nu}\z^{(1)}-\widehat{\lambda^{(2)}}\|_2^2} -\sum_{i=n-k+1}^{n}\widehat{\lambda_i^{(2)}}\xtilde_i\geq
\sigma\sqrt{\|\g\|_2^2-\|\h+\nu_{gen}\z^{(1)}-\lambda^{(gen)}\|_2^2}
\end{equation}
and $\widehat{\lambda_i^{(2)}}\geq 0,n-k+1\leq i\leq n$, one has that
\begin{equation}
\|\h+\hat{\nu}\z^{(1)}-\widehat{\lambda^{(2)}}\|_2 \leq \|\h+\nu_{gen}\z^{(1)}-\lambda^{(gen)}\|_2.\label{eq:genlasso2}
\end{equation}
Furthermore, one then has for the norm of error vectors
\begin{equation}
\|\hat{\w}\|_2=\sigma\frac{\|\h+\hat{\nu}\z^{(1)}-\widehat{\lambda^{(2)}}\|_2}{\sqrt{\|\g\|_2^2-{\|\h+\hat{\nu}\z^{(1)}-\widehat{\lambda^{(2)}}\|_2^2}}} \leq
\sigma\frac{\|\h+\nu_{gen}\z^{(1)}-\lambda^{(gen)}\|_2}{\sqrt{\|\g\|_2^2-{\|\h+\nu_{gen}\z^{(1)}-\lambda^{(gen)}\|_2^2}}} =\|\w_{gen}\|_2.\label{eq:genlasso3}
\end{equation}
Then the following \emph{generic} equivalent to Theorem \ref{thm:mainlasso} can be established.
\begin{theorem}
Assume the setup of Theorem \ref{thm:mainlasso}. Consider the following optimization problem:
\begin{eqnarray}
\xi_{ov}^{(gen)}(\sigma,\g,\h)=\min_{\nu,\lambda^{(2)}} & & \|\h+\nu\z^{(1)}-\lambda^{(2)}\|_2\nonumber \\
\mbox{subject to}
& & \nu\geq 0\nonumber \\
& & \lambda_i^{(2)}=0,n-k+1\leq i\leq n\nonumber \\
& & 0 \leq\lambda_i^{(2)}\leq 2\nu, 1\leq i\leq n-k.\label{eq:genlasso4}
\end{eqnarray}
Let $\nu_{gen}$ and $\lambda^{(gen)}$ be the solution of (\ref{eq:genlasso4}). Set
\begin{equation}
\|\w_{gen}\|_2=\sigma\frac{\|\h+\nu_{gen}\z^{(1)}-\lambda^{(gen)}\|_2}{\sqrt{\|\g\|_2^2-\|\h+\nu_{gen}\z^{(1)}-\lambda^{(gen)}\|_2^2}}.\label{eq:genlasso5}
\end{equation}
Then:
\begin{eqnarray}
P(\exists\w_{lasso}|\|\w_{lasso}\|_2\in((1-\epsilon_1^{(lasso)})E\|\w_{gen}\|_2, (1+\epsilon_1^{(lasso)})E\|\w_{gen}\|_2)) & \geq & 1-e^{-\epsilon_2^{(lasso)}n}\nonumber \\
P(\|\w_{lasso}\|_2\leq (1+\epsilon_1^{(lasso)})E\|\w_{gen}\|_2)) & \geq & 1-e^{-\epsilon_3^{(lasso)}n},\nonumber \\\label{eq:genlasso6}
\end{eqnarray}
where $\epsilon_1^{(lasso)}>0$ is an arbitrarily small constant and $\epsilon_2^{(lasso)}$ and $\epsilon_3^{(lasso)}$ are constants dependent on $\epsilon_1^{(lasso)}$ and $\sigma$ but independent of $n$.
\label{thm:genlasso}
\end{theorem}
\begin{proof}
Follows from the above discussion, Theorem \ref{thm:mainlasso}, and by noting that the optimization problems in (\ref{eq:genlasso4}) and (\ref{eq:genlasso1}) are equivalent.
\end{proof}
The following corollary then provides a quick way of computing the concentrating point of the ``worst case" norm of the error vector.
\begin{corollary}
Assume the setup of Theorems \ref{thm:mainlasso} and \ref{thm:genlasso}. Let $\alpha=\frac{m}{n}$ and $\beta_w=\frac{k}{n}$. Then
\begin{eqnarray}
P(\exists\w_{lasso}|\|\w_{lasso}\|_2\in((1-\epsilon_1^{(lasso)})\sigma\sqrt{\frac{\alpha_w}{\alpha-\alpha_w}}, (1+\epsilon_1^{(lasso)})\sigma\sqrt{\frac{\alpha_w}{\alpha-\alpha_w}})) & \geq & 1-e^{-\epsilon_2^{(lasso)}n}\nonumber \\
P(\|\w_{lasso}\|_2\leq (1+\epsilon_1^{(lasso)})\sigma\sqrt{\frac{\alpha_w}{\alpha-\alpha_w}}) & \geq & 1-e^{-\epsilon_2^{(lasso)}n},\nonumber \\\label{eq:genlasso7}
\end{eqnarray}
where $\alpha_w<\alpha$ is such that
\begin{equation}
(1-\beta_w)\frac{\sqrt{\frac{2}{\pi}}e^{-(\erfinv(\frac{1-\alpha_w}{1-\beta_w}))^2}}{\alpha_w}-\sqrt{2}\erfinv (\frac{1-\alpha_w}{1-\beta_w})=0
\label{eq:genfundl1}
\end{equation}
and $\epsilon_1^{(lasso)}>0$ is an arbitrarily small constant and $\epsilon_2^{(lasso)}$ is a constant dependent on $\epsilon_1^{(lasso)}$ and $\sigma$ but independent of $n$.
\label{thm:gencomperror}
\end{corollary}
\begin{proof}
Let $\bar{\h}$ and $\z^{(2)}$ be as in Section \ref{sec:connectl1}.
Then
\begin{eqnarray}
\xi_{ov}^{(gen)}(\sigma,\g,\h)=\min_{\nu,\lambda^{(2)}} & & \|\bar{\h}-\nu\z^{(2)}+\lambda^{(2)}\|_2\nonumber \\
\mbox{subject to}
& & \nu\geq 0\nonumber \\
& & \lambda_i^{(2)}=0,n-k+1\leq i\leq n\nonumber \\
& & 0 \leq\lambda_i^{(2)}\leq \nu, 1\leq i\leq n-k,\label{eq:genlasso8}
\end{eqnarray}
is equivalent to (\ref{eq:genlasso7}). Moreover
\begin{eqnarray}
E\|\w_{gen}\|_2 & \doteq & \sigma\frac{E\xi_{ov}^{(gen)}(\sigma,\g,\h)}{\sqrt{E\|\g\|_2^2-E\xi_{ov}^{(gen)}(\sigma,\g,\h)^2}}\nonumber \\
& \doteq & \sigma\sqrt{\frac{\alpha_w}{\alpha-\alpha_w}},\label{eq:genlasso9}
\end{eqnarray}
where $\alpha_w m\doteq E\xi_{ov}^{(gen)}(\sigma,\g,\h)^2$ is one of the main contributions of \cite{StojnicCSetam09}.
The rest then trivially follows from (\ref{eq:genlasso6}).
\end{proof}
Using (\ref{eq:genfundl1}) and (\ref{eq:genlasso1}) one can then for any $\sigma$ and any pair $(\alpha,\beta_w)$ (that is below fundamental characterization (\ref{eq:fundl1})) determine the value of the worst case $E\|\w_{lasso}\|_2$ as $\sigma\sqrt{\frac{\alpha_w}{\alpha-\alpha_w}}$. We present the obtained results in Figure \ref{fig:lassoweakthr}. For several fixed values of worst case $E\|\w_{lasso}\|_2$ we determine curves of points $(\alpha,\beta_w)$ for which these fixed values are achieved (of course for any $\alpha$ that is below a curve the value of the corresponding worst case $E\|\w_{lasso}\|_2$ is smaller). As can be seen from the plots the lower the norm-2 of the error vector the smaller the allowable region for pairs $(\alpha,\beta_w)$.

The results of the above corollary match those obtained in \cite{DonMalMon10,BayMon10lasso} through a state evolution/bilief propagation type of analysis. The above corollary relates to the LASSO from (\ref{eq:lassol1}) whereas the results from \cite{DonMalMon10,BayMon10lasso} are derived for somewhat different LASSO from (\ref{eq:biglasso}). However, as mentioned earlier, in Section \ref{sec:connectlasso} we will establish a nice connection between the LASSO from (\ref{eq:lassol1}) and one that is fairly similar to (\ref{eq:biglasso}).
\begin{figure}[htb]
\centering
\centerline{\epsfig{figure=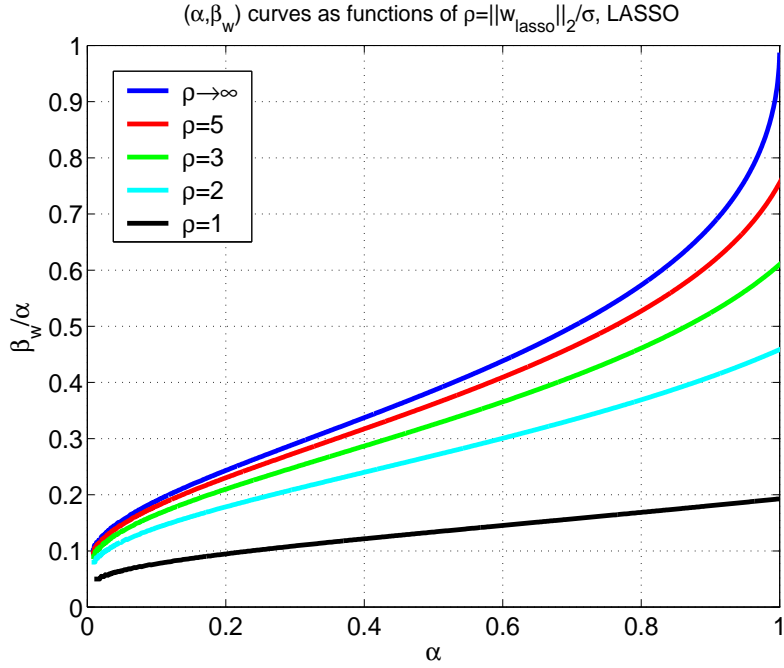,width=10.5cm,height=9cm}}
\vspace{0in} \caption{$(\alpha,\beta_w)$ curves as functions of $\rho=\frac{E\|\w_{lasso}\|_2}{\sigma}$ for LASSO algorithm from (\ref{eq:lassol1})}
\label{fig:lassoweakthr}
\end{figure}

\section{LASSO's performance analysis framework -- signed $\x$} \label{sec:signed}

In this section we show how the LASSO's performance analysis framework developed in the previous section can be specialized to the case when signals are \emph{a priori} known to have nonzero components of certain sign. All major assumptions stated at the beginning of the previous section will continue to hold in this section as well; namely, we will continue to consider matrices $A$ with i.i.d. standard normal random variables; elements of $\v$ will again be i.i.d. Gaussian random variables with zero mean and variance $\sigma$. The main difference, though, comes in the definition of $\xtilde$. We will in this section assume that $\xtilde$ is the original $\x$ in (\ref{eq:systemnoise}) that we are trying to recover and that it is \emph{any} $k$-sparse vector with a given fixed location of its nonzero elements and with a priori known signs of its elements. Given the statistical context,
it will be fairly easy to see later on that everything that we will present in this section will be irrelevant with respect to what particular location and what particular combination of signs of nonzero elements are chosen. We therefore for the simplicity of the exposition and without loss of generality assume that the components $\x_{1},\x_{2},\dots,\x_{n-k}$ of $\x$ are equal to zero and the components $\x_{n-k+1},\x_{n-k+2},\dots,\x_n$ of $\x$ are greater than or equal to zero. However, differently from what was assumed in the previous section, we now assume that this information is \emph{a priori} known. That essentially means that
this information is also known to the solving algorithm. Then instead of (\ref{eq:lassol1}) one can consider its a better (``signed") version
\begin{eqnarray}
\min_{\x} & & \|\y-A\x\|_2 \nonumber \\
\mbox{subject to} & & \|\x\|_1\leq \|\xtilde\|_1\nonumber \\
& & \x_{i}\geq 0,1\leq i\leq n.
\label{eq:lassol1non}
\end{eqnarray}

In what follows we will mimic the procedure presented in the previous section, skip all the obvious parallels, and emphasize the points that are different.
The framework that we will present below will again center around finding the optimal value of the objective function in (\ref{eq:lassol1non}). In the first of the following two subsections we will create a lower bound on this optimal value (this will essentially amount to creating a procedure that is analogous to the one presented in Section \ref{sec:unsignedlbzetaobj}). We will then afterwards in the second of the subsections create an upper bound on this optimal value. As it was done in the case of general $\xtilde$ in the previous section we will in the third subsection show that the two bounds actually match. To make further writing easier and clearer we set already here
\begin{eqnarray}
\zeta_{obj+}=\min_{\x} & & \|\y-A\x\|_2 \nonumber \\
\mbox{subject to} & & \|\x\|_1\leq \|\xtilde\|_1\nonumber \\
& & \x_i\geq 0,1\leq i\leq n.\label{eq:objlassol1non}
\end{eqnarray}

\subsection{Lower-bounding $\zeta_{obj+}$} \label{sec:unsignedlbzetaobjnon}

In this section we present the part of the framework that relates to finding a ``high-probability" lower bound on $\zeta_{obj}^+$.
As in the previous section we again assume that there is a (if necessary, arbitrarily large) constant $C_\w$ such that
\begin{equation}
P(\|\w_{lasso}\|_2\leq C_\w)=1-e^{-\epsilon_{C_{\w}}n}.\label{eq:assumplassonon}
\end{equation}
We again start by noting that if one knows that $\y=A\xtilde+\v$ holds then (\ref{eq:objlassol1non}) can be rewritten as
\begin{eqnarray}
\min_{\x} & & \|\v+A\xtilde-A\x\|_2 \nonumber \\
\mbox{subject to} & & \|\x\|_1\leq \|\xtilde\|_1\nonumber \\
& & \x_i\geq 0,1\leq i\leq n.\label{eq:objlassol11non}
\end{eqnarray}
After a small change of variables, $\x=\xtilde+\w$, one has an equivalent to (\ref{eq:objlassol13})
\begin{eqnarray}
\min_{\w} & & \|A_{\v}\begin{bmatrix} \w\\\sigma\end{bmatrix}\|_2 \nonumber \\
\mbox{subject to} & & \sum_{i=1}^{n}\w_i\leq 0\nonumber \\
& & \xtilde_i+\w_i\geq 0,1\leq i\leq n,\label{eq:objlassol13non}
\end{eqnarray}
where as earlier $A_{\v}=\begin{bmatrix} -A & \v \end{bmatrix}$ is an $m\times (n+1)$ random matrix with i.i.d. standard normal components. Let
\begin{equation}
S_{\w}^+(\sigma,\xtilde,C_\w)=\{\begin{bmatrix}\w\\\sigma\end{bmatrix} \in R^{n+1}| \quad \|\w\|_2\leq C_\w \quad \mbox{and} \quad \sum_{i=1}^{n}\w_i\leq 0 \quad \mbox{and} \quad \xtilde_i+\w_i\geq 0, 1\leq i\leq n\}.\label{eq:defSnon}
\end{equation}
Further, let
\begin{equation}
f_{obj+}(\sigma,\w)=\|A_{\v}\begin{bmatrix} \w\\\sigma\end{bmatrix}\|_2 \label{eq:deffobjnon}
\end{equation}
and set,
\begin{equation}
\hspace{-.3in}\zeta_{obj+}^{(help)}=\min_{[\w^T \sigma]^T\in S_{\w}^+(\sigma,\xtilde,C_\w)} f_{obj+}(\sigma,\w)= \min_{[\w^T \sigma]^T\in S_{\w}^+(\sigma,\xtilde,C_\w)}  \|A_{\v}\begin{bmatrix} \w\\\sigma\end{bmatrix}\|_2=
\min_{[\w^T \sigma]^T\in S_{\w}^+(\sigma,\xtilde,C_\w)}\max_{\|\a\|_2=1}  \a^T A_{\v}\begin{bmatrix} \w\\\sigma\end{bmatrix}.\label{eq:objlassol14non}
\end{equation}
Now, after applying Lemma \ref{thm:unsignedlemma} and following the procedure from the previous section one has
\begin{equation}
\hspace{-0in}P(\min_{[\w^T \sigma]^T\in S_{\w}^+(\sigma,\xtilde,C_\w)}(f_{obj+}(\sigma,\w)+\sqrt{\|\w\|_2^2+\sigma^2}g)\geq \zeta_{obj+}^{(l)})\geq p_l^+.\label{eq:objlassol15non}
\end{equation}
where
\begin{equation}
p_l^+=P\left (\min_{[\w^T \sigma]^T\in S_{\w}^+(\sigma,\xtilde,C_\w)}\left (  \sqrt{\|\w\|_2^2+\sigma^2}\|\g\|_2+\sum_{i=1}^{n}\h_i\w_i\right )+\h_{n+1}\sigma \geq \zeta_{obj+}^{(l)}\right ).\label{eq:probint1non}
\end{equation}
As in previous section we will essentially show that for certain $\zeta_{obj+}^{(l)}$ this probability is close to $1$ which will imply that we have a ``high probability" lower bound on $\zeta_{obj+}$. Let
\begin{equation}
\xi_{+}(\sigma,\g,\h,\xtilde)=\min_{[\w^T \sigma]^T\in S_{\w}^+(\sigma,\xtilde,C_\w)} \left ( \sqrt{\|\w\|_2^2+\sigma^2}\|\g\|_2+\sum_{i=1}^{n}\h_i\w_i\right ).\label{eq:defxinon}
\end{equation}
Now we split the analysis into two parts. The first one will be the deterministic analysis of $\xi_+(\sigma,\g,\h,\xtilde)$ and will be presented in Subsection \ref{sec:unsigneddetnon}. In the second part (that will be presented in Subsection \ref{sec:unsignedconcnon}) we will use the results of such a deterministic analysis and continue the above probabilistic analysis applying various concentration results.

\subsubsection{Optimizing $\xi_{+}(\sigma,\g,\h,\xtilde)$} \label{sec:unsigneddetnon}

In this section we compute $\xi_{+}(\sigma,\g,\h)$. We first rewrite the optimization problem from (\ref{eq:defxinon}) in the following way
\begin{eqnarray}
\xi_{+}(\sigma,\g,\h,\xtilde)=\min_{\w} & & \sqrt{\|\w\|_2^2+\sigma^2}\|\g\|_2+\sum_{i=1}^{n}\h_i\w_i \nonumber \\
\mbox{subject to} & & \sum_{i=1}^{n}\w_i\leq 0\nonumber \\
& & \xtilde_i+\w_i\geq 0,1\leq i\leq n\nonumber \\
& & \sqrt{\|\w\|_2^2+\sigma^2}\leq \sqrt{C_\w^2+\sigma^2}.\label{eq:defxi2non}
\end{eqnarray}
The Lagrange dual of the above problem then becomes
\begin{multline}
{\cal L}(\nu,\lambda^{(2)},\w,\gamma)=\sqrt{\|\w\|_2^2+\sigma^2}\|\g\|_2+\sum_{i=1}^{n}\h_i\w_i +\nu\sum_{i=1}^n \w_i\\-\sum_{i=n-k+1}^{n}\lambda_i^{(2)}(\xtilde_i+\w_i)
-\sum_{i=1}^{n-k}\lambda_i^{(2)}\w_i+\gamma(\sqrt{\|\w\|_2^2+\sigma^2}- \sqrt{C_\w^2+\sigma^2}).\label{eq:Lagran1non}
\end{multline}
After a few further arrangements we finally have
\begin{multline}
{\cal L}(\nu,\lambda^{(2)},\w,\gamma)=\sqrt{\|\w\|_2^2+\sigma^2}(\|\g\|_2+\gamma)+\sum_{i=1}^{n}(\h_i+\nu-\lambda_i^{(2)})\w_i  -\sum_{i=n-k+1}^{n}\lambda_i^{(2)}\xtilde_i
-\gamma \sqrt{C_\w^2+\sigma^2}.\label{eq:Lagran3non}
\end{multline}
One can then write the following dual problem of (\ref{eq:defxi2non})
\begin{eqnarray}
\xi_{+}(\sigma,\g,\h,\xtilde)=\max_{\nu,\lambda^{(2)},\gamma}\min_{\w} & & {\cal L}(\nu,\lambda^{(2)},\w)\nonumber \\
\mbox{subject to} & & \lambda_i^{(2)}\geq 0, 1\leq i\leq n\nonumber \\
& & \nu\geq 0\nonumber \\
& & \gamma\geq 0,\label{eq:Lagran4non}
\end{eqnarray}
where we of course use the fact that the strict duality obviously holds. The inner minimization over $\w$ is now doable. Setting the derivatives with respect to $\w_i$ to zero one obtains
\begin{equation}
\frac{\w(\|\g\|_2+\gamma)}{\sqrt{\|\w\|_2^2+\sigma^2}}+(\h+\nu\z^{(1)}-\lambda^{(2)})=0,\label{eq:solvewnon}
\end{equation}
where $\z^{(1)}$ and $\lambda^{(2)}$ are as defined in the previous section.
From (\ref{eq:solvewnon}) one then has
\begin{equation}
\w(\|\g\|_2+\gamma)=-\sqrt{\|\w\|_2^2+\sigma^2}(\h+\nu\z^{(1)}-\lambda^{(2)})\label{eq:solvew1non}
\end{equation}
or in a norm form
\begin{equation}
\|\w\|_2^2(\|\g\|_2+\gamma)^2=(\|\w\|_2^2+\sigma^2)\|\h+\nu\z^{(1)}-\lambda^{(2)}\|_2^2.\label{eq:solvew2non}
\end{equation}
From (\ref{eq:solvew2non}) we then find
\begin{equation}
\|\w_{sol+}\|_2=\frac{\sigma\|\h+\nu\z^{(1)}-\lambda^{(2)}\|_2}{\sqrt{(\|\g\|_2+\gamma)^2-\|\h+\nu\z^{(1)}-\lambda^{(2)}\|_2^2}},\label{eq:solvew3non}
\end{equation}
and from (\ref{eq:solvew1non})
\begin{equation}
\w_{sol+}=\frac{\sigma(\h+\nu\z^{(1)}-\lambda^{(2)})}{\sqrt{(\|\g\|_2+\gamma)^2-\|\h+\nu\z^{(1)}-\lambda^{(2)}\|_2^2}}\label{eq:solvew4non}
\end{equation}
where $\w_{sol+}$ is of course the solution of the inner minimization over $\w$. As in the previous section, one should note that (\ref{eq:solvew3non}) and (\ref{eq:solvew4non}) are of course possible only if $\|\g\|_2+\gamma-\|\h+\nu\z^{(1)}-\lambda^{(2)}\|_2\geq 0$. (Also, as in the previous section if for $\nu$ and $\lambda^{(2)}$ that are optimal in (\ref{eq:Lagran4non}) the condition is not met then for the corresponding ($\alpha,\beta$) the worst-case $\|\w\|_2$ is infinite with overwhelming probability). Plugging the value of $\w_{sol+}$ from (\ref{eq:solvew4non}) back in (\ref{eq:Lagran4non}) gives
\begin{eqnarray}
\hspace{-.5in}\xi_{+}(\sigma,\g,\h,\xtilde)=\max_{\nu,\lambda^{(2)},\gamma} & & \sigma\sqrt{(\|\g\|_2+\gamma)^2-\|\h+\nu\z^{(1)}-\lambda^{(2)}\|_2^2} -\sum_{i=n-k+1}^{n}\lambda_i^{(2)}\xtilde_i-\gamma \sqrt{C_\w^2+\sigma^2}\nonumber \\
\mbox{subject to} & & \lambda_i^{(2)}\geq 0, 1\leq i\leq n\nonumber \\
& & \nu\geq 0\nonumber \\
& & \|\g\|_2+\gamma-\|\h+\nu\z^{(1)}-\lambda^{(2)}\|_2\geq 0\nonumber \\
& & \gamma\geq 0.\label{eq:Lagran7non}
\end{eqnarray}
Now, the maximization over $\gamma$ can be done. After setting the derivative to zero one finds
\begin{equation}
\frac{\|\g\|_2+\gamma}{\sqrt{(\|\g\|_2+\gamma)^2-\|\h+\nu\z^{(1)}-\lambda^{(2)}\|_2^2}}-\sqrt{C_\w^2+\sigma^2}=0\label{eq:dergamma}
\end{equation}
and after some algebra
\begin{equation}
\gamma_{opt+}=\sqrt{1+\frac{\sigma^2}{C_\w^2}}\|\h+\nu\z^{(1)}-\lambda^{(2)}\|_2-\|\g\|_2,\label{eq:optgamma}
\end{equation}
where of course $\gamma_{opt+}$ would be the solution of (\ref{eq:Lagran7non}) only if larger than or equal to zero. Alternatively of course $\gamma_{opt+}=0$. Now, based on these two scenarios we distinguish two different optimization problems:
\begin{enumerate}
\item \underline{\emph{The ``overwhelming" optimization --- signed $\xtilde$}}
\begin{eqnarray}
\xi_{ov+}(\sigma,\g,\h,\xtilde)=\max_{\nu,\lambda^{(2)}} & & \sigma\sqrt{\|\g\|_2^2-\|\h+\nu\z^{(1)}-\lambda^{(2)}\|_2^2} -\sum_{i=n-k+1}^{n}\lambda_i^{(2)}\xtilde_i\nonumber \\
\mbox{subject to}
& & \nu\geq 0\nonumber \\
& & \lambda_i^{(2)}\geq 0,1\leq i\leq n.\label{eq:Lagran12non}
\end{eqnarray}
\item \underline{\emph{The ``non-overwhelming" optimization --- signed $\xtilde$}}
\begin{eqnarray}
\xi_{nov+}(\sigma,\g,\h,\xtilde)=\max_{\nu,\lambda^{(2)}} & & \sqrt{C_\w^2+\sigma^2}\|\g\|_2-C_\w\|\h+\nu\z^{(1)}-\lambda^{(2)}\|_2 -\sum_{i=n-k+1}^{n}\lambda_i^{(2)}\xtilde_i\nonumber \\
\mbox{subject to}
& & \nu\geq 0\nonumber \\
& & \lambda_i^{(2)}\geq 0,1\leq i\leq n.\label{eq:Lagran13non}
\end{eqnarray}
\end{enumerate}
The ``overwhelming" optimization is the equivalent to (\ref{eq:Lagran7non}) if for its optimal values $\widehat{\nu^+}$ and $\widehat{\lambda^{(2+)}}$ holds
\begin{equation}
\sqrt{1+\frac{\sigma^2}{C_\w^2}}\|\h+\widehat{\nu^+}\z^{(1)}-\widehat{\lambda^{(2+)}}\|_2\leq \|\g\|_2,\label{eq:ovnoncondnon}
\end{equation}
We now summarize in the following lemma the results of this subsection.
\begin{lemma}
Let $\widehat{\nu^+}$ and $\widehat{\lambda^{(2+)}}$ be the solutions of (\ref{eq:Lagran12non}) and analogously let $\widetilde{\nu^+}$ and $\widetilde{\lambda^{(2+)}}$ be the solutions of (\ref{eq:Lagran13non}). Let $\xi_{+}(\sigma,\g,\h,\xtilde)$ be, as defined in (\ref{eq:defxinon}), the optimal value of the objective function in (\ref{eq:defxinon}). Then
\begin{equation}
\hspace{-.8in}\xi_{+}(\sigma,\g,\h,\xtilde)=\begin{cases}\sigma\sqrt{\|\g\|_2^2-\|\h+\widehat{\nu^+}\z^{(1)}-\widehat{\lambda^{(2+)}}\|_2^2} -\sum_{i=n-k+1}^{n}\widehat{\lambda_i^{(2+)}}\xtilde_i, &
\hspace{-.62in}\mbox{if}\quad  \sqrt{1+\frac{\sigma^2}{C_\w^2}}\|\h+\widehat{\nu^+}\z^{(1)}-\widehat{\lambda^{(2+)}}\|_2\leq \|\g\|_2\\
\sqrt{C_\w^2+\sigma^2}\|\g\|_2-C_\w\|\h+\widetilde{\nu^+}\z^{(1)}-\widetilde{\lambda^{(2+)}}\|_2 -\sum_{i=n-k+1}^{n}\widetilde{\lambda_i^{(2+)}}\xtilde_i, & \mbox{otherwise} \end{cases}.\label{eq:defhatxinon}
\end{equation}
Moreover, let $\widehat{\w^+}$ be the solution of (\ref{eq:defxinon}). Then
\begin{equation}
\widehat{\w^+}(\sigma,\g,\h,\xtilde)=\begin{cases}
\frac{\sigma(\h+\widehat{\nu^+}\z^{(1)}-\widehat{\lambda^{(2+)}})}{\sqrt{\|\g\|_2^2-\|\h+\widehat{\nu^+}\z^{(1)}-\widehat{\lambda^{(2+)}}\|_2^2}}, &
\mbox{if}\quad  \sqrt{1+\frac{\sigma^2}{C_\w^2}}\|\h+\widehat{\nu^+}\z^{(1)}-\widehat{\lambda^{(2+)}}\|_2\leq \|\g\|_2\\
\frac{C_\w(\h+\widetilde{\nu^+}\z^{(1)}-\widetilde{\lambda^{(2+)}})}{\|\h+\widetilde{\nu^+}\z^{(1)}-\widetilde{\lambda^{(2+)}}\|_2}, &
\mbox{otherwise}\end{cases},\label{eq:defhatwnon}
\end{equation}
and
\begin{equation}
\|\widehat{\w^+}(\sigma,\g,\h,\xtilde)\|_2=\begin{cases}
\frac{\sigma\|\h+\widehat{\nu^+}\z^{(1)}-\widehat{\lambda^{(2+)}})\|_2}{\sqrt{\|\g\|_2^2-\|\h+\widehat{\nu^+}\z^{(1)}-\widehat{\lambda^{(2+)}}\|_2^2}}, &
\mbox{if}\quad  \sqrt{1+\frac{\sigma^2}{C_\w^2}}\|\h+\widehat{\nu^+}\z^{(1)}-\widehat{\lambda^{(2+)}}\|_2\leq \|\g\|_2\\
C_\w, & \mbox{otherwise}
\end{cases}.
\label{eq:defhatwnormnon}
\end{equation}\label{thm:optsollowernon}
\end{lemma}
\begin{proof}
The first part follows trivially. The second one follows from (\ref{eq:solvew4non}) by choosing the optimal $\widehat{\nu^+}$ and $\widehat{\lambda^{(2+)}}$ or alternatively $\widetilde{\nu^+}$ and $\widetilde{\lambda^{(2+)}}$.
\end{proof}

\subsubsection{Concentration of $\xi_{+}(\sigma,\g,\h,\xtilde)$} \label{sec:unsignedconcnon}

In this section we establish that $\xi_{+}(\sigma,\g,\h,\xtilde)$ concentrates with high probability around its mean. The following lemma is an analogue to Lemma
\ref{thm:lipschunsigned}
\begin{lemma}
Let $\g$ and $\h$ be $m$ and $n$ dimensional vectors, respectively, with i.i.d. standard normal variables as their components. Let $\sigma>0$ be an arbitrary scalar. Let $\xi(\sigma,\g,\h,\xtilde)$ be as in (\ref{eq:defxi}). Further let $\epsilon_{lip}>0$ be any constant. Then
\begin{equation}
P(|\xi_{+}(\sigma,\g,\h,\xtilde)-E\xi_{+}(\sigma,\g,\h,\xtilde)|\geq \epsilon_{lip}E\xi_{+}(\sigma,\g,\h,\xtilde))\leq \exp \left \{  -\frac{(\epsilon_{lip} E\xi_{+}(\sigma,\g,\h,\xtilde))^2}{2(2C_\w^2+\sigma^2)} \right \}.\label{eq:lipsch1non}
\end{equation}
\label{thm:lipschunsignednon}
\end{lemma}
\begin{proof}It follows by literally repeating every step of proof of Lemma \label{thm:lipschunsigned}. The only difference is that
one now has $S_{\w}^+(\sigma,\xtilde,C_\w)$ instead of $S_{\w}(\sigma,\xtilde,C_\w)$ .
\end{proof}
Moreover one then has that $\|\h+\widehat{\nu^+}\z^{(1)}-\widehat{\lambda^{(2+)}}\|_2$ and $\|\h+\widetilde{\nu^+}\z^{(1)}-\widetilde{\lambda^{(2+)}}\|_2$ concentrate as well which automatically implies that $\widehat{\w^+}$ also concentrates. More formally, one then has analogues to (\ref{eq:lipsch1non})
\begin{eqnarray}
P(|\|\h+\widehat{\nu^+}\z^{(1)}-\widehat{\lambda^{(2+)}}\|_2-E\|\h+\widehat{\nu^+}\z^{(1)}-\widehat{\lambda^{(2+)}}\|_2|\geq
\epsilon_1^{(norm)}E\|\h+\widehat{\nu^+}\z^{(1)}-\widehat{\lambda^{(2+)}}\|_2) & \leq & e^{-\epsilon_2^{(norm)}n}\nonumber \\
P(|\|\h+\widetilde{\nu^+}\z^{(1)}-\widetilde{\lambda^{(2+)}}\|_2-E\|\h+\widetilde{\nu^+}\z^{(1)}-\widetilde{\lambda^{(2+)}}\|_2|\geq
\epsilon_3^{(norm)}E\|\h+\widetilde{\nu^+}\z^{(1)}-\widetilde{\lambda^{(2+)}}\|_2) & \leq & e^{-\epsilon_4^{(norm)}n}\nonumber \\
P(|\|\widehat{\w^+}\|_2-E\|\widehat{\w^+}\|_2|\geq
\epsilon_1^{(\w)}E\|\widehat{\w^+}\|_2) & \leq & e^{-\epsilon_2^{(\w)}n},\nonumber \\\label{eq:conchwnon}
\end{eqnarray}
where as usual $\epsilon_1^{(norm)}>0$, $\epsilon_2^{(norm)}>0$, and $\epsilon_1^{(\w)}>0$ are arbitrarily small constants and $\epsilon_3^{(norm)}$, $\epsilon_4^{(norm)}$, and $\epsilon_2^{(\w)}$ are constant dependent on $\epsilon_1^{(norm)}>0$, $\epsilon_2^{(norm)}>0$, and $\epsilon_1^{(\w)}>0$, respectively, but independent of $n$. After repeating every step between (\ref{eq:conchw}) and (\ref{eq:lowerboundobj}) one arrives to the following analogue to Lemma \ref{thm:lowerbound}.
\begin{lemma}
Let $\v$ be an $n\times 1$ vector of i.i.d. zero-mean variance $\sigma^2$ Gaussian random variables and let $A$ be an $m\times n$ matrix of i.i.d. standard normal random variables. Consider an $\xtilde$ defined in (\ref{eq:xtildedef}) and a $\y$ defined in (\ref{eq:systemnoise}) for $\x=\xtilde$. Let then $\zeta_{obj+}$ be as defined in (\ref{eq:objlassol1non}) and let $\w^+$ be the solution of (\ref{eq:objlassol13non}).
Assume $P(\|\w^+\|_2\leq C_\w)\geq 1-e^{-\epsilon_{C_\w}n}$ for an arbitrarily large constant $C_\w$ and a constant $\epsilon_{C_\w}>0$ dependent on $C_\w$ but independent of $n$. Then there is a constant $\epsilon_{lower}>0$
\begin{equation}
P(\zeta_{obj+}\geq \zeta_{obj+}^{(lower)})\geq (1-e^{-\epsilon_{lower}n})(1-e^{-\epsilon_{C_\w}n}),\label{eq:lowerboundobjthm1non}
\end{equation}
where
\begin{equation}
\zeta_{obj+}^{(lower)}=(1-\epsilon_{lip})E\xi_{+}(\sigma,\g,\h,\xtilde)-\epsilon_1^{(\h)}\sqrt{n}-\epsilon_1^{(g)}\sqrt{n},\label{eq:lowerboundobjthm2non}
\end{equation}
$\xi_{+}(\sigma,\g,\h,\xtilde)$ is as defined in (\ref{eq:defxinon}), and $\epsilon_{lip},\epsilon_1^{(\h)},\epsilon_1^{(g)}$ are all positive arbitrarily small constants.
\label{thm:lowerboundnon}
\end{lemma}
\begin{proof}
Follows from the previous discussion.
\end{proof}

\subsection{Upper-bounding $\zeta_{obj+}$} \label{sec:unsignedubzetaobjnon}

In this section we present a general framework for finding a ``high-probability" upper bound on $\zeta_{obj+}$. To that end, let $r_{+}$ and $C_{\w_{up+}}$ be positive scalars (as in Section \ref{sec:unsignedubzetaobjnon}, we in this subsection present a general framework and take these scalars to be arbitrary; however to make the bound as tight sa possible in the following subsection we will make them take particular values). As earlier, if we can show that there is a $\w\in R^n$ such that
$\|\xtilde+\w\|_1\leq \|\xtilde\|_1$ and $\|\v-A\w\|_2\leq r_+$ with overwhelming probability then $r_+$ can act as an upper bound on $\zeta_{obj+}$. We then start by looking at the following optimization problem
\begin{eqnarray}
\min_{\w} & & \|\xtilde+\w\|_1-\|\xtilde\|_1 \nonumber \\
& & \|A_\v\begin{bmatrix}\w \\ \sigma\end{bmatrix}\|_2\leq r_+\nonumber \\
& & \xtilde_i+\w_i\geq 0,1\leq i\leq n\nonumber \\
& & \|\w\|_2^2\leq C_{\w_{up+}}^2,\label{eq:upperobjlassol11extnon}
\end{eqnarray}
where $A_\v$ is as defined right after (\ref{eq:objlassol13}). If we can show that with overwhelming probability the objective value of the above optimization problem is negative then $r_+$ will be a valid ``high probability" upper-bound on $\zeta_{obj+}$. Moreover, it will be achieved by a $\w$ for which it will hold that  $\|\w\|_2\leq C_{\w_{up+}}$.

First let us rewrite the objective value of the above optimization problem in a slightly more convenient form
\begin{eqnarray}
\min_{\x} & & \sum_{i=1}^{n}\w_i \nonumber \\
& & \|A_\v\begin{bmatrix}\w \\ \sigma\end{bmatrix}\|_2\leq r_+\nonumber \\
& & \xtilde_i+\w_i\geq 0,1\leq i\leq n\nonumber \\
& & \|\w\|_2^2\leq C_{\w_{up+}}^2.\label{eq:upperobjlassol11non}
\end{eqnarray}
Now, we proceed in a fashion similar to the one from Subsection \ref{sec:unsigneddetnon}. We first do a slight modification of the first constraint from (\ref{eq:upperobjlassol11non}) in the following way
\begin{eqnarray}
\min_{\x} & & \sum_{i=1}^{n}\w_i \nonumber \\
& & \|A_\v\begin{bmatrix}\w \\ \sigma\end{bmatrix}\|_2\leq r\nonumber \\
& & \xtilde_i+\w_i\geq 0,1\leq i\leq n\nonumber \\
& & \|\b\|_2^2\leq r_+^2\nonumber \\
& &  \begin{bmatrix} -A \v\end{bmatrix}\begin{bmatrix}\w \\ \sigma\end{bmatrix}=b\nonumber \\
& & \|\w\|_2^2\leq C_{\w_{up+}}^2.\label{eq:upperdefxi4non}
\end{eqnarray}
The Lagrange dual of the above problem then becomes
\begin{multline}
{\cal L}(\lambda^{(2)},\nu^{(1)},\gamma_1,\gamma_2,\w,\b)=\sum_{i=1}^n \w_i
-\sum_{i=n-k+1}^{n}\lambda_i^{(2)}(\xtilde_i+\w_i)\\
-\sum_{i=1}^{n-k}\lambda_i^{(2)}\w_i-\nu^{(1)} A\w +\nu^{(1)}\v\sigma -\nu^{(1)}\b+\gamma_1(\sum_{i=1}^{n}\b_1^2-r_+^2)+\gamma_2(\|\w\|_2^2- C_{\w_{up+}}^2),\label{eq:upperLagran1non}
\end{multline}
where $\nu^{(1)}$ and $\lambda^{(2)}$ are vectors of Lagrange variables as in previous sections.
After rearranging terms we further have
\begin{multline}
{\cal L}(\lambda^{(2)},\nu^{(1)},\gamma_1,\gamma_2,\w,\b)=-\sum_{i=n-k+1}^{n}\lambda_i^{(2)}\xtilde_i\\
+((\z^{(1)}-\lambda^{(2)})^T-\nu^{(1)} A)\w +\nu^{(1)}\v\sigma -\nu^{(1)}\b+\gamma_1(\sum_{i=1}^{n}\b_1^2-r_+^2)+\gamma_2(\|\w\|_2^2- C_{\w_{up+}}^2).\label{eq:upperLagran4non}
\end{multline}
Finally we can write a dual problem to (\ref{eq:upperdefxi4non})
\begin{eqnarray}
\max_{\lambda^{(2)},\nu^{(1)},\gamma_1,\gamma_2}\min_{\w,\b} & & {\cal L}(\lambda^{(2)},\nu^{(1)},\gamma_1,\gamma_2,\w,\b)\nonumber \\
\mbox{subject to} & & \lambda_i^{(2)}\geq 0, 1\leq i\leq n\nonumber \\
& & \gamma_1\geq 0,\nonumber \\
& & \gamma_2\geq 0,\label{eq:upperLagran5non}
\end{eqnarray}
where we of course use the fact that the strict duality obviously holds. Now, after repeating all the steps from (\ref{eq:upperLagran6})
to (\ref{eq:upperLagran14}) (wherever we had $2\lambda^{(2)}$ we would now have $\lambda^{(2)}$ and there will be no upper bound on components of $\lambda^{(2)}$ in the corresponding optimization problems) one obtains and analogue to (\ref{eq:upperLagran14})
\begin{eqnarray}
-\min_{\lambda^{(2)},\nu^{(1)}} \max_{\|\a\|_2=C_{\w_{up+}}}& &
((\z^{(1)}-\lambda^{(2)})^T-\nu^{(1)} A)\a -\nu^{(1)}\v\sigma +\|\nu^{(1)}\|_2r_{+}+\sum_{i=n-k+1}^{n}\lambda_i^{(2)}\xtilde_i\nonumber \\
\mbox{subject to} & &  \lambda_i^{(2)}\geq 0, 1\leq i\leq n.\label{eq:upperLagran14non}
\end{eqnarray}
Now let us define $f_{obj+}^{(up)}$ as
\begin{eqnarray}
-f_{obj+}^{(up)}=-\min_{\lambda^{(2)},\nu^{(1)}} \max_{\|\a\|_2=C_{\w_{up+}}}& &
((\z^{(1)}-\lambda^{(2)})^T-\nu^{(1)} A)\a -\nu^{(1)}\v\sigma +\|\nu^{(1)}\|_2r_{+}+\sum_{i=n-k+1}^{n}\lambda_i^{(2)}\xtilde_i\nonumber \\
\mbox{subject to} & &  \lambda_i^{(2)}\geq 0, 1\leq i\leq n.\label{eq:upperLagran141non}
\end{eqnarray}
Any $r_+$ such that $\lim_{n\rightarrow}P(f_{obj+}^{(up)}\geq 0)=1$ is then a valid ``high-probability" upper bound.
Set
\begin{equation}
\Lambda^{(2+)}=\{\lambda^{(2)}\in R^n | \lambda_i^{(2)}\geq 0,1\leq i\leq n\},\label{eq:upperlambda2non}
\end{equation}
and
\begin{equation}
\xi_{up+}(\sigma,\g,\h,\xtilde,C_{\w_{up}})=\max_{\lambda^{(2)}\in \Lambda^{(2+)},\nu\geq 0}(\sqrt{C_{\w_{up+}}^2+\sigma^2}\|\g\|_2-C_{\w_{up+}}\|\h+\nu\z^{(1)}-\lambda^{(2)})\|_2
-\sum_{i=n-k+1}^{n}\lambda_i^{(2)}\xtilde_i).\label{eq:upperdefxinon}
\end{equation}
After further repeating all the steps between (\ref{eq:upperLagran14}) and Lemma \ref{thm:upperbound} (the only difference is that $\lambda_i^{(2)}\in \Lambda^{(2+)}$ in the ``signed" scenario) one then has the following ``signed" analogue to Lemma \ref{thm:upperbound} (which in essence gives a way of finding an $r_+$ such that $\lim_{n\rightarrow}P(f_{obj+}^{(up)}\geq 0)=1$).
\begin{lemma}
Let $\v$ be an $n\times 1$ vector of i.i.d. zero-mean variance $\sigma^2$ Gaussian random variables and let $A$ be an $m\times n$ matrix of i.i.d. standard normal random variables. Consider an $\xtilde$ defined in (\ref{eq:xtildedef}) and a $\y$ defined in (\ref{eq:systemnoise}) for $\x=\xtilde$. Let then $\zeta_{obj+}$ be as defined in (\ref{eq:objlassol1non}) and let $\w^+$ be the solution of (\ref{eq:objlassol13non}). There is a constant $\epsilon_{upper}>0$
\begin{equation}
P(\zeta_{obj+}\leq \zeta_{obj+}^{(upper)})\geq 1-e^{-\epsilon_{upper}n},\label{eq:lowerboundobjthm1non}
\end{equation}
where
\begin{equation}
\zeta_{obj+}^{(upper)}=(1+\epsilon_{lip})E\xi_{up+}(\sigma,\g,\h,\xtilde,C_{\w_{up+}})+\epsilon_1^{(\h)}\sqrt{n}+\epsilon_3^{(g)}\sqrt{n},\label{eq:upperboundobjthm2non}
\end{equation}
$\xi_{up+}(\sigma,\g,\h,\xtilde,C_{\w_{up+}})$ is as defined in (\ref{eq:upperdefxinon}), $\epsilon_{lip},\epsilon_1^{(\h)},\epsilon_3^{(g)}$ are all positive arbitrarily small constants, and $C_{\w_{up+}}$ is a constant such that $\|\w^+\|_2\leq C_{\w_{up+}}$.
\label{thm:upperbound}
\end{lemma}
\begin{proof}
Follows from the previous discussion.
\end{proof}

\subsection{Matching upper and lower bounds}\label{sec:matchingnon}

In this section we specialize the general bounds introduced above and show how they match. We will again divide presentation in three subsections. In the first of the subsections we will make a connection to the noiseless ``signed" case and show how one can then remove the constraint from (\ref{eq:defhatxinon}), (\ref{eq:defhatwnon}), and (\ref{eq:defhatwnormnon}). In the second subsection we will consider a $\w$ such that $|\|\w\|_2-\|\widehat{\w^+}\|_2|\geq \epsilon_{\w_{up}}\|\widehat{\w^+}\|_2$. We will then quantify how much the lower bound that can be computed for such a $\w$ through the framework presented in Section \ref{sec:unsignedlbzetaobjnon} deviates from the optimal one obtained for $\widehat{\w^+}$. In the last subsection we will then show that there will be a $\w$ such that the upper bound computed through the framework presented in Section \ref{sec:unsignedubzetaobjnon} will deviate less. That will in essence establish that upper and lower bounds computed in the previous sections indeed match.

\subsubsection{Connection to the $\ell_1$ optimization of signed $\x$}\label{sec:connectl1non}

In this subsection we establish a connection between the constraint in (\ref{eq:defhatxinon}), (\ref{eq:defhatwnon}), and (\ref{eq:defhatwnormnon}) and the fundamental performance characterization of $\ell_1$ optimization derived in \cite{StojnicUpper10} (and of course earlier in the context of neighborly polytopes in \cite{DT}). We first recall on the condition from Lemma \ref{thm:optsollowernon}. The condition states
\begin{equation}
\sqrt{1+\frac{\sigma^2}{C_\w^2}}\|\h+\widehat{\nu^+}\z^{(1)}-\widehat{\lambda^{(2+)}}\|_2\leq \|\g\|_2,\label{eq:condoptsollowernon}
\end{equation}
where $C_\w$ is an arbitrarily large constant and $\widehat{\nu^+}$ and $\widehat{\lambda^{(2+)}}$ are the solutions of
\begin{eqnarray}
\max & & \sigma\sqrt{\|\g\|_2^2-\|\h+\nu\z^{(1)}-\lambda^{(2)}\|_2^2} -\sum_{i=n-k+1}^{n}\lambda_i^{(2)}\xtilde_i\nonumber \\
\mbox{subject to} & & \lambda_i^{(2)}\geq 0,1\leq i\leq n\nonumber \\
& & \nu\geq 0.\label{eq:matchoptnon}
\end{eqnarray}
Now we note the following equivalent to (\ref{eq:matchoptnon}) for the case when nonzero components of $\xtilde$ are infinite
\begin{eqnarray}
\max & & \sigma\sqrt{\|\g\|_2^2-\|\h+\nu\z^{(1)}-\lambda^{(2)}\|_2^2} \nonumber \\
\mbox{subject to} & & \lambda_i^{(2)}\geq 0,1\leq i\leq n-k\nonumber \\
 & & \lambda_i^{(2)}=0,n-k+1\leq i\leq n\nonumber \\
& & \nu\geq 0.\label{eq:matchl1non}
\end{eqnarray}
To make the new observations easily comparable to the corresponding ones from \cite{StojnicCSetam09,StojnicEquiv10} we set
\begin{equation}
\bar{\h}^+=[\h_{(1)}^{(1)},\h_{(2)}^{(2)},\dots,\h_{(n-k)}^{(n-k)},\h_{n-k+1},\h_{n-k+2},\dots,\h_n]^T,\label{eq:defhbarnon}
\end{equation}
where $[\h_{(1)}^{(1)},\h_{(2)}^{(2)},\dots,\h_{(n-k)^{(n-k)}}]$ are $[\h_{1},\h_{2},\dots,\h_{n-k}]$ sorted in increasing order (possible ties in the sorting process are of course broken arbitrarily). Also we let $\z^{(2)}$ be as in the previous section, i.e. let it be such that $\z_i^{(2)}=-\z_i^{(1)},n-k+1\leq i\leq n$ and $\z_i^{(2)}=\z_i^{(1)},1\leq i\leq n-k$. It is then relatively easy to see that the above optimization problem is equivalent to
\begin{eqnarray}
\max & & \sigma\sqrt{\|\g\|_2^2-\|\bar{\h}^+-\nu\z^{(2)}+\lambda^{(2)}\|_2^2} \nonumber \\
\mbox{subject to}
& & \lambda_i^{(2)}\geq 0, 1\leq i\leq n-k\nonumber \\
& & \lambda_i^{(2)}=0,n-k+1\leq i\leq n\nonumber \\
& & \nu\geq 0.
\label{eq:matchl11non}
\end{eqnarray}
Let $\nu_{\ell_1+}$ and $\lambda^{(\ell_1+)}$ be the solution of the above maximization. Further, consider the following ``signed" version of the standard $\ell_1$-optimization
\begin{eqnarray}
\mbox{min} & & \|\x\|_{1}\nonumber \\
\mbox{subject to} & & A\x=\y\nonumber \\
& & \x_i\geq 0. \label{eq:l1non}
\end{eqnarray}
Then, as we showed in \cite{StojnicCSetam09} and \cite{StojnicUpper10}, the inequality
\begin{equation}
E\|\g\|_2> E\|\bar{\h}^+-\nu_{\ell_1+}\z^{(2)}+\lambda^{(\ell_1+)}\|_2\label{eq:fundl1expnon}
\end{equation}
establishes the following fundamental performance characterization of the $\ell_1$ optimization algorithm from (\ref{eq:l1non}) that could be used instead of LASSO from (\ref{eq:lassol1non}) to recover signed $\x$ in (\ref{eq:system}) (which is a noiseless version of (\ref{eq:systemnoise}))
\begin{equation}
(1-\beta_w^+)\frac{\sqrt{\frac{1}{2\pi}}e^{-(\erfinv(2\frac{1-\alpha_w^+}{1-\beta_w^+}-1))^2}}{\alpha_w^+}-\sqrt{2}\erfinv (2\frac{1-\alpha_w^+}{1-\beta_w^+}-1)=0,
\label{eq:fundl1non}
\end{equation}
where of course $\alpha_w^+=\frac{m}{n}$ and $\beta_w^+=\frac{k}{n}$. As it is also shown in \cite{StojnicCSetam09} and \cite{StojnicUpper10} both of the quantities under the expected values in (\ref{eq:fundl1expnon}) nicely concentrate. Then with overwhelming probability one has that for any pair $(\alpha,\beta)$ that satisfies (or lies below) the fundamental performance characterization of $\ell_1$ optimization given in (\ref{eq:fundl1non})
\begin{equation}
\|\g\|_2> \|\bar{\h}^+-\nu_{\ell_1+}\z^{(2)}+\lambda^{(\ell_1+)}\|_2.\label{eq:fundl1noexpnon}
\end{equation}
Moreover, since $\lambda_i^{(2)}\geq 0, n-k+1\leq i\leq n$, in (\ref{eq:matchoptnon}) one actually has that (\ref{eq:fundl1noexpnon}) implies that with overwhelming probability
\begin{equation}
\|\g\|_2> \|\h+\widehat{\nu^+}\z^{(1)}-\widehat{\lambda^{(2+)}}\|_2,
\end{equation}
which for sufficiently large $C_\w$ is the same as (\ref{eq:condoptsollowernon}).  We then in what follows assume that pair $(\alpha,\beta)$ is such that it satisfies the ``signed" fundamental $\ell_1$ optimization performance characterization (or is in the region below it) from (\ref{eq:fundl1non}) and therefore proceed by ignoring condition (\ref{eq:condoptsollowernon}).

\subsubsection{Deviation from the lower-bound}\label{sec:devlbnon}

In this subsection we establish that $\|\w_{lasso+}\|_2$ (of course, $\w_{lasso+}=\widehat{\x^+}-\xtilde$, where $\widehat{\x^+}$ is the solution of (\ref{eq:lassol1non})) can not deviate substantially from $\|\widehat{\w^+}\|_2$ without substantially affecting the value of the lower bound on the objective in (\ref{eq:lassol1non}) that is derived in Section \ref{sec:unsignedlbzetaobjnon}. To that end let us assume that there is a $\w_{off+}$ that is the solution of the LASSO from (\ref{eq:lassol1non}) (or to be slightly more precise that is such that $\widehat{\x^+}=\xtilde+\w_{off+}$, where obviously $\widehat{\x^+}$ is the solution of (\ref{eq:lassol1non})). Further, let $|\|\w_{off+}\|_2-\|\widehat{\w^+}\|_2|\geq \epsilon_{\w_{up}}\|\widehat{\w^+}\|_2$, where $\epsilon_{\w_{up}}$ is an arbitrarily small constant.

One can then write a ``signed" analogue to (\ref{eq:matchdefxi})
\begin{equation}
\xi_{off+}(\sigma,\g,\h,\xtilde,\w_{off+})=\max_{\lambda^{(2)}\in \Lambda^{(2+)},\nu\geq 0}(\sqrt{\w_{off}^2+\sigma^2}\|\g\|_2-\w_{off}\|\h+\nu\z^{(1)}-\lambda^{(2)})\|_2
-\sum_{i=n-k+1}^{n}\lambda_i^{(2)}\xtilde_i).\label{eq:matchdefxinon}
\end{equation}
After repeating all the arguments between (\ref{eq:matchdefxi}) and Lemma \ref{thm:matchlowerbound} one obtains the following analogue to Lemma \ref{thm:matchlowerbound}.
\begin{lemma}
Let $\v$ be an $n\times 1$ vector of i.i.d. zero-mean variance $\sigma^2$ Gaussian random variables and let $A$ be an $m\times n$ matrix of i.i.d. standard normal random variables. Consider an $\xtilde$ defined in (\ref{eq:xtildedef}) and a $\y$ defined in (\ref{eq:systemnoise}) for $\x=\xtilde$. Let then $\zeta_{obj+}$ be as defined in (\ref{eq:objlassol1non}) and let $\w_{off+}$ be the solution of (\ref{eq:objlassol13non}). Let $\alpha$ and $\beta$ be below the fundamental characterization (\ref{eq:fundl1non}) and let $\widehat{\w^+}$ be as defined in (\ref{eq:defhatwnon}). Assume that $|\|\w_{off+}\|_2-\|\widehat{\w^+}\|_2|\geq \epsilon_{\w_{up}}\|\widehat{\w^+}\|_2$, where $\epsilon_{\w_{up}}$ is an arbitrarily small but fixed constant. Then there would be a constant $\epsilon_{off}>0$, and arbitrarily small positive constants $\epsilon_{lip},\epsilon_1^{(\h)},\epsilon_1^{(g)}$ such that
\begin{equation}
P(\zeta_{obj+}\geq \zeta_{obj+}^{(off)})\geq 1-e^{-\epsilon_{off}n},\label{eq:matchlowerboundobjthm1non}
\end{equation}
where
\begin{equation}
\zeta_{obj}^{(off)}=
(1-\epsilon_{lip})(1+\frac{\epsilon_{\w_{up}}^2}{2(1+\epsilon_{\w_{up}})})E\xi_{ov+}
(\sigma,\g,\h,\xtilde)-\epsilon_1^{(\h)}\sqrt{n}-\epsilon_1^{(g)}\sqrt{n},\label{eq:matchlowerboundobjthm2non}
\end{equation}
and $\xi_{ov+}(\sigma,\g,\h,\xtilde)$ is as defined in (\ref{eq:Lagran12non}).
\label{thm:matchlowerboundnon}
\end{lemma}
\begin{proof}
Follows from the discussion in Section \ref{sec:devlb}.
\end{proof}

\subsubsection{Deviation of the upper bound}\label{sec:devubnon}

In this section we establish that $\|\w_{lasso+}\|_2$ can not deviate from $\|\widehat{\w^+}\|_2$ as much as it was assumed in the previous section which is conceptually enough to make the bounds from Sections \ref{sec:unsignedlbzetaobjnon} and \ref{sec:unsignedubzetaobjnon} match. All arguments from Section \ref{sec:devub} can be repeated again. The only difference will be that in all optimization problems from Section \ref{sec:devub} one will now have no upper bound on $\lambda_i^{(2)},1\leq i\leq n$ (this essentially amounts to using set $\Lambda^{(2+)}$ instead of set $\Lambda^{(2)}$). One then has a ``signed" analogue to (\ref{eq:devubfinal})
\begin{multline}
\hspace{-.7in}E\xi_{up+}(\sigma,\g,\h,\xtilde,E\|\widehat{\w^+}\|_2)=E\max_{\lambda^{(2)}\in \Lambda^{(2+)},\nu\geq 0}(\sqrt{(E\|\widehat{\w^+}\|_2)^2+\sigma^2}\|\g\|_2-E\|\widehat{\w^+}\|_2\|\h+\nu\z^{(1)}-\lambda^{(2)})\|_2
-\sum_{i=n-k+1}^{n}\lambda_i^{(2)}\xtilde_i)\\
=E\max_{\lambda^{(2)}\in \Lambda^{(2+)},\nu\geq 0}(\sqrt{(E\widehat{d^+})^2+\sigma^2}\|\g\|_2-E\widehat{d^+}\|\h+\nu\z^{(1)}-\lambda^{(2)})\|_2
-\sum_{i=n-k+1}^{n}\lambda_i^{(2)}\xtilde_i)\\
\doteq E\min_{d\geq 0}\max_{\lambda^{(2)}\in \Lambda^{(2+)},\nu\geq 0}(\sqrt{d^2+\sigma^2}\|\g\|_2-d\|\h+\nu\z^{(1)}-\lambda^{(2)})\|_2
-\sum_{i=n-k+1}^{n}\lambda_i^{(2)}\xtilde_i)
= E \xi_{ov+}(\sigma,\g,\h,\xtilde),\label{eq:devubfinalnon}
\end{multline}
where $\widehat{d^+}=\|\widehat{\w^+}\|_2$ would be the solution of a ``signed" analogue to (\ref{eq:devubLagran11}).
Following the arguments after (\ref{eq:devubfinal}) one then has that the assumption of Lemma \ref{thm:lowerboundnon} is unsustainable and that $\|\w_{lasso+}\|_2$ can not deviate substantially from $\|\widehat{\w^+}\|_2$. This then implies that with overwhelming probability the objective value of (\ref{eq:lassol1non}) concentrates around $E \xi_{ov+}(\sigma,\g,\h,\xtilde)$  and consequently that $\|\w_{lasso+}\|_2$ concentrates around $E\|\widehat{\w^+}\|_2$.

\subsection{Connecting all pieces}\label{sec:connectpiecesnon}

In this section we connect all of the above. The following theorem essentially does so.
\begin{theorem}
Let $\v$ be an $n\times 1$ vector of i.i.d. zero-mean variance $\sigma^2$ Gaussian random variables and let $A$ be an $m\times n$ matrix of i.i.d. standard normal random variables. Further, let $\g$ and $\h$ be $m\times 1$ and $n\times 1$ vectors of i.i.d. standard normals, respectively. Consider a $k$-sparse $\xtilde$ defined in (\ref{eq:xtildedef}) and a $\y$ defined in (\ref{eq:systemnoise}) for $\x=\xtilde$. Let the solution of (\ref{eq:lassol1non}) be $\widehat{\x^+}$ and let the so-called error vector of LASSO from (\ref{eq:lassol1non}) be $\w_{lasso+}=\widehat{\x^+}-\xtilde$. Let $n$ be large and let constants $\alpha=\frac{m}{n}$ and $\beta=\frac{k}{n}$ be below the fundamental characterization (\ref{eq:fundl1non}). Consider the following optimization problem:
\begin{eqnarray}
\xi_{ov+}(\sigma,\g,\h,\xtilde)=\max_{\nu,\lambda^{(2)}} & & \sigma\sqrt{\|\g\|_2^2-\|\h+\nu\z^{(1)}-\lambda^{(2)}\|_2^2} -\sum_{i=n-k+1}^{n}\lambda_i^{(2)}\xtilde_i\nonumber \\
\mbox{subject to}
& & \nu\geq 0\nonumber \\
& & \lambda_i^{(2)}\geq 0,1\leq i\leq n.\label{eq:mainlasso1non}
\end{eqnarray}
Let $\widehat{\nu^+}$ and $\widehat{\lambda^{(2+)}}$ be the solution of (\ref{eq:mainlasso1non}). Set
\begin{equation}
\|\widehat{\w^+}\|_2=\sigma\frac{\|\h+\widehat{\nu^+}\z^{(1)}-\widehat{\lambda^{(2+)}}\|_2}{\sqrt{\|\g\|_2^2-\|\h+\widehat{\nu^+}\z^{(1)}-
\widehat{\lambda^{(2+)}}\|_2^2}}.\label{eq:mainlasso2non}
\end{equation}
Then:
\begin{equation}
P((1-\epsilon_1^{(lasso)})E\xi_{ov+}(\sigma,\g,\h,\xtilde)\leq \|\y-A\hat{\x}\|_2
\leq (1+\epsilon_1^{(lasso)})E\xi_{ov+}(\sigma,\g,\h,\xtilde))=1-e^{-\epsilon_2^{(lasso)}n}\label{eq:mainlasso3non}
\end{equation}
and
\begin{equation}
P((1-\epsilon_1^{(lasso)})E\|\widehat{\w^+}\|_2\leq \|\w_{lasso+}\|_2
\leq (1+\epsilon_1^{(lasso)})E\|\widehat{\w^+}\|_2) =1-e^{-\epsilon_2^{(lasso)}n},\label{eq:mainlasso4non}
\end{equation}
where $\epsilon_1^{(lasso)}>0$ is an arbitrarily small constant and $\epsilon_2^{(lasso)}$ is a constant dependent on $\epsilon_1^{(lasso)}$ and $\sigma$ but independent of $n$.
\label{thm:mainlassonon}
\end{theorem}
\begin{proof}
Follows from the above discussion.
\end{proof}

\subsubsection{LASSO's generic performance}\label{sec:genericnon}

In this section we show how the results presented in the above theorem can be adapted to the so-called ``worst-case" scenario or as we refer to it ``generic performance" scenario. Repeating the line of arguments from Section \ref{sec:generic} one can establish the following \emph{generic} equivalent to Theorem \ref{thm:mainlassonon}.
\begin{theorem}
Assume the setup of Theorem \ref{thm:mainlassonon}. Consider the following optimization problem:
\begin{eqnarray}
\xi_{ov+}^{(gob)}(\sigma,\g,\h)=\min_{\nu,\lambda^{(2)}} & & \|\h+\nu\z^{(1)}-\lambda^{(2)}\|_2\nonumber \\
\mbox{subject to}
& & \nu\geq 0\nonumber \\
& & \lambda_i^{(2)}=0,n-k+1\leq i\leq n\nonumber \\
& & \lambda_i^{(2)}\geq 0, 1\leq i\leq n-k.\label{eq:genlasso4non}
\end{eqnarray}
Let $\nu_{gen+}$ and $\lambda^{(gen+)}$ be the solution of (\ref{eq:genlasso4non}). Set
\begin{equation}
\|\w_{gen+}\|_2=\sigma\frac{\|\h+\nu_{gen+}\z^{(1)}-\lambda^{(gen+)}\|_2}{\sqrt{\|\g\|_2^2-\|\h+\nu_{gen+}\z^{(1)}-\lambda^{(gen+)}\|_2^2}}.\label{eq:genlasso5non}
\end{equation}
Then:
\begin{eqnarray}
P(\exists\w_{lasso+}|\|\w_{lasso+}\|_2\in((1-\epsilon_1^{(lasso)})E\|\w_{gen+}\|_2, (1+\epsilon_1^{(lasso)})E\|\w_{gen+}\|_2)) & \geq & 1-e^{-\epsilon_2^{(lasso)}n}\nonumber \\
P(\|\w_{lasso+}\|_2\leq (1+\epsilon_1^{(lasso)})E\|\w_{gen+}\|_2) & \geq & 1-e^{-\epsilon_3^{(lasso)}n},\nonumber \\\label{eq:genlasso6non}
\end{eqnarray}
where $\epsilon_1^{(lasso)}>0$ is an arbitrarily small constant and $\epsilon_2^{(lasso)}$ and $\epsilon_3^{(lasso)}$ are constants dependent on $\epsilon_1^{(lasso)}$ and $\sigma$ but independent of $n$.
\label{thm:genlassonon}
\end{theorem}
\begin{proof}
Follows by the use of the same arguments that were used to establish Theorem \ref{thm:mainlassonon}.
\end{proof}
The following corollary then provides a quick way of computing the concentrating point of the ``worst case" norm of the error vector.
\begin{corollary}
Assume the setup of Theorems \ref{thm:mainlassonon} and \ref{thm:genlassonon}. Let $\alpha=\frac{m}{n}$ and $\beta_w^+=\frac{k}{n}$. Then
\begin{eqnarray}
P(\exists\w_{lasso+}|\|\w_{lasso+}\|_2\in((1-\epsilon_1^{(lasso)})\sigma\sqrt{\frac{\alpha_w^+}{\alpha-\alpha_w^+}}, (1+\epsilon_1^{(lasso)})\sigma\sqrt{\frac{\alpha_w^+}{\alpha-\alpha_w^+}})) & \geq & 1-e^{-\epsilon_2^{(lasso)}n}\nonumber \\
P(\|\w_{lasso+}\|_2\leq (1+\epsilon_1^{(lasso)})\sigma\sqrt{\frac{\alpha_w^+}{\alpha-\alpha_w^+}}) & \geq & 1-e^{-\epsilon_3^{(lasso)}n},\nonumber \\\label{eq:genlasso7non}
\end{eqnarray}
where $\alpha_w^+<\alpha$ is such that
\begin{equation}
(1-\beta_w^+)\frac{\sqrt{\frac{1}{2\pi}}e^{-(\erfinv(2\frac{1-\alpha_w^+}{1-\beta_w^+}-1))^2}}{\alpha_w^+}-\sqrt{2}\erfinv (2\frac{1-\alpha_w^+}{1-\beta_w^+}-1)=0,
\label{eq:genfundl1non}
\end{equation}
$\epsilon_1^{(lasso)}>0$ is an arbitrarily small constant, and $\epsilon_2^{(lasso)}$ and $\epsilon_2^{(lasso)}$ are constants dependent on $\epsilon_1^{(lasso)}$ and $\sigma$ but independent of $n$.
\label{thm:gencomperrornon}
\end{corollary}
\begin{proof}
Follows by the use of the same arguments that were used to establish Corollary \ref{thm:gencomperror} and a recognition that the fundamental characterization of interest in the ``signed" case is the one given in (\ref{eq:fundl1non}).
\end{proof}
 Based on the above corollary one can then for any $\sigma$ and any pair $(\alpha,\beta_w^+)$ (that is below fundamental characterization (\ref{eq:genfundl1non}) or alternatively (\ref{eq:fundl1non})) determine the value of the worst case $E\|\w_{lasso+}\|_2$ as $\sigma\sqrt{\frac{\alpha_w^+}{\alpha-\alpha_w^+}}$. We present the obtained results in Figure \ref{fig:lassoweakthrnon}. For several fixed values of the worst case $E\|\w_{lasso+}\|_2$ we determine curves of points $(\alpha,\beta_w^+)$ for which these fixed values are achieved (of course for any $\alpha$ that is below a curve the value of the corresponding worst case $E\|\w_{lasso+}\|_2$ is smaller). As in the previous section, the lower the norm-2 of the error vector the smaller the allowable region for pairs $(\alpha,\beta_w^+)$. Also as it was the case in the previous section, the results of the above corollary match those obtained in \cite{DonMalMon10,BayMon10lasso} through a state evolution/bilief propagation type of analysis for the ``signed" version of the LASSO from (\ref{eq:biglasso}) (signed version of the LASSO from (\ref{eq:biglasso}) as expected assumes just simple adding of the positivity constraints on the components of $\x$).
\begin{figure}[htb]
\centering
\centerline{\epsfig{figure=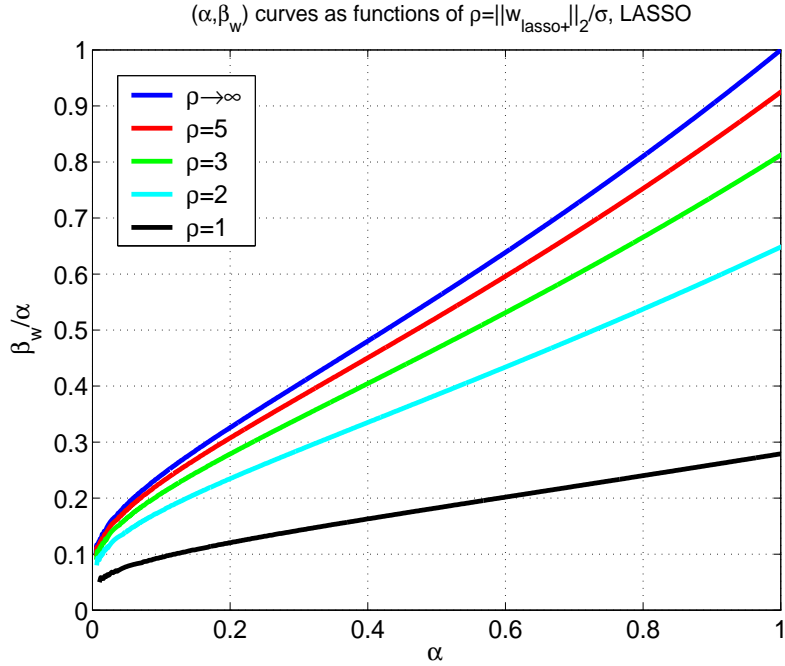,width=10.5cm,height=9cm}}
\vspace{-0in} \caption{$(\alpha,\beta_w^+)$ curves as functions of $\rho=\frac{E\|\w_{lasso+}\|_2}{\sigma}$ for LASSO algorithm from (\ref{eq:lassol1non})}
\label{fig:lassoweakthrnon}
\end{figure}

\section{Connecting LASSO's from (\ref{eq:lassol1}) and (\ref{eq:biglasso})}
\label{sec:connectlasso}

In this section we establish a connection between the LASSO algorithm from (\ref{eq:lassol1}) that we analyzed in Section \ref{sec:unsigned} and the more well known form of LASSO from (\ref{eq:biglasso}). Instead of well-known (\ref{eq:biglasso}) we will consider its a slight modification
\begin{equation}
\min_{\x} \|\y-A\x\|_2+\lambda_{lasso}\|\x\|_1.\label{eq:biglassover}
\end{equation}
Both LASSO's, (\ref{eq:lassol1}) and (\ref{eq:biglassover}), (as well as the one from (\ref{eq:biglasso})) rely on some type of the prior knowledge that can be available about $A$, $\v$, or $\xtilde$. In (\ref{eq:lassol1}) we assumed that one knows $\|\xtilde\|_1$ (of course, if one has no knowledge of $\|\xtilde\|_1$ LASSO from (\ref{eq:lassol1}) simply can not be run). On the other hand the LASSO from (\ref{eq:biglassover}) (as well as the one from (\ref{eq:biglasso})) requires that one sets in advance parameter $\lambda_{lasso}$ which can be a tough task if there is no a priori knowledge about $A$, $\v$, or $\xtilde$. Now even if there is some a priori available knowledge about these objects there are still many ways how one can set $\lambda_{lasso}$. We will below show a particular way of setting $\lambda_{lasso}$ in (\ref{eq:biglassover}) that can make LASSO's from
(\ref{eq:lassol1}) and (\ref{eq:biglassover}) essentially equivalent (of course, as long as one is interested in performance measures discussed in this paper). In the interest of saving space we will sketch only the key arguments without going into tedious details similar to the ones presented in earlier sections. (All that we mention below can be made precise, though; in fact, one can pretty much reach the same level of exactness demonstrated in Sections \ref{sec:unsigned} and \ref{sec:signed}; however, the length of the precise probabilistic arguments would equal (if not exceed) the length of the arguments presented in Sections \ref{sec:unsigned} and \ref{sec:signed}.)

Now, let $\lambda_{lasso}$ in (\ref{eq:biglassover}) be such that $\lambda_{lasso}=E\hat{\nu}$ where $\hat{\nu}$ is the solution of (\ref{eq:Lagran12}). Then (\ref{eq:biglassover}) becomes
\begin{equation}
\min_{\x} \|\y-A\x\|_2+E\hat{\nu}\|\x\|_1,\label{eq:biglassoconn}
\end{equation}
or in a more convenient form
\begin{equation}
\min_{\x} \|\y-A\x\|_2+E\hat{\nu}\|\x\|_1-E\hat{\nu}\|\xtilde\|_1,\label{eq:biglasso1conn}
\end{equation}
This could be rewritten in a way analogous to (\ref{eq:objlassol13}) as
\begin{eqnarray}
\min_{\w} & & \|A_{\v}\begin{bmatrix} \w\\\sigma\end{bmatrix}\|_2+E\hat{\nu}\|\xtilde+\w\|_1-E\hat{\nu}\|\xtilde\|_1,\label{eq:objlassol13conn}
\end{eqnarray}
where $A_{\v}$ is as in (\ref{eq:objlassol13}). One can then repeat all arguments from the beginning of Section \ref{sec:unsignedlbzetaobj} (essentially those before Section \ref{sec:unsigneddet}) to arrive at the following analogue of (\ref{eq:defxi2})
\begin{eqnarray}
\xi_{conn}(\sigma,\g,\h,\xtilde)=\min_{\w} & & \sqrt{\|\w\|_2^2+\sigma^2}\|\g\|_2+\sum_{i=1}^{n}\h_i\w_i+E\hat{\nu}\|\xtilde+\w\|_1-E\hat{\nu}\|\xtilde\|_1 \nonumber \\
\mbox{subject to}
& & \sqrt{\|\w\|_2^2+\sigma^2}\leq \sqrt{C_\w^2+\sigma^2}.\label{eq:defxi2conn}
\end{eqnarray}
Now, one should note that $E\hat{\nu}$ in the above optimization is chosen as the ``optimal" (it is actually the concentrating point of the optimal one; to make this really precise one would need to go through all the probabilistic arguments of Section \ref{sec:unsigned} and plus some more) $\nu$ in the Lagrange dual of (\ref{eq:defxi2}). One then has that arguments from Section \ref{sec:unsignedlbzetaobj} (essentially an appropriate repetition of those that follow (\ref{eq:defxi2})) will produce the lower bound on the objective of (\ref{eq:biglasso1conn}) that is with overwhelming probability arbitrarily close to the one derived in Lemma \ref{thm:lowerbound}. The arguments from Section \ref{sec:unsignedubzetaobj} related to the upper bound can be trivially repeated as well since the negativity of the objective in (\ref{eq:upperobjlassol11}) implies that $r$ is also an upper bound on the optimal value of the objective in (\ref{eq:biglasso1conn}). The matching arguments from Section \ref{sec:matching} then follow as well. Now if one let $\w_{conn}$ be the solution of
(\ref{eq:objlassol13conn}), then with overwhelming probability $\|\w_{conn}\|_2$ concentrates around $E\|\hat{\w}\|_2$, where $\hat{\w}$ is as defined in Theorem \ref{thm:mainlasso}.

For the signed case the arguments are the same, only instead of $E\hat{\nu}$ in (\ref{eq:biglassoconn}), (\ref{eq:biglasso1conn}), and (\ref{eq:objlassol13conn}) one should use $E\widehat{\nu^+}$ where $\widehat{\nu^+}$ is the solution of (\ref{eq:Lagran12non}). Also, as it is probably obvious, this time $\|\w_{conn}\|_2$ concentrates around $E\|\widehat{\w^+}\|_2$ where $\widehat{\w^+}$ is as defined in Theorem \ref{thm:mainlassonon}.

\section{A relation between a LASSO and an SOCP}
\label{sec:socplasso}

In this section we show that there is an SOCP equivalent to the LASSO from (\ref{eq:lassol1}) (as long as the norm-2 of the error vector is a performance measure of interest). To that end let us recall that an SOCP algorithm for finding an approximation of $\xtilde$ if $A$ and $\y$ from (\ref{eq:systemnoise}) are known can be (see, e.g. \cite{CanRomTao06})
\begin{eqnarray}
\min_{\x} & & \|\x\|_1 \nonumber \\
\mbox{subject to} & & \|\y-A\x\|_2\leq r_{socp}.\label{eq:socprel}
\end{eqnarray}
The choice of $r_{socp}$ critically impacts the outcome of the above optimization. In fact more is true, the choice of $r_{socp}$ heavily depends on what type of approximation error one is looking for. As we have mentioned in Section \ref{sec:back} a popular choice for $r_{socp}$ is the smallest quantity that is with high probability larger than $\|\v\|_2$. There are probably many reasons for such a choice. One of them is that it would with high probability guarantee that the original $\xtilde$ in (\ref{eq:systemnoise}) is permissible in (\ref{eq:socprel}). Now, if one is looking for an $\x$ that will be close in norm-2 to the original $\xtilde$ then it is not necessary to look for the original $\xtilde$ (especially so given that finding original $\xtilde$ is in general pretty much impossible). So if one gives up on that then the value of $r_{socp}$ can go even lower than the smallest quantity larger (with overwhelming probability) than $\|\v\|_2$. One should also note that by lowering $r_{socp}$ one would give up not only possibility to find $\xtilde$ (which is tiny anyway) but also highly likely more when it comes to the structure of the solution vector. This is of course a problem on its own that requires a thorough discussion. However, since we now look only at the norm of the error vector as a performance measure we stop short of pursuing this discussion here any further.

To go along these lines we choose
\begin{equation}
r_{socp}=E\xi_{up}(\sigma,\g,\h,E\|\hat{\w}\|_2)\doteq E\xi_{ov}(\sigma,\g,\h,\xtilde)\leq E\xi_{ov}^{(gob)}(\sigma,\g,\h)\leq \sigma \sqrt{m}\doteq E\|\v\|_2.\label{eq:setrsocp}
\end{equation}
Now, let $\x_{socp}$ be the solution of (\ref{eq:socprel}) (with $r_{socp}$ as in (\ref{eq:setrsocp})). Further let $\x_{socp}=\w_{socp}+\xtilde$. Let $\hat{\w}$ be as defined in Theorem \ref{thm:mainlasso}. Then as
shown in Sections \ref{sec:unsignedlbzetaobj}, \ref{sec:unsignedubzetaobj}, and \ref{sec:matching} $\|\w_{socp}\|_2$ concentrates around $E\|\hat{\w}\|_2$. Basically, the argument is that if $|\|\w_{socp}\|_2-\|\hat{\w}\|_2|\geq \epsilon_{\w_{up}} \|\hat{\w}\|_2$ for a fixed arbitrarily small positive $\epsilon_{\w_{up}}$ and $\|\xtilde+\w_{socp}\|_1\leq \|\xtilde\|_1$ then with overwhelming probability $\|\y-A\x_{socp}\|_2> r_{socp}$. On the other hand, as it was also established in Sections \ref{sec:unsignedlbzetaobj}, \ref{sec:unsignedubzetaobj}, and \ref{sec:matching}, there is a $\w$ (which concentrates around $E\|\hat{\w}\|_2$ with overwhelming probability) such that
$\|\xtilde+\w\|_1\leq \|\xtilde\|_1$ and $\|\y-A\x_{socp}\|_2\leq r_{socp}$. This essentially establishes that if one chooses $r_{socp}$ in (\ref{eq:socprel}) as suggested in (\ref{eq:setrsocp}) then the norm-2 of the error vector will be the same as the norm-2 of the error vector one obtains through LASSO's from (\ref{eq:lassol1}) and (\ref{eq:biglassover}) (the latter one of course with an appropriate choice of $\lambda_{lasso}$).

As we hinted above what we presented here is only a characterization of a particular performance measure of an SOCP algorithm (the same is of course true for the LASSO algorithms). How adequate is such a performance measure is whole another story that we will explore in more detail elsewhere.

\section{Numerical results}
\label{sec:numres}

In this section we present a set of numerical results related to the theoretical predictions that we derived in earlier sections. We will divide the presentation into two groups: 1) the set of results that will relate to the general (unsigned) unknown sparse vectors and 2) the set of results that will relate to signed unknown sparse vectors. To make scaling easier in all experiments we set $\sigma=1$. We also assumed that nonzero components of $\xtilde$ are all of equal and large magnitude. For the concreteness we set this magnitude to be $\frac{1000}{\sqrt{n}}$. For every setup that we discuss below we ran $100$ numerical experiments.

\subsection{Numerical results related to general $\x$} \label{sec:simgen}

In this subsection we will present numerical results that relate to the theoretical ones created in Sections \ref{sec:unsigned} and \ref{sec:connectlasso}. We will consider two groups of $(\alpha,\beta_w)$ regimes, one that we will refer to as the low $(\alpha,\beta_w)$ regime and the other that we will refer to as the high $(\alpha,\beta_w)$ regime.

\textbf{\underline{\emph{2) Low $(\alpha,\beta_w)$ regime --- $\rho=\frac{E\|\w_{lasso}\|_2}{\sigma}=2$}}}

\vspace{.05in}

We ran a carefully designed set of experiments intended to show a specific behavior of the LASSO's from (\ref{eq:lassol1}) and (\ref{eq:biglassoconn}) in what we will refer to as the low $(\alpha,\beta_w)$ regime. For $\alpha\in\{0.3,0.5,0.7\}$ we determined three values of $\beta_w$ from the contour LASSO line that corresponds to $\rho=2$ in the figure given in Section \ref{sec:unsigned}. We then ran (\ref{eq:lassol1}) assuming that $\|\xtilde\|_1$ is known and (\ref{eq:biglassoconn}) using theoretical value for $E\hat{\nu}$ where, as mentioned in Section \ref{sec:connectlasso}, $\hat{\nu}$ is the solution of (\ref{eq:Lagran12}). We call the optimal value of the objective in (\ref{eq:biglasso1conn}) $\zeta_{conn}$ (this value is the optimal value of (\ref{eq:biglassoconn}) shifted by a constant). Also, for this set of experiments we set $n=2000$. Obtained results are presented in Table \ref{tab:simlowerspec}. The theoretical values for any of the simulated quantities in any of the simulated scenarios are given in parallel as bolded numbers. We observe a solid agreement between the theoretical predictions and the results obtained through numerical experiments.

\begin{table}
\caption{Experimental/\textbf{theoretical} results for the noisy recovery through LASSO's;  $\sigma=1$, $\rho=E\|\w_{lasso}\|_2=2$; (\ref{eq:lassol1}) and (\ref{eq:biglassoconn}) were run $100$ times with $n=2000$}\vspace{.1in}
\hspace{-0in}\centering
\begin{tabular}{||c|c|c|c|c|c|c||}\hline\hline
$\alpha$ & $\beta_w/\alpha$  &  $E\hat{\nu}$ & \raisebox{.18in}{} $\frac{E\zeta_{conn}}{\sqrt{n}}$  &  $E\|\w_{conn}\|_2$ & $\frac{E\zeta_{obj}}{\sqrt{n}}$ &  $E\|\w_{lasso}\|_2$  \\ \hline\hline
$0.3$ &  $0.21$  & $\bf{1.3141}$  &  $0.2444$/$\bf{0.2449}$  & $2.0225$/$\bf{2}$  &  $0.2449$/$\bf{0.2449}$ & $2.0188$/$\bf{2}$  \\ \hline
$0.5$ &  $0.27$ &  $\bf{1.0227}$  &  $0.3159$/$\bf{0.3162}$  & $2.0058$/$\bf{2}$  &  $0.3162$/$\bf{0.3162}$ & $2.0018$/$\bf{2}$  \\ \hline
$0.7$ &  $0.33$ &  $\bf{0.7959}$  &  $0.3717$/$\bf{0.3742}$  & $2.0168$/$\bf{2}$  &  $0.3721$/$\bf{0.3742}$ & $2.0155$/$\bf{2}$ \\ \hline\hline
\end{tabular}
\label{tab:simlowerspec}
\end{table}

\textbf{\underline{\emph{2) High $(\alpha,\beta_w)$ regime --- $\rho=\frac{E\|\w_{lasso}\|_2}{\sigma}=3$}}}

\vspace{.05in}

We also ran a carefully designed set of experiments intended to show a specific behavior of the LASSO's from (\ref{eq:lassol1}) and (\ref{eq:biglassoconn}) in what we will refer to as the high $(\alpha,\beta_w)$ regime. For $\alpha\in\{0.3,0.5,0.7\}$ we now determined three values of $\beta_w$ from the contour LASSO line that corresponds to $\rho=3$ in the figure given in Section \ref{sec:unsigned}. We then again ran (\ref{eq:lassol1}) assuming that $\|\xtilde\|_1$ is known and (\ref{eq:biglassoconn}) using the theoretical values for $E\hat{\nu}$. For the scenario when $\alpha=0.3$ we set $n=3000$ while for the scenarios with other two values of $\alpha$ we set $n=2000$. Obtained results are presented in Table \ref{tab:simhigherspec}. The theoretical values for any of the simulated quantities in any of the simulated scenarios are again given in parallel as bolded numbers. We again observe a solid agreement between the theoretical predictions and the results obtained through numerical experiments.

\begin{table}
\caption{Experimental/\textbf{theoretical} results for the noisy recovery through LASSO's;  $\sigma=1$, $\rho=E\|\w_{lasso}\|_2=3$; (\ref{eq:lassol1}) and (\ref{eq:biglassoconn}) were run $100$ times}\vspace{.1in}
\hspace{-0in}\centering
\begin{tabular}{||c|c|c|c|c|c|c||}\hline\hline
$\alpha$ & $\beta_w/\alpha$  & $E\hat{\nu}$ &  \raisebox{.18in}{} $\frac{E\zeta_{conn}}{\sqrt{n}}$  &  $E\|\w_{conn}\|_2$ & $\frac{E\zeta_{obj}}{\sqrt{n}}$ &  $E\|\w_{lasso}\|_2$  \\ \hline\hline
$0.3$ &  $0.249$  & $\bf{1.2508}$  &  $0.1699$/$\bf{0.1732}$  & $3.1714$/$\bf{3}$  &  $0.1705$/$\bf{0.1732}$ & $3.1507$/$\bf{3}$  \\ \hline
$0.5$ &  $0.325$ &  $\bf{0.9477}$  &  $0.2231$/$\bf{0.2236}$  & $3.0560$/$\bf{3}$  &  $0.2239$/$\bf{0.2236}$ & $3.0405$/$\bf{3}$  \\ \hline
$0.7$ &  $0.41$ &  $\bf{0.7046}$  &  $0.2579$/$\bf{0.2646}$  & $3.1166$/$\bf{3}$  &  $0.2585$/$\bf{0.2646}$ & $3.1069$/$\bf{3}$ \\ \hline\hline
\end{tabular}
\label{tab:simhigherspec}
\end{table}

\subsection{Numerical results related to signed $\x$} \label{sec:simnon}

In this subsection we will present numerical results that relate to the theoretical ones created in Sections \ref{sec:signed} and \ref{sec:connectlasso}. We will again consider two groups of $(\alpha,\beta_w^+)$ regimes, one that we will refer to as the low $(\alpha,\beta_w^+)$ regime and the other that we will refer to as the high $(\alpha,\beta_w^+)$ regime.

\textbf{\underline{\emph{2) Low $(\alpha,\beta_w^+)$ regime --- $\rho=\frac{E\|\w_{lasso+}\|_2}{\sigma}=2$}}}

\vspace{.05in}

We first ran a set of experiments intended to show a specific behavior of the LASSO's from (\ref{eq:lassol1non}) and (\ref{eq:biglassoconn}) in what we will refer to as the low $(\alpha,\beta_w^+)$ regime. For $\alpha\in\{0.3,0.5,0.7\}$ we determined three values of $\beta_w^+$ from the contour LASSO line that corresponds to $\rho=2$ in the figure given in Section \ref{sec:signed}. We then ran (\ref{eq:lassol1non}) assuming that $\|\xtilde\|_1$ is known and (\ref{eq:biglassoconn}) using theoretical value for $E\widehat{\nu^+}$ where, as mentioned in Section \ref{sec:connectlasso}, $\widehat{\nu^+}$ is the solution of (\ref{eq:Lagran12non}). Also when running (\ref{eq:biglassoconn}) we now added positivity constraints on the elements of $\x$. We call the optimal value of the objective in (\ref{eq:biglasso1conn}) $\zeta_{conn+}$. When $\alpha=0.7$ we set $n=1500$ while for the other two values of $\alpha$ we set $n=2000$. Obtained results are presented in Table \ref{tab:simlowerspecnon}. The theoretical values for any of the simulated quantities in any of the simulated scenarios are as usual given in parallel as bolded numbers. We once again observe a solid agreement between the theoretical predictions and the results obtained through numerical experiments.

\begin{table}
\caption{Experimental/\textbf{theoretical} results for the noisy recovery through LASSO's;  $\sigma=1$, $\rho=E\|\w_{lasso+}\|_2=2$; (\ref{eq:lassol1}) and (\ref{eq:biglassoconn}) were run $100$ times}\vspace{.1in}
\hspace{-0in}\centering
\begin{tabular}{||c|c|c|c|c|c|c||}\hline\hline
$\alpha$ & $\beta_w^+/\alpha$  & $E\widehat{\nu^+}$ &  \raisebox{.18in}{} $\frac{E\zeta_{conn+}}{\sqrt{n}}$  &  $E\|\w_{conn+}\|_2$ & $\frac{E\zeta_{obj+}}{\sqrt{n}}$ &  $E\|\w_{lasso+}\|_2$  \\ \hline\hline
$0.3$ &  $0.286$  & $\bf{0.9592}$  &  $0.2454$/$\bf{0.2449}$  & $1.9939$/$\bf{2}$  &  $0.2461$/$\bf{0.2449}$ & $1.9876$/$\bf{2}$  \\ \hline
$0.5$ &  $0.3842$ &  $\bf{0.6516}$  &  $0.3133$/$\bf{0.3162}$  & $2.0229$/$\bf{2}$  &  $0.3140$/$\bf{0.3162}$ & $2.0177$/$\bf{2}$  \\ \hline
$0.7$ &  $0.4849$ &  $\bf{0.4292}$  &  $0.3786$/$\bf{0.3742}$  & $1.9947$/$\bf{2}$  &  $0.3794$/$\bf{0.3742}$ & $1.9886$/$\bf{2}$ \\ \hline\hline
\end{tabular}
\label{tab:simlowerspecnon}
\end{table}

\textbf{\underline{\emph{2) High $(\alpha,\beta_w^+)$ regime --- $\rho=\frac{E\|\w_{lasso+}\|_2}{\sigma}=3$}}}

\vspace{.05in}

As in the previous subsection we also ran a carefully designed set of experiments intended to show a specific behavior of the LASSO's from (\ref{eq:lassol1non}) and (\ref{eq:biglassoconn}) in what we will refer to as the high $(\alpha,\beta_w^+)$ regime. Following further the methodology of the previous subsection for $\alpha\in\{0.3,0.5,0.7\}$ we determined three values of $\beta_w^+$ from the contour LASSO line that corresponds to $\rho=3$ in the figure given in Section \ref{sec:signed}. We then again ran (\ref{eq:lassol1non}) and (\ref{eq:biglassoconn}) (when running (\ref{eq:biglassoconn}) we of course again added positivity constraints and we again used theoretical value for $E\widehat{\nu^+}$). When $\alpha=0.7$ we set $n=1500$ while for the other two values of $\alpha$ we set $n=2000$. Obtained results are presented in Table \ref{tab:simhigherspecnon}. The theoretical values for all quantities of interest  are again given in parallel as bolded numbers. We once again observe a solid agreement between the theoretical predictions and the results obtained through numerical experiments.

\begin{table}
\caption{Experimental/\textbf{theoretical} results for the noisy recovery through LASSO's;  $\sigma=1$, $\rho=E\|\w_{lasso+}\|_2=3$; (\ref{eq:lassol1}) and (\ref{eq:biglassoconn}) were run $100$ times}\vspace{.1in}
\hspace{-0in}\centering
\begin{tabular}{||c|c|c|c|c|c|c||}\hline\hline
$\alpha$ & $\beta_w^+/\alpha$  & $E\widehat{\nu^+}$ &  \raisebox{.18in}{} $\frac{E\zeta_{conn+}}{\sqrt{n}}$  &  $E\|\w_{conn+}\|_2$ & $\frac{E\zeta_{obj+}}{\sqrt{n}}$ &  $E\|\w_{lasso+}\|_2$  \\ \hline\hline
$0.3$ &  $0.3423$  & $\bf{0.8197}$  &  $0.1713$/$\bf{0.1732}$  & $3.1213$/$\bf{3}$  &  $0.1723$/$\bf{0.1732}$ & $3.0898$/$\bf{3}$  \\ \hline
$0.5$ &  $0.4672$ &  $\bf{0.5757}$  &  $0.2245$/$\bf{0.2236}$  & $2.9983$/$\bf{3}$  &  $0.2255$/$\bf{0.2236}$ & $2.9860$/$\bf{3}$  \\ \hline
$0.7$ &  $0.5971$ &  $\bf{0.3470}$  &  $0.2644$/$\bf{0.2646}$  & $3.0373$/$\bf{3}$  &  $0.2654$/$\bf{0.2646}$ & $3.0218$/$\bf{3}$ \\ \hline\hline
\end{tabular}
\label{tab:simhigherspecnon}
\end{table}

\section{Discussion}
\label{sec:discuss}

In this paper we considered ``noisy" under-determined systems of linear equations with sparse solutions.
We looked from a theoretical point of view at classical polynomial-time LASSO algorithms.
Under the assumption that the system matrix $A$ has i.i.d. standard normal components,
we created a general framework that can be used to characterize various quantities of interest in analyzing the LASSO's performance.
Among other things, the framework enables one to precisely estimate the norm of the error vector in ``noisy" under-determined systems. Moreover, it can do so for any given $k$-sparse vector $\xtilde$.

While many quantities of interest in LASSO recovery can be computed through the mechanism presented here, to demonstrate its power we in this introductory paper focused only on, what we called, LASSO's \emph{generic} performance. We essentially established the precise values of the ``worst-case" norm-2 of the error vector. On the other hand, using the framework one can create a massive set of results for the LASSO's non-generic or as we will refer to it \emph{problem dependent} performance. However, this goes significantly over the scope of an introductory paper.  We will dissect problems from this direction into tiny details in one of the forthcoming papers. Also, the existence of an SOCP type of the recovery algorithm that achieves the same norm-2 of the error vector as the LASSO does followed as a by-product of our analysis.

As for the applications, further developments are pretty much unlimited. Literally every problem that we were able to solve in the so-called noiseless case (and there was hardly any that we were not) through the mechanisms from \cite{StojnicCSetam09} and \cite{StojnicUpper10} can now be handled in the noisy case as well. For example, quantifying performance of LASSO or SOCP optimization problems in solving ``noisy" systems with special structure of the solution vector (block-sparse, binary, box-constrained, low-rank matrix, partially known locations of nonzero components, just to name a few), ``noisy" systems with noisy (or approximately sparse)) solution vectors can then easily be handled to an ultimate precision. In a series of forthcoming papers we will present some of these applications.

\begin{singlespace}
\bibliographystyle{plain}
\bibliography{GenericLasso}
\end{singlespace}

\end{document}